\documentclass[10pt,aps,pra,reprint,superscriptaddress,nofootinbib,floatfix,longbibliography]{revtex4-2}
\usepackage[T1]{fontenc}
\usepackage[utf8]{inputenc}
\usepackage{lmodern}
\usepackage{amsmath,amssymb,amsthm,mathtools,mathrsfs}
\usepackage{braket}
\usepackage{bbm}
\usepackage{microtype}
\usepackage{graphicx,booktabs}
\usepackage{tikz}
\usetikzlibrary{arrows.meta,positioning,decorations.pathreplacing}
\usepackage[colorlinks=true,linkcolor=blue,citecolor=blue,urlcolor=blue]{hyperref}
\DeclareMathOperator{\Tr}{Tr}
\DeclareMathOperator{\meas}{meas}
\newcommand{\Id}{\mathbbm{1}}
\newcommand{\Dt}{D_{\mathrm{tr}}}
\newcommand{\calK}{\mathcal K}
\newcommand{\calL}{\mathcal L}
\newcommand{\calV}{\mathcal V}
\newcommand{\norm}[1]{\left\lVert#1\right\rVert}
\newcommand{\abs}[1]{\left\lvert#1\right\rvert}
\newcommand{\Ewin}{\mathscr E}
\newcommand{\Gwin}{\mathscr G}

\usepackage{bm}
\newtheorem{theorem}{Theorem}
\newtheorem{lemma}[theorem]{Lemma}
\newtheorem{proposition}[theorem]{Proposition}
\newtheorem{corollary}[theorem]{Corollary}
\newtheorem{definition}[theorem]{Definition}
\newtheorem{assumption}[theorem]{Assumption}
\newtheorem{remark}[theorem]{Remark}

\DeclareMathOperator{\Ran}{Ran}

\newcommand{\calB}{\mathcal{B}}
\newcommand{\av}[1]{\overline{#1}}
\hypersetup{
  pdftitle={Finite-time quantum equilibration from local correlations},
  pdfauthor={Piotr Cwiklinski},
  pdfsubject={Main manuscript and Supplementary Information}
}
\begin{document}
\title{Finite-time quantum equilibration from local correlations}
\author{Piotr {\'C}wikli{\'n}ski}
\email{piotr.cwiklinski@ug.edu.pl}
\affiliation{International Centre for Theory of Quantum Technologies, University of Gda{\'n}sk, ul. prof. Marii Janion 7, 80-309 Gda\'nsk, Poland}

\begin{abstract}
An isolated quantum system evolves reversibly, yet a small subsystem
can lose nearly all information about its initial state. Conserved
quantities may retain some of this information, depending on how
strongly they overlap with local observables. Here we consider a
subsystem in an arbitrary initial state coupled to a maximally mixed
bath. We show that the worst-input mean-square equilibration error,
maximised over window origins, is bounded from above and below by
the mean square of a complete local autocorrelation, with the
stationary part subtracted. The two mean-square quantities are linearly comparable. A correlation record over a finite interval
therefore controls the equilibration error in every interval of the
same length, and the result extends to nonnegative observation
weights with nonnegative Fourier transform. The result allows
arbitrary degeneracies of energies and gaps. We apply it to two
static Hamiltonian families. For fixed spin subsystems in the
one-particle sector of a bounded-range periodic clock,
fixed-accuracy thermalisation is guaranteed on windows that grow
only slightly faster than linearly with system size,
despite exponentially small positive energies. A random commuting
model with logarithmic interaction range retains an extensive set
of conserved observables. Its expected mean-square error is bounded
by a finite-size term and an inverse-time term. The inverse-time
dependence is unavoidable because of the initial transient.
These examples show how local relaxation can coexist with exact
conservation laws and a finely spaced spectrum.
\end{abstract}
\maketitle

\section{Introduction}

It is usually assumed that, if left to itself, a piece of matter
will sooner or later reach equilibrium. Everyday experience makes
this seem quite natural: a hot cup of tea cools down, and a spinning
top eventually comes to rest. In both cases, the object exchanges
energy with its surroundings. If we include those surroundings in
our description, however, the underlying dynamics is reversible.
The difficulty is to understand how this reversible dynamics gives
rise to the steady behaviour seen in ordinary measurements.

In classical mechanics, one way to connect these descriptions is
to assume that the dynamics is ergodic on the relevant energy
surface. For almost every initial condition, time averages then
agree with equilibrium ensemble averages.\cite{Birkhoff1931}
But this is a statement about long times. It does not, by itself,
tell us how long a particular observable takes to relax.

For an isolated quantum system with a time-independent Hamiltonian,
the state evolves unitarily. Its populations in the energy basis
remain fixed, while coherences between distinct energies acquire
time-dependent phases. In finite dimension, this motion is
generally quasiperiodic and allows recurrences, so a nonstationary
state does not simply converge to a fixed equilibrium state.
A small subsystem can nevertheless behave quite differently from
the whole. Its observables may remain close to steady values for
most times, even though information about the preparation is still
present in the full state. How, then, does that information become
inaccessible to local measurements?\cite{GogolinEisert2016}

This question goes back at least to von Neumann's quantum ergodic
theorem,\cite{vonNeumann1929} but it is also accessible to experiment.
Ultracold-atom experiments have followed the relaxation of local
observables in nearly isolated systems,\cite{Trotzky2012} and have
observed local thermal behaviour while the complete quantum state
remains measurably pure.\cite{Kaufman2016}
They give a direct reason to ask both whether a subsystem
equilibrates and how long this takes.

General results show that, under suitable conditions, subsystems
and macroscopic observables spend most of their time close to their
temporal averages. The estimates depend on the initial energy
populations and on spectral resonances.\cite{Reimann2008,
LindenPopescuShortWinter2009,Short2011,
GoldsteinLebowitzMastrodonatoTumulkaZanghi2010}
This is equilibration. To conclude that the resulting state is
thermal requires further information about the Hamiltonian and the
preparation.\cite{GogolinEisert2016}
The eigenstate thermalisation hypothesis (ETH),
including its subsystem form, provides one way to obtain that
information.\cite{RigolDunjkoOlshanii2008,
DAlessioKafriPolkovnikovRigol2016,DymarskyLashkariLiu2018}
Nonintegrability alone does not ensure
thermalisation.\cite{Rigol2009,GogolinMuellerEisert2011}

The time needed to equilibrate remains a separate question.
Finite-time bounds can be obtained from the concentration of energy
gaps,\cite{ShortFarrelly2012} while observable-dependent estimates
show why the smallest global gap need not determine the relaxation
time.\cite{GarciaPintosLindenMalabarbaShortWinter2017}
Rapid equilibration has also been proved for typical nonequilibrium
subspaces and observables,\cite{GoldsteinHaraTasaki2013,
GoldsteinHaraTasaki2015,MalabarbaGarciaPintosLindenFarrellyShort2014}
and for random Hamiltonians or sufficiently complex
eigenvectors.\cite{MasanesRoncagliaAcin2013,Cramer2012,
VinayakZnidaric2012,Reimann2016}
Such random models need not have Markovian reduced
dynamics.\cite{ZnidaricPinedaGarciaMata2011}
In our earlier work, a Haar-random eigenbasis led to relaxation
times governed by averaged spectral
information.\cite{BrandaoCwiklinskiHorodecki2012}
The assumptions on preparations, observables, and ensembles matter
when applying these results to a physical system.

For local Hamiltonians, direct proofs are available in particular
lattice models.\cite{CramerDawsonEisertOsborne2008,CramerEisert2010}
Hara and Koike, for example, obtain an optimal equilibration time
proportional to system length for coarse-grained densities in a
class of free-fermion systems.\cite{HaraKoike2026}
Other results treat local disturbances, correlation functions,
typical low-entanglement preparations, and translation-invariant
systems at high temperature.\cite{FarrellyBrandaoCramer2017,
AlhambraRiddellGarciaPintos2020,BertoniWassnerGuarnieriEisert2025,
PilatowskyChoi2026}

Recent work by Vikram is directly related to the question considered
here. Refs.~\cite{Vikram2026,VikramNoAverages2026} show that
finite-time autocorrelation data, together with positive-Fourier time
weights, can constrain thermalisation in later or shifted observation
intervals. This framework already includes a subsystem coupled to a
statistical bath and a cloned construction for thermal equilibrium.
Ref.~\cite{VikramAperiodicity2026} gives a related description in
terms of subspace return probabilities and aperiodicity. Thus
finite-time forecasting from autocorrelations, positive-Fourier
weights, and subspace-return arguments are not new ingredients of the
present work.

Here we ask how much a fixed subsystem can remember when its initial
state is arbitrary. Relaxation of one selected observable is not
sufficient, because the information may remain in another local
observable. We therefore average the autocorrelation over a complete
traceless local operator basis. Its zero-frequency part measures the
overlap with exactly conserved operators, while the remaining part
measures the dynamical return of local information.

We consider the preparation \(\rho_S\otimes\Id_B/m\), where
\(\rho_S\) is arbitrary and the bath is maximally mixed and initially
uncorrelated with the subsystem. We show that the mean square of the
stationary-subtracted autocorrelation on \(0\le t\le T\) is linearly
comparable to the worst-window, worst-input equilibration error. The
input supremum is taken at each time, before the time average. The
upper estimate therefore applies separately to every translated
interval of length \(T\). The lower estimate uses the initial interval,
which is included in the supremum over window origins. The constants
depend only on the subsystem dimension, and no nondegeneracy of
energies or gaps is required.

The new step is to use the complete local operator space. The
Liouvillian spectral projectors compressed to this space are positive,
so the quadratic return of the same space can be compared directly
with the squared complete-local autocorrelation. This gives a linear
two-sided comparison at the same observation time, instead of the
generic Cauchy--Schwarz square root that appears for two different
operator directions. In Liouville space, the return is a particular
subspace-return probability. For the unital reduced dynamics
considered here, it also determines the purity of the normalized Choi
state, a quantity recently studied under the name Choi echo.
\cite{deLeonGonzalezDiazMejia2026} Channel formulations of
equilibration for generalized subsystems have also been developed in a
complementary infinite-time setting through effective-dimension
bounds.\cite{CorreiaVariziCarvalho2026}

The comparison itself is deterministic; spectral estimates enter
only when it is applied to a particular model. It keeps the spectral
weights seen by local observables, so the smallest global spacing
alone does not set the relaxation bound. This is important because
near-resonances of energy differences need not follow the level
statistics of the Hamiltonian.\cite{RiddellPagliaroli2024}
If both the stationary memory and the dynamical return are small, the
subsystem is close, in the stated time-averaged sense, to the maximally
mixed state. Figure~\ref{fig:information} shows the preparation and the
two contributions to the local error.

We illustrate the result with two static Hamiltonian families.
In the one-particle sector of a bounded-range periodic clock, a
fixed spin subsystem thermalises on polynomial observation windows,
even though the spectrum contains exponentially small positive
energies. A randomly dressed commuting-patch model retains an
extensive family of conserved observables and still relaxes
locally. Its interaction range grows logarithmically with system
size, and evolved local operators remain within neighborhoods of
that size. The expected mean-square error is bounded by a finite-size
term and an inverse-window term. These examples show how local
relaxation can coexist with exact conservation laws or a finely
spaced spectrum, and why their effect must be assessed through the
observables available on the subsystem.

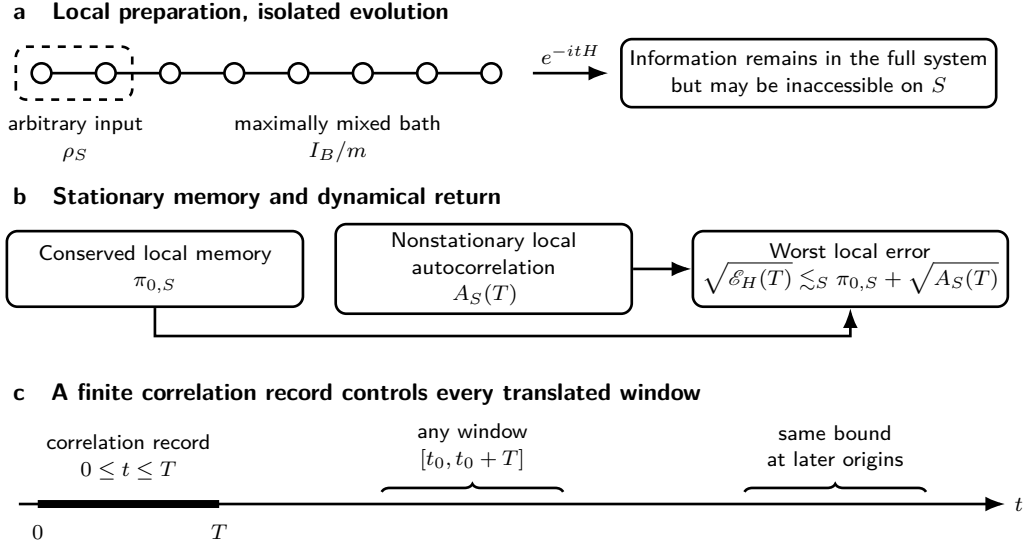
\begin{figure*}[t]
\centering

\begin{tikzpicture}[x=1cm,y=1cm,>=Latex,
  every node/.style={font=\sffamily\fontsize{8}{10}\selectfont},
  line width=1pt]
\node[anchor=west,font=\sffamily\bfseries\small] at (0,4.25)
 {a\quad Local preparation, isolated evolution};
\foreach \x in {0.5,1.35,2.2,3.05,3.9,4.75,5.6,6.45}{
 \draw (\x,3.45) circle (0.13);
}
\foreach \x in {0.5,1.35,2.2,3.05,3.9,4.75,5.6}{
 \draw (\x+0.13,3.45)--(\x+0.72,3.45);
}
\draw[dashed,rounded corners] (0.16,3.1) rectangle (1.7,3.8);
\node[anchor=north,align=center] at (0.95,3.02) {arbitrary input\\$\rho_S$};
\node[anchor=north,align=center] at (4.42,3.02)
 {maximally mixed bath\\$I_B/m$};
\draw[->] (7.0,3.45)--(8.0,3.45) node[midway,above] {$e^{-itH}$};
\node[draw,rounded corners,align=center,minimum width=4.65cm,
minimum height=0.9cm] at (10.65,3.45)
 {Information remains in the full system\\but may be inaccessible on $S$};

\node[anchor=west,font=\sffamily\bfseries\small] at (0,1.8)
 {b\quad Stationary memory and dynamical return};
\node[draw,rounded corners,align=center,text width=3.7cm,
minimum height=1.0cm] (mem) at (2.0,0.86)
 {Conserved local memory\\$\pi_{0,S}$};
\node[draw,rounded corners,align=center,text width=3.7cm,
minimum height=1.0cm] (ret) at (6.35,0.86)
 {Nonstationary local\\autocorrelation\\$A_S(T)$};
\node[draw,rounded corners,align=center,text width=3.95cm,
minimum height=1.0cm] (err) at (11.2,0.86)
 {Worst local error\\$\sqrt{\mathscr E_H(T)}\lesssim_S\pi_{0,S}+\sqrt{A_S(T)}$};
\draw[->] (mem.south)--++(0,-0.37)-|(err.south);
\draw[->] (ret.east)--(err.west);

\node[anchor=west,font=\sffamily\bfseries\small] at (0,-0.80)
 {c\quad A finite correlation record controls every translated window};
\draw[->] (0.2,-2.25)--(13.25,-2.25) node[right] {$t$};
\draw[line width=3pt] (0.45,-2.25)--(2.85,-2.25);
\node[anchor=south,align=center] at (1.65,-2.05)
 {correlation record\\$0\le t\le T$};
\node[anchor=north] at (0.45,-2.40) {$0$};
\node[anchor=north] at (2.85,-2.40) {$T$};
\draw[decorate,decoration={brace,amplitude=3pt}]
 (5.0,-2.09)--(7.4,-2.09);
\node[anchor=south,align=center] at (6.2,-1.93)
 {any window\\$[t_0,t_0+T]$};
\draw[decorate,decoration={brace,amplitude=3pt}]
 (9.8,-2.09)--(12.2,-2.09);
\node[anchor=south,align=center] at (11.0,-1.93)
 {same bound\\at later origins};
\end{tikzpicture}

\caption{
\textbf{Local preparation, conserved memory, and finite observation time.}
\textbf{a}, An arbitrary input on a fixed subsystem evolves under the
full Hamiltonian, with the remaining degrees of freedom initially
maximally mixed.  The global evolution remains reversible.
\textbf{b}, The complete local autocorrelation separates exact
stationary memory, \(\pi_{0,S}\), from nonzero-frequency return,
measured by \(A_S(T)\).  For fixed subsystem dimension, the latter is
linearly comparable to the worst-window equilibration error, while
stationary and dynamical contributions add at the level of
root-mean-square distinguishability.
\textbf{c}, Correlation information from \(0\le t\le T\), together
with the stationary component when it is subtracted, bounds the error
in every translated window of length \(T\).  The input supremum is
taken at each time, so the good times apply simultaneously to all
subsystem inputs.  The figure is schematic.
}
\label{fig:information}
\end{figure*}

\section{Results}

\subsection{Local preparation and equilibration}

Let \(\mathcal H=\mathcal H_S\otimes\mathcal H_B\), with
\(d=\dim\mathcal H_S\ge2\), \(m=\dim\mathcal H_B\), and \(D=dm\).
We allow arbitrary interactions between the subsystem and the bath.
Starting from a product state, the reduced dynamics is
\begin{equation}
 \rho_0=\rho_S\otimes\frac{\Id_B}{m},\qquad
 \Phi_t(\rho_S)=
 \Tr_B\!\left[e^{-itH}\rho_0 e^{itH}\right].
 \label{eq:prep}
\end{equation}
We set \(\hbar=1\). The global maximally mixed state is stationary,
so \(\tau_S=\Id_S/d\) is a fixed point of every \(\Phi_t\).
The subsystem dimension is held fixed as the bath grows.

To quantify the information remaining on \(S\), we take the largest
trace distance between \(\Phi_t(\rho_S)\) and \(\tau_S\) over all
subsystem inputs. A small value means that no initial state leaves
a large signal distinguishable by a measurement on \(S\).
Finite systems can recur, so we average the squared distance over
a window and take the supremum over its starting time:
\begin{equation}
 \Ewin_H(T)=
 \sup_{t_0\ge0}\frac1T\int_{t_0}^{t_0+T}
 \sup_{\rho_S}\Dt\!\left(\Phi_t(\rho_S),\tau_S\right)^2dt.
 \label{eq:E}
\end{equation}
The input supremum is inside the integral. Thus, in every window,
the fraction of times at which \emph{any} input has error at least
\(\eta\) is at most \(\Ewin_H(T)/\eta^2\). The remaining times
work simultaneously for all local inputs.

For equilibration, the reference state is the temporal state
\(\overline\Phi(\rho_S)\), which may depend on the input.
We write \(\Gwin_H(T)\) for Eq.~\eqref{eq:E} with this target.
Infinite-temperature local thermalisation additionally requires
\(\overline\Phi(\rho_S)\) to be close to \(\tau_S\).
All these statements concern the maximally mixed product-bath
preparation and measurements on \(S\) alone.

\subsection{Local autocorrelations and finite observation times}

Choose a Hermitian basis \(P_1,\ldots,P_q\) of traceless operators
on \(S\), with \(\Tr(P_aP_b)=d\delta_{ab}\) and \(q=d^2-1\).
For qubits, these can be the nonidentity Pauli strings. Set
\(O_a=P_a\otimes\Id_B\) and define
\begin{equation}
 \varphi_S(t)=\frac1q\sum_{a=1}^{q}
       \frac1D\Tr\!\left[O_a e^{-itH}O_a e^{itH}\right].
 \label{eq:phi}
\end{equation}
This is a single scalar quantity built from the complete local
operator space and is independent of the chosen orthonormal basis. For
a fixed subsystem, the number of correlators does not grow with the
bath. Physically, the complete basis matters because every difference
between two possible initial states of \(S\) lies in this traceless
operator space. If an evolved local operator has a component that
returns to the same space, some measurement on \(S\) can still retain
information about the preparation.

Let \(\pi_{0,S}\) be the infinite-time average of
\(\varphi_S(t)\).  It is the average weight of local operators in the
conserved operator subspace and is nonnegative.  Subtracting this
stationary part gives
\begin{equation}
 A_S(T)=\frac1{2T}\int_{-T}^{T}
       |\varphi_S(t)-\pi_{0,S}|^2dt.
 \label{eq:A}
\end{equation}
We also use the full mean square
\begin{equation}
 \mathcal A_S(T)=\frac1{2T}\int_{-T}^{T}
       |\varphi_S(t)|^2dt.
 \label{eq:A-full-main}
\end{equation}
The average over operators is taken before the modulus is squared.
Both integrands are even, so these quantities are also their averages
on \([0,T]\).  No monotonic decay of individual correlators is assumed.

With \(q=d^2-1\), our main rectangular-window result is
\begin{equation}
 \frac{d-1}{2d}A_S(T)
 \le \Gwin_H(T)
 \le (d-1)q^2 A_S(T).
 \label{eq:central}
\end{equation}
The same comparison holds for thermalisation to the maximally mixed
state when the stationary contribution is retained:
\begin{equation}
 \frac{d-1}{2d}\mathcal A_S(T)
 \le \Ewin_H(T)
 \le (d-1)q^2\mathcal A_S(T).
 \label{eq:central-full}
\end{equation}
For separating exact conserved memory from dynamical return we obtain
the sharper form
\begin{equation}
 \Ewin_H(T)
 \le
 (d-1)q^2
 \left[
 \sqrt{A_S(T)}+\frac{\pi_{0,S}}{2}
 \right]^2.
 \label{eq:central-separated}
\end{equation}
It follows by applying the trace-distance triangle inequality at the
level of root-mean-square errors. Thus conserved local memory and
dynamical return combine as distinguishability amplitudes, rather than
through a separate factor-two estimate for their squares. The two-sided
comparisons in Eqs.~\eqref{eq:central} and \eqref{eq:central-full}
concern the worst-window functionals. Their upper bounds hold
separately in every translated window of length \(T\). Their lower
bounds use the window beginning at zero, which is included in the
supremum over window origins. All constants depend only on the fixed
subsystem dimension, and the inequalities hold for every
finite-dimensional Hamiltonian and every \(T>0\).

The rectangular-window statement is one form of a more general
result. Let \(w(t)\) describe how observations are weighted in time and
assume
\[
\begin{aligned}
 w(t)&\ge0,\qquad
 \int_{\mathbb R}w(t)\,dt=1,\\
 \widehat w(\omega)&=\int_{\mathbb R}w(t)e^{it\omega}dt\in[0,\infty).
\end{aligned}
\]
For the stationary-subtracted correlation define
\[
 \mathcal C_w=
 \int_{\mathbb R}w(t)|\varphi_S(t)-\pi_{0,S}|^2dt,
\]
and let
\[
 \mathcal G_w=
 \sup_{c\in\mathbb R}\int_{\mathbb R}w(t-c)
 \sup_{\rho_S}
 \Dt\!\left(\Phi_t(\rho_S),\overline\Phi(\rho_S)\right)^2dt.
\]
Then
\begin{equation}
 \frac{d-1}{2d}\mathcal C_w
 \le \mathcal G_w
 \le \frac{(d-1)q^2}{4}\mathcal C_w.
 \label{eq:general-positive-kernel-main}
\end{equation}
The same statement holds for the full correlation and the distance
from \(\Id_S/d\). Taking a triangular weight gives the
rectangular-window upper estimate above. The lower estimate follows
instead from the pointwise lower comparison on the initial window.

The complete local average is stronger than a one-observable
autocorrelation. In Ref.~\cite{VikramNoAverages2026}, the
autocorrelation of one operator controls a different correlator. The
two operator directions need not coincide, so the proof uses
Cauchy--Schwarz and acquires a square root. Here both quantities concern
the return of the \emph{same complete local operator space}. If \(P_S\)
projects onto this space and \(Q_\omega\) is a Liouvillian spectral
projector, then
\[
 A_\omega=P_SQ_\omega P_S\ge0,
\]
and
\[
 0\le\Tr(A_\omega A_\nu)
 \le\Tr(A_\omega)\Tr(A_\nu).
\]
This coefficientwise positivity gives the direct comparison of the two
quadratic quantities. In the reverse direction, averaging over pure
subsystem inputs gives the Frobenius norm of the local return matrix.
Methods and Supplementary Note~S2 contain both calculations.

The trace-distance result concerns measurements on the subsystem
itself. We also obtain a diamond-norm corollary, which allows the
subsystem to be initially entangled with an arbitrary external
reference. Let
\[
 \mathcal D(X)=\Tr(X)\frac{\Id_S}{d},
 \qquad
 \Delta_\diamond(\Phi,\Psi)
 =\frac12\|\Phi-\Psi\|_\diamond,
\]
and define
\begin{align}
 \Gwin_H^\diamond(T)
 &=\sup_{t_0\ge0}\frac1T\int_{t_0}^{t_0+T}
   \Delta_\diamond(\Phi_t,\overline\Phi)^2dt,\\
 \Ewin_H^\diamond(T)
 &=\sup_{t_0\ge0}\frac1T\int_{t_0}^{t_0+T}
   \Delta_\diamond(\Phi_t,\mathcal D)^2dt.
\end{align}
Then
\begin{equation}
 \begin{aligned}
 \frac{d-1}{2d}A_S(T)
 &\le \Gwin_H^\diamond(T)
 \le d^2q^2 A_S(T),\\
 \frac{d-1}{2d}\mathcal A_S(T)
 &\le \Ewin_H^\diamond(T)
 \le d^2q^2 \mathcal A_S(T).
 \end{aligned}
 \label{eq:diamond-main}
\end{equation}
This is an additional, stronger operational corollary, not a
replacement for the trace-distance statement. The diamond distance
allows an arbitrary reference system, including one initially
entangled with \(S\).\cite{ChenDingHuang2022,Watrous2018} The norm
comparison used below is the upper inequality in Eq.~(2.2) of
Ref.~\cite{ChenDingHuang2022}. The larger dimension-dependent constant
is the cost of this stronger test. For the target \(\mathcal D\), small
\(\Ewin_H^\diamond(T)\) means that the subsystem is approximately
decoupled from the reference. Methods and Supplementary Note~S2 give
the proof.

When only the stationary and fluctuating pieces are estimated
separately, Eq.~\eqref{eq:central-separated} is useful.  For random
Hamiltonians one first squares the correlation for each realization and
then averages over disorder; averaging the correlation before squaring
would discard the fluctuations needed in the bound.

\subsection{Conserved local memory}

The positive spectral measure of the complete local autocorrelation
can be written as
\(\mu_S=\pi_{0,S}\delta_0+\widetilde\mu_S\).
Its weights measure the overlap of each Bohr-frequency sector with
the local operator space. They are therefore sensitive to which
transitions can be seen by the subsystem. The zero-frequency part
contains the information retained by exactly conserved operators.

We show that the largest stationary trace-distance error is bounded
by a constant depending on \(d\) times \(\pi_{0,S}\). Its
contribution to the mean-square thermalisation bound is consequently
quadratic. Conservation laws are thus compatible with local
thermalisation when their total overlap with the local operator
space is small. The projection onto conserved quantities is closely
related to the Mazur bound and Suzuki's description of long-time
correlations.\cite{Mazur1969,Suzuki1971}

Energy conservation gives a simple example. For the bounded-degree
ensemble of independently Haar-rotated local terms specified in
Methods, and a fixed subsystem containing a nontrivial interaction
support, we obtain \(\mathbb E\pi_{0,S}\ge c_S/N\).
A matching upper bound requires an independent argument. If it gives
\(\pi_{0,S}=O(N^{-1})\), the
stationary contribution to Eq.~\eqref{eq:central} is
\(O_S(N^{-2})\). The time-dependent contribution must still be
estimated from \(A_S(T)\).

For nondegenerate energies, weak ETH provides one route to a small
stationary weight; Methods gives the corresponding
infinite-temperature estimate. Weak-ETH results for
translation-invariant systems are relevant here,\cite{Mori2016,
PilatowskyChoi2026} while energy degeneracies require control of the
whole conserved block rather than diagonal matrix elements in one
chosen eigenbasis.\cite{DabelowEideckerDunkelReimann2026}
Subsystem ETH and return-to-equilibrium results address related
questions about the equilibrium state.\cite{DymarskyLashkariLiu2018,
FarrellyBrandaoCramer2017} They do not, by themselves, supply the
finite-time correlation estimate needed in Eq.~\eqref{eq:central}.

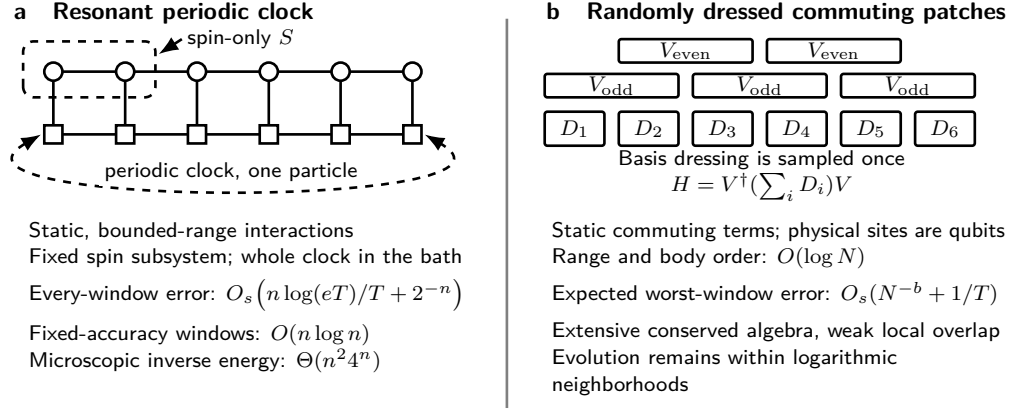
\begin{figure*}[t]
\centering

\begin{tikzpicture}[x=1cm,y=1cm,>=Latex,
 every node/.style={font=\sffamily\fontsize{8}{10}\selectfont},
 line width=1pt]
\node[anchor=west,font=\sffamily\bfseries\small] at (0,4.2)
 {a\quad Resonant periodic clock};
\foreach \x in {0.65,1.6,2.55,3.5,4.45,5.4}{
 \draw (\x,3.42) circle (0.12);
 \draw (\x-0.12,2.47) rectangle (\x+0.12,2.71);
 \draw (\x,2.71)--(\x,3.3);
}
\foreach \x in {0.65,1.6,2.55,3.5,4.45}{
 \draw (\x+0.12,3.42)--(\x+0.83,3.42);
 \draw (\x+0.12,2.59)--(\x+0.83,2.59);
}
\draw[rounded corners,dashed] (0.25,3.08) rectangle (2.0,3.82);
\node[anchor=west] at (2.3,3.85) {spin-only $S$};
\draw[->] (2.35,3.85)--(2.04,3.65);
\draw[<->,dashed] (0.49,2.59) to[out=210,in=330] (5.56,2.59);
\node[fill=white,inner sep=2pt] at (3.0,2.1) {periodic clock, one particle};
\node[align=flush left,text width=6.0cm,anchor=north west] at (0.2,1.56)
 {Static, bounded-range interactions\\
 Fixed spin subsystem; whole clock in the bath\\[4pt]
 Every-window error:
 $O_s\!\left(n\log(eT)/T+2^{-n}\right)$\\[4pt]
 Fixed-accuracy windows: $O(n\log n)$\\
 Microscopic inverse energy: $\Theta(n^2 4^n)$};

\draw[gray] (6.65,-1.03)--(6.65,4.12);
\node[anchor=west,font=\sffamily\bfseries\small] at (7.05,4.2)
 {b\quad Randomly dressed commuting patches};
\foreach \x/\i in {7.15/1,8.13/2,9.11/3,10.09/4,11.07/5,12.05/6}{
 \draw[rounded corners=1pt] (\x,2.44) rectangle ++(0.79,0.45);
 \node at (\x+0.395,2.665) {$D_{\i}$};
}
\foreach \x in {7.15,9.11,11.07}{
 \draw[rounded corners=1pt] (\x,3.08) rectangle ++(1.77,0.31);
 \node at (\x+0.885,3.235) {$V_{\rm odd}$};
}
\foreach \x in {8.13,10.09}{
 \draw[rounded corners=1pt] (\x,3.54) rectangle ++(1.77,0.31);
 \node at (\x+0.885,3.695) {$V_{\rm even}$};
}
\node[align=center] at (10.02,2.05)
 {Basis dressing is sampled once\\$H=V^\dagger(\sum_iD_i)V$};
\node[align=flush left,text width=6.0cm,anchor=north west] at (7.12,1.56)
 {Static commuting terms; physical sites are qubits\\
 Range and body order: $O(\log N)$\\[4pt]
 Expected worst-window error:
 $O_s(N^{-b}+1/T)$\\[4pt]
 Extensive conserved algebra, weak local overlap\\
 Evolution remains within logarithmic neighborhoods};
\end{tikzpicture}

\caption{
\textbf{Local relaxation in two static Hamiltonian families.}
\textbf{a}, The periodic spin--clock Hamiltonian has bounded-range
interactions.  The theorem concerns spin-only subsystems within its
one-clock-particle sector; the bath contains the full one-particle
clock.  The sufficient fixed-accuracy observation scale is
\(O(n\log n)\), whereas the inverse smallest positive sector energy is
\(\Theta(n^2 4^n)\).  An \(O(n^{-1})\) mean-square error is guaranteed on
\(O(n^2\log n)\) windows.
\textbf{b}, Random basis dressing of independent patches produces
commuting Hamiltonian terms with logarithmically growing range and body
order.  The dressing circuit specifies a static eigenbasis, not a
time-dependent drive.  The model retains an extensive conserved
algebra, yet its expected worst-window local error is
\(O_S(N^{-b}+T^{-1})\), where \(b>0\) is fixed by the patch-size
choice.  Local operators remain within a logarithmic-size neighborhood.
Both panels are schematic illustrations of the analytic models.
}
\label{fig:models}
\end{figure*}

\subsection{Local relaxation in a periodic clock}

The periodic Feynman--Kitaev construction of Ippoliti and Long
embeds a circuit in a static, geometrically local spin--clock
Hamiltonian.\cite{IppolitiLong2026} It builds on the
Hamiltonian-computation constructions of Feynman and
Kitaev.\cite{Feynman1985,KitaevShenVyalyi2002}
We use the maximal linear-feedback shift-register (LFSR) circuit
on \(n\ge2\) spin qubits. The dynamics is restricted to the
invariant sector with one clock particle, and \(S\) contains a
fixed number of spin qubits. The bath consists of the remaining
spins and the whole one-particle clock space, each maximally mixed.
This gives the tensor-product preparation of Eq.~\eqref{eq:prep}
without splitting the fixed-particle clock sector into spatial parts.

The circuit eigenphases lie on a regular grid. The clock maps them
to a cosine energy spectrum in a sector of dimension \(n2^n\).
Including the two exceptional \(+1\) Floquet eigenvectors, the
one-clock-particle sector contains the complete
\(n(2^n-1)\)-point cosine grid.  Its smallest positive energy is
therefore exactly
\begin{equation}
 \Delta_n
 =
 1-\cos\left[
 \frac{2\pi}{n(2^n-1)}
 \right],
 \qquad
 \Delta_n^{-1}=\Theta(n^2 4^n).
 \label{eq:clock-exact-small-energy-main}
\end{equation}
The hopping strengths are fixed and the sector bandwidth stays
bounded, so this fine energy scale is not introduced by rescaling
the Hamiltonian.

Nevertheless, for every fixed spin subsystem,
\begin{equation}
 \Ewin_{H_{\rm FK}}(T)
 \le C_s\!\left[\frac{n\log(eT)}{T}+2^{-n}\right],
 \qquad T\ge1.
 \label{eq:clock}
\end{equation}
The largest stationary local trace-distance error is
\(O_s(2^{-n})\). The conserved energy direction has exactly zero
overlap with spin-local traceless operators because every hopping
term has zero clock trace. Other conserved directions remain, and
their total local weight is bounded separately.

Equation~\eqref{eq:clock} gives two observation scales. Windows of
length \(n^2\log(en)\) give mean-square error \(O_s(n^{-1})\).
If the trace-distance tolerance and the allowed fraction of bad
times are fixed, windows of order \(n\log n\) suffice for all
large enough \(n\). Both are much shorter than
\(\Delta_n^{-1}\). The distinction between fixed accuracy and
an error that vanishes with system size is useful when interpreting
these time bounds.

The proof combines the LFSR Pauli matrix-element estimate of
Ref.~\cite{IppolitiLong2026} with a count of differences between
cosine energies. The matrix-element estimate limits the weight of
each regular transition. Counting the transitions in an interval
of width \(\delta\) then gives local spectral weight
\(O(n\delta\log(e/\delta))+O(2^{-n})\).
Taking \(\delta=1/T\) yields Eq.~\eqref{eq:clock}.
The logarithm comes from the unweighted cosine-difference density;
this argument does not establish a matching lower bound on the
physical relaxation time. Supplementary Notes~S4--S5 give the
proof, a structural clock theorem, and a prime-power extension.

The example shows that an exactly solvable, resonant local
Hamiltonian can lose the information stored in a fixed spin
subsystem on polynomial observation windows. Its fine spectrum
alone does not obstruct local thermalisation.

This application is distinct from the Floquet-circuit examples used
in Ref.~\cite{Vikram2026}: the circuit here is embedded in a static
periodic Feynman--Kitaev Hamiltonian, and the relaxation bound is
proved from the spectrum and local matrix elements of that static
Hamiltonian.

\subsection{Local relaxation with many conserved observables}

For the second example, divide a chain of \(N=M\xi\) qubits into
an even number \(M\) of equal patches. Following Cui and
collaborators,\cite{CuiSchusterMaoHuangBrandao2025} take
\begin{equation}
H_N=V^\dagger\left(\sum_{i=1}^{M}D_i(J)\right)V,
 \qquad
 D_i(J)=\sum_z J_{i,z}|z\rangle\langle z|_i.
 \label{eq:patch-H}
\end{equation}
Here \(V\) is a two-layer circuit of independent Haar unitaries
on neighboring patch pairs. Initially we take all energies
\(J_{i,z}\) to be independent and uniform on \([-1,1]\).
The circuit and energies are sampled once; they define a fixed
Hamiltonian for the subsequent evolution.

The terms \(V^\dagger D_iV\) commute and have norm at most one.
Each is supported on a bounded number of adjacent patches.
There are also \(N\) independent conserved observables
\(V^\dagger Z_jV\), obtained by dressing the single-qubit Pauli
operators \(Z_j\). With \(\xi=O(\log N)\), the physical
sites remain qubits, while the interaction range and the number of
qubits in a term both grow logarithmically.

For every fixed \(b>0\), an admissible logarithmic patch size
can be chosen so that
\begin{equation}
 \mathbb E\,\Ewin_{H_N}(T)
 \le C_s\left[N^{-b}+\frac1T\right],
 \qquad T>0.
 \label{eq:ensemble}
\end{equation}
Methods specifies the allowed sequence of system sizes and patch
lengths. The expectation is taken after the suprema over inputs
and window origins. Markov's inequality therefore gives an event
on which the local error is controlled for all inputs and all
starting times. For example, with \(b=2\) and \(T=N^2\),
\(\Ewin_{H_N}(N^2)\le N^{-1}\) except with probability
\(O_s(N^{-1})\). On the same event,
\(\Ewin_{H_N}(L)\le2N^{-1}\) for every \(L\ge N^2\).

The proof uses two properties of the ensemble. At the reference
time \(t=\pi\), its evolution forms an approximate unitary
design in the sense established by Theorem~3 of
Ref.~\cite{CuiSchusterMaoHuangBrandao2025}. This bounds a
positive sum of squared spectral weights. At other times,
averaging over the independent continuous energies bounds the
dephasing remainder by \(t^{-2}\). The spectral weights add to
one, so the remainder carries no factor growing with the patch
dimension or the number of patches. Integrating it and applying
the local-return theorem gives Eq.~\eqref{eq:ensemble}.
The design property is used at one reference time and for the
ensemble; it is not a statement about design formation along
the trajectory of a single Hamiltonian.

The inverse-window contribution is sharp. A pure local input
starts a finite distance from \(\tau_S\), and only a bounded
number of bounded-norm terms meet the subsystem. The initial
signal therefore persists for a nonzero time. Its contribution
to the window beginning at zero gives a lower bound \(c_s/T\)
for every realization. Together with Eq.~\eqref{eq:ensemble},
this yields
\(\mathbb E\Ewin_{H_N}(T)=\Theta_s(T^{-1})\) for
\(1\le T\le N^b\). This is a statement about the integrated
initial transient, not a pointwise late-time decay law.

The argument also covers independent, nonidentically distributed
energy densities on \([-1,1]\), provided they have a common
Fourier-decay bound. A norm-small perturbation \(W\), including
one that breaks commutativity, contributes
\(O_s(\|W\|_\infty^2T^2)\) to the estimate. This controls
perturbations at the chosen observation scale; it does not cover
an extensive sum of fixed-strength additional interactions.

The physical mechanism is local dephasing within growing patch
neighborhoods. The conserved observables are spread over those
neighborhoods, so their overlap with a fixed subsystem can be
small. Since the Hamiltonian terms commute, an evolved local
operator remains supported on \(O(\log N)\) qubits for all
times. Local relaxation therefore occurs without transport
through the whole chain. Figure~\ref{fig:models} compares this
model with the bounded-range clock.

\section{Discussion}

For an arbitrary subsystem state coupled to a maximally mixed bath,
we have related local equilibration to the mean square of a complete
local autocorrelation. The worst-window functional is bounded from
both sides at the same observation time. Its upper estimate holds in
every translated window, while the lower estimate follows from the
initial window included in the supremum. Subtracting the stationary
part separates memory retained by conserved operators from the
dynamical return of the remaining local information.

The periodic clock and the commuting model realise these two effects
in different ways. In the clock, a fixed spin subsystem thermalises on
\(O(n\log n)\) fixed-accuracy windows despite exponentially small
positive energies. In the commuting model, the conserved operators
are spread over logarithmic neighborhoods and have a small expected
overlap with a fixed subsystem. The inverse-window contribution is
sharp because the worst-window criterion includes the initial
transient. These statements concern the two model families and the
preparations used here; they are not a general relaxation theorem for
local spin chains.

The general finite-time mechanism is closely related to the work of
Vikram. Refs.~\cite{Vikram2026,VikramNoAverages2026} already use
finite-time autocorrelations and positive-Fourier weights to constrain
thermalisation in later or shifted intervals, including the
subsystem--statistical-bath setting. Ref.~\cite{VikramAperiodicity2026}
gives the related formulation in terms of return probabilities and
aperiodicity. We therefore do not claim the forecasting mechanism,
positive-Fourier weights, subspace return, or converse reasoning as
new.

Our result uses this mechanism for the complete traceless local
operator algebra. A small autocorrelation of one observable can miss
memory stored in another local direction, while the complete local
return contains every direction in which two subsystem preparations
can differ. Positivity of the compressed Liouvillian spectral blocks
then gives the linear same-time comparison with the worst-input
reduced-state error. In Liouville space this is a particular
subspace-return probability; for the present unital channel it is also
a shifted and rescaled Choi-state purity.\cite{deLeonGonzalezDiazMejia2026}
Related channel-based formulations of equilibration for generalized
subsystems derive effective-dimension bounds in a complementary
infinite-time setting.\cite{CorreiaVariziCarvalho2026}

The local-return criterion can be compared with eigenoperator
thermalisation and studies of eigenstate correlations at fixed
frequencies.\cite{Buca2023,HahnLuitzChalker2024}
Full-ETH and nonlinear-hydrodynamic descriptions address dynamical
correlations beyond a single-eigenstate picture,\cite{PappalardiFritzschProsen2025,
WangMishraYangDelacretazPappalardi2026} while Lanczos-coefficient
methods estimate correlation equilibration
times.\cite{WangFullgrafGemmer2025}
Within this finite-time autocorrelation setting, our theorem
identifies a complete local correlation whose decay controls every
subsystem input, and is in turn bounded by the worst-input local
error. Establishing that decay for a broad family of constant-range,
noncommuting Hamiltonians remains the main dynamical problem.

Our criterion requires decay of a \emph{temporal} local correlation.
Any independent dynamical estimate of this decay can therefore be
translated into an equilibration time. For example, if
for some \(C,\tau>0\)
\[
 |\varphi_S(t)-\pi_{0,S}|
 \le C\left(1+\frac{|t|}{\tau}\right)^{-\alpha},
\]
then Eq.~\eqref{eq:A} gives, for \(T\gg\tau\),
\begin{equation}
 A_S(T)\lesssim C^2
 \begin{cases}
 (\tau/T)^{2\alpha},&0<\alpha<1/2,\\
 (\tau/T)\log(T/\tau),&\alpha=1/2,\\
 \tau/T,&\alpha>1/2,
 \end{cases}
 \label{eq:temporal-decay-discussion}
\end{equation}
and the same scaling then bounds \(\Gwin_H(T)\). Exponential temporal
decay gives the last, \(O_S(\tau/T)\), behaviour for this every-window
mean-square criterion.

Spatial clustering and entanglement area laws concern a different
property. In one dimension, exponential decay of spatial correlations
is sufficient for an entanglement area law,
\cite{BrandaoHorodecki2015AreaLaw} and area laws are closely connected
to efficient tensor-network descriptions of many-body states.
\cite{EisertCramerPlenio2010} These static properties do not by
themselves imply the temporal decay required here. Many-body
localized systems give an important counterexample: they fail to
thermalise and can retain local information about their initial
conditions for arbitrarily long times.
\cite{NandkishoreHuse2015} For local Hamiltonians, however, combining
locality and spatial clustering, or efficient tensor-network
approximations of local Heisenberg observables, may provide a useful
route to estimating the temporal autocorrelation in our theorem. We
leave this connection for future work.

Random spatially local interactions and local random-matrix
models provide possible settings in which to pursue this
question.\cite{MorampudiLaumann2019,
PollockKrothPagliaroliIadecolaRiddell2025}
Recent commutator-gap results constrain approximately conserved
local directions under their assumptions.\cite{CedilloCotlerHuang2026}
They do not yet give the estimate required here: an operator may
change substantially while remaining within the local operator
space. A bound on its return to that whole space would connect
such dynamical results directly to local equilibration.

Recent propagation algorithms provide controlled approximations to
local Heisenberg observables in weakly interacting systems,
including finite intervals of noiseless evolution.\cite{FacelliEtAl2026Propagation}
Such estimates may help evaluate the autocorrelation entering our
bounds. The computational and dynamical questions are distinct:
an accurate approximation to an evolved observable does not by
itself imply small return to the subsystem. Applying these methods
to our criterion would also require control of time-integration
errors and of the stationary contribution.

The fixed subsystem size and maximally mixed, initially
uncorrelated bath are essential to our statements. The constants
can grow with the observed region, and finite-temperature bath
states require a separate treatment. The time averages also allow
recurrences; they bound the fraction of bad times in each window.
Within this setting, the local autocorrelation provides a way to
study finite-time relaxation without resolving the microscopic
spectrum.

\section{Methods}

\subsection{Operational quantities and exact constants}

For density matrices, \(\Dt(\rho,\sigma)=\|\rho-\sigma\|_1/2\).
The temporal channel is
\[
 \overline\Phi(X)=
\lim_{L\to\infty}\frac1L\int_0^L\Phi_t(X)\,dt.
\]
The limit exists in finite dimension, including degenerate energies.
Define
\[
\Gwin_H(T)=\sup_{t_0\ge0}\frac1T\int_{t_0}^{t_0+T}
\sup_{\rho_S}\Dt(\Phi_t(\rho_S),\overline\Phi(\rho_S))^2dt.
\]
With \(q=d^2-1\), the exact rectangular-window comparisons used in the
main text are
\begin{align}
 \frac{d-1}{2d}A_S(T)
 &\le\Gwin_H(T)\le(d-1)q^2A_S(T),
 \label{eq:exact-G}\\
 \frac{d-1}{2d}\mathcal A_S(T)
 &\le\Ewin_H(T)\le(d-1)q^2\mathcal A_S(T),
 \label{eq:exact-E-full}\\
 \Ewin_H(T)
 &\le
 (d-1)q^2
 \left[
 \sqrt{A_S(T)}+\frac{\pi_{0,S}}2
 \right]^2.
 \label{eq:exact-E}
\end{align}
The lower constants improve the more elementary probe-state estimate:
they follow from an exact second-moment identity over pure subsystem
states and the trace-zero structure of the channel difference.  For
the separated upper bound, put
\[
 s_0=
 \sup_{\rho_S}
 \Dt\!\left(
 \overline\Phi(\rho_S),\frac{\Id_S}{d}
 \right).
\]
The trace-distance triangle inequality followed by Minkowski's
inequality in the time variable gives
\begin{equation}
 \sqrt{\Ewin_H(T)}
 \le
 \sqrt{\Gwin_H(T)}+s_0.
 \label{eq:Minkowski-separated-methods}
\end{equation}
Using
\(s_0\le q\sqrt{d-1}\,\pi_{0,S}/2\) and
Eq.~\eqref{eq:exact-G} gives Eq.~\eqref{eq:exact-E}.

More generally, for every normalised \(w\ge0\) with
\(\widehat w(\omega)\ge0\), define
\[
 \mathcal C_w=
 \int w(t)|\varphi_S(t)-\pi_{0,S}|^2dt
\]
and
\[
 \mathcal G_w=
 \sup_{c\in\mathbb R}\int w(t-c)
 \sup_{\rho_S}
 \Dt(\Phi_t(\rho_S),\overline\Phi(\rho_S))^2dt.
\]
Then
\begin{equation}
 \frac{d-1}{2d}\mathcal C_w
 \le \mathcal G_w
 \le \frac{(d-1)q^2}{4}\mathcal C_w.
 \label{eq:general-positive-kernel-methods}
\end{equation}
The same statement holds for the full correlation and the error to
\(\Id_S/d\).

\subsection{Reference-assisted channel distinguishability}

For a linear Hermiticity-preserving map \(\Lambda\) on a
\(d\)-dimensional system, write
\[
 \|\Lambda\|_{2\to2}
 =
 \sup_{X\ne0}\frac{\|\Lambda(X)\|_2}{\|X\|_2}.
\]
A standard finite-dimensional norm comparison is
\begin{equation}
 \|\Lambda\|_\diamond
 \le
 d\,\|\Lambda\|_{2\to2}.
 \label{eq:diamond-vs-2to2}
\end{equation}
This is the upper inequality in Eq.~(2.2) of Chen, Ding and
Huang.\cite{ChenDingHuang2022}  They formulate it for real
superoperators on a \(d\)-dimensional system.  For the operational
channel-discrimination interpretation of the diamond norm see, for
example, Ref.~\cite{Watrous2018}.  We use this only as an additional
corollary of the trace-distance result above, not as a replacement for
it.  For
\(\Lambda_t=\Phi_t-\overline\Phi\), or
\(\Lambda_t=\Phi_t-\mathcal D\), both
\(\Lambda_t(\Id_S)=0\) and
\(\Tr\Lambda_t(X)=0\).  The identity/traceless decomposition therefore
gives
\[
 \|\Lambda_t\|_{2\to2}
 =
 \|K_*(t)\|_{\rm op},
\]
where \(K_*(t)\) is respectively the nonzero-frequency or full local
transfer block.  Hence, pointwise,
\begin{equation}
 \Delta_\diamond(\Phi_t,\Psi)^2
 \le
 \frac{d^2}{4}\|K_*(t)\|_{\rm op}^2
 \le
 \frac{d^2}{4}\|K_*(t)\|_{\rm F}^2,
 \label{eq:diamond-by-local-return}
\end{equation}
with \(\Psi=\overline\Phi\) or \(\mathcal D\).
Combining Eq.~\eqref{eq:diamond-by-local-return} with the same
positive-kernel return estimate used above proves
Eq.~\eqref{eq:diamond-main}.  This is an immediate channel-norm
corollary of the complete-local theorem; the norm inequality itself is
standard.

The trace-distance functionals above are unassisted; the diamond-norm
corollary allows an arbitrary external reference.  The functionals are
Borel measurable in \(H\):
the input set is compact, the finite-time maps are continuous, and the
window-origin supremum can be taken over nonnegative rational origins.
The temporal channel is a pointwise limit of continuous finite-time
averages.  Supplementary Notes~S1 and~S2 provide the full argument and
all norm estimates.

\subsection{Positive local spectral weights and the time-window argument}

Give operator space the inner product
\(\langle X,Y\rangle=D^{-1}\Tr(X^\dagger Y)\).
Let \(P_S\) project onto
\(\calV_S=\{X_S\otimes\Id_B:\Tr X_S=0\}\).
The Liouvillian \(\calL=[H,\cdot]\) is self-adjoint in this inner product,
with spectral projectors \(Q_\omega\). Put
\[
 A_\omega=P_SQ_\omega P_S,\qquad
 \pi_\omega=q^{-1}\Tr_{\calK}A_\omega.
\]
Here \(\Tr_{\calK}\) is the trace on operator space. Each \(A_\omega\)
is positive, \(\sum_\omega\pi_\omega=1\), and
\(\mu_S=\sum_\omega\pi_\omega\delta_\omega\) has characteristic
function \(\varphi_S\).

The local return matrix is
\[
 C_{ba}(t)=D^{-1}\Tr[O_b e^{-itH}O_a e^{itH}],
\]
and \(R_S(t)=\|C(t)\|_{\rm F}^2\).
Expansion of a subsystem state in the \(P_a\) basis gives, pointwise,
\[
 \sup_{\rho_S}\Dt(\Phi_t(\rho_S),\tau_S)^2
 \le\frac{d-1}{4}R_S(t).
\]
The normalised quantity \(R_S(t)/q\) is also the unitarity of the
traceless transfer block in the terminology of channel
characterisation.\cite{WallmanGranadeHarperFlammia2015}
Here it describes information returning to the complete local algebra.
Let
\[
 \ket{\Omega}=\frac1{\sqrt d}\sum_{j=1}^{d}\ket j\ket j,
 \qquad
 J(\Phi_t)=(\Phi_t\otimes\mathrm{id})(\ket{\Omega}\!\bra{\Omega}).
\]
With this normalized Choi convention,
\begin{equation}
 \Tr\!\left[J(\Phi_t)^2\right]
 =\frac{1+R_S(t)}{d^2}.
 \label{eq:choi-purity-return-methods}
\end{equation}
Indeed, the transfer matrix has an identity block equal to one and the
traceless block \(C(t)\), while normalized Choi purity is \(d^{-2}\)
times its squared Hilbert--Schmidt norm. Related channel-purity
dynamics has recently been studied under the name Choi echo.
\cite{deLeonGonzalezDiazMejia2026}

For the reverse direction, let \(\Lambda_t\) denote either
\[
 X\mapsto\Phi_t(X)-\Tr(X)\Id_S/d
\]
or \(\Phi_t-\overline\Phi\), with traceless transfer matrix \(C_*(t)\).
For a Haar-random pure state \(\psi\),
\[
 \mathbb E_\psi\|\Lambda_t(\psi)\|_2^2
 =\frac{\|C_*(t)\|_{\rm F}^2}{d(d+1)}.
\]
This is the same second-moment structure that underlies the unitarity
of a quantum channel.\cite{WallmanGranadeHarperFlammia2015}
Moreover every \(\Lambda_t(\psi)\) is Hermitian and traceless, so
\(\|X\|_1^2\ge2\|X\|_2^2\).  Since the worst-input error is at least its
Haar average,
\begin{equation}
 \sup_{\rho_S}\Dt(\Lambda_t(\rho_S),0)^2
 \ge
 \frac{\|C_*(t)\|_{\rm F}^2}{2d(d+1)}.
 \label{eq:haar-pointwise-lower-methods}
\end{equation}
Supplementary Note~S2 gives the elementary proof.

The use of positive-Fourier time weights follows the same basic idea as
Refs.~\cite{Vikram2026,VikramNoAverages2026}. Appendix~A.1 and
Theorem~A.1 of Ref.~\cite{VikramNoAverages2026} give the general
inner-product autocorrelation-to-correlator construction, and
Appendix~C.2 gives its proof. In the present problem the target is more
specific: it is the quadratic return of the complete local operator
subspace itself. The spectral expansion of \(R_S\) has coefficients
\[
 0\le\Tr_{\calK}(A_\omega A_\nu)
       \le q^2\pi_\omega\pi_\nu .
\]
The first inequality follows from \(A_\omega,A_\nu\ge0\). For the
second, \(A_\omega\le(\Tr A_\omega)\Id\) gives
\[
 \Tr(A_\omega A_\nu)
 \le\Tr(A_\omega)\Tr(A_\nu).
\]
For a normalised positive-Fourier weight \(w\), this coefficientwise
bound gives Eq.~\eqref{eq:general-positive-kernel-methods} uniformly
under translations of \(w\). The generic Cauchy--Schwarz step in
Appendix~C.2 of Ref.~\cite{VikramNoAverages2026} is replaced here by
a linear comparison of two quadratic quantities built from the same
positive compressed spectral blocks.

For the triangular weight
\[
 \mathcal F_S(T)=\frac1T\int_{-T}^{T}
    \left(1-\frac{|t|}{T}\right)|\varphi_S(t)|^2dt,
\]
a translated tent is at least \(1/2\) on a rectangular window of length
\(T\).  Its Fourier transform is a phase times
\(\operatorname{sinc}^2(T\Delta/2)\).  Hence
\begin{equation}
 \Ewin_H(T)\le\mathfrak c_d\mathcal F_S(T),\qquad
 \mathfrak c_d=\frac{(d-1)q^2}{2}.
 \label{eq:kernel-main}
\end{equation}
Restricting both frequencies to nonzero values gives the corresponding
bound for \(\Gwin_H\).

For spectral applications, define
\[
 \chi_S(\delta)=\sup_x\mu_S([x,x+\delta)).
\]
This is the L\'evy concentration function of the Liouvillian
probability measure seen by the complete local operator algebra; see
Ref.~\cite{Petrov1975} for general Fourier bounds on concentration
functions.  It is related to, but not identical with, the energy-band
quantities of Ref.~\cite{Vikram2026}: the supremum is over every
Liouvillian-frequency interval of width \(\delta\), and the measure is
defined for the complete local algebra rather than one selected
observable.

Partitioning the positive sinc-squared kernel into shells of width
\(\delta\) gives
\begin{equation}
 \Ewin_H(T)\le(d-1)q^2\chi_S(\delta)
       \left[1+\frac{2\pi^2}{3(T\delta)^2}\right].
 \label{eq:band-main}
\end{equation}
The shell series is summable, which allows the direct choice
\(\delta=1/T\).  This differs from replacing the entire outside-band
contribution by a single worst Fourier tail.  For example, the
rectangular-output/tent-input specialisation discussed after
Theorem~A.1 of Ref.~\cite{VikramNoAverages2026} has
\(w_0=T_{\rm obs}/(\xi T)\).  The shell-by-shell local concentration
estimate remains useful already at \(T\delta=1\).  We use this
same-scale form for the clock application.

\subsection{Stationary memory and the energy-conservation floor}

The temporal traceless block is \(A_0\), and
\[
 \sup_{\rho_S}\Dt(\overline\Phi(\rho_S),\tau_S)^2
 \le\frac{d-1}{4}\Tr_{\calK}A_0^2
 \le\frac{d-1}{4}q^2\pi_{0,S}^2.
\]
For nondegenerate energies, define
\[
 \delta_S(H)=\frac1D\sum_j
 \Dt\!\left(\Tr_B|j\rangle\langle j|,\tau_S\right).
\]
Parseval gives \(\pi_{0,S}\le2d\,\delta_S(H)/(d^2-1)\).
For degenerate energies, a bound in one chosen eigenbasis need not
control the temporal channel. We then use the basis-independent
block \(A_0\).

For the energy-floor statement, take
\(H_N=\sum_e U_e\Lambda_eU_e^\dagger\), with independent Haar \(U_e\)
on bounded local supports. Centre each term to zero trace.
Assume at most \(zN\) terms, at most \(\Delta_*\) other overlapping
supports per term, and uniformly bounded normalised Hilbert--Schmidt
norm \(\bar h_2\). Suppose \(S\) contains one support \(e_0\) with
centred spectral variance \(a>0\), independent of \(N\).
Disjoint centred terms are orthogonal, giving
\(\|\widetilde H_N\|_{2,D}^2\le C_1N\), where
\(C_1=z(\Delta_*+1)\bar h_2^2\).
Projection onto the conserved direction \(\widetilde H_N\), followed by
Haar centring of the \(e_0\) term, yields
\[
 \mathbb E\pi_{0,S}\ge\frac{a}{q C_1N}.
\]
The proof, including zero-norm cases, is in Supplementary Note~\hyperref[sec:supp3]{S3}.
This ensemble is distinct from the commuting-patch ensemble.

\subsection{Periodic-clock model and spectral input}

In the one-particle clock sector, take
\[
 H_{\rm FK}=\Id-\frac12\sum_{t=0}^{n-1}
 \left(u_t\otimes|t+1\rangle\langle t|
       +u_t^\dagger\otimes|t\rangle\langle t+1|\right),
\]
with \(n\ge2\), periodic clock labels, and the maximal-LFSR Clifford gate sequence
of Ref.~\cite{IppolitiLong2026}. The spin--clock realization is
geometrically local; only its invariant one-clock-particle sector is
used. For spin-only \(S\),
\[
 \tau_B=\frac{\Id_{B_s}}{2^{n-s}}
             \otimes\frac{\Id_{\rm clock}^{(1)}}n .
\]
Every input in Eq.~\eqref{eq:prep} has mean energy one in this sector.
Each hopping term has zero clock trace, so the centred energy direction
is exactly orthogonal to the spin-local algebra.

Put \(L=2^n-1\) and \(M_c=nL\). The regular sector has energies
\(1-\cos[2\pi(mL+k)/(nL)]\), with \(0\le m<n\), \(1\le k<L\).
Two exceptional circuit eigenvectors produce \(2n\) exceptional static
states. The external Pauli matrix-element bound is
\(c_n=2^{1+n/2}/(2^n-1)\) between regular eigenstates; Clifford
conjugation and clock orthogonality transfer it from the circuit to
the static eigenstates.\cite{IppolitiLong2026}

For the complete \(M_c\)-point cosine grid, our difference count is
\[
 N_{M_c}(\delta)\le
 2M_c^2\delta\log(4e/\delta)+8M_c,\qquad 0<\delta\le1.
\]
Combining this with \(c_n\), and bounding exceptional rows and columns
by Pauli unitarity, gives
\[
 \chi_S(\delta)\le
 8n\delta\log(4e/\delta)
       +\frac{32}{2^n-1}+\frac4{2^n}.
\]
The limit \(\delta\downarrow0\) controls the zero atom.
Equations~\eqref{eq:band-main} and \eqref{eq:kernel-main} use pointwise
bounds uniform in the input, so the input supremum may be taken inside
the time integral. Supplementary Notes~S4 and~S5 give the full
diagonalisation, counting proof, and prime-power extension.

\subsection{Random commuting ensemble and its reference-time input}

In Eq.~\eqref{eq:patch-H}, \(M\) is even and patches have exactly
\(\xi\) qubits; we consider system sizes \(N=M\xi\) without a leftover
patch. The first matching of \(V\) pairs patches
\((2,3),(4,5),\ldots\); the second pairs
\((1,2),(3,4),\ldots\).
All gates are independent Haar unitaries and are independent of \(J\).
Each \(h_i=V^\dagger D_iV\) is supported within patches
\(i-2,\ldots,i+2\), truncated at a boundary. Each physical qubit
belongs to at most five term supports.

We use Theorem~3 of Ref.~\cite{CuiSchusterMaoHuangBrandao2025}
at time \(\pi\), specialised to exact Haar gates and two copies.
It gives
\[
 \xi\ge2\log_2(2N/\varepsilon)+C_0
 \ \Longrightarrow\
 \|\mathcal M^{(2)}_{\pi}-\mathcal M^{(2)}_{\rm Haar}\|_\diamond
 \le\varepsilon .
\]
Here \(\mathcal M^{(2)}(X)=\mathbb E U^{\otimes2}XU^{\dagger\otimes2}\).
For each fixed \(b>0\), arbitrarily large admissible system sizes
satisfy this condition with \(\varepsilon=N^{-b}\) and
\(\xi=O_b(\log N)\). Supplementary Note~\hyperref[sec:supp6]{S6} gives an explicit sequence.

Write each normalised local correlation as
\[
 c_a(U)=D^{-1}\Tr(O_aUO_aU^\dagger).
\]
The required two-copy scalar moment is
\[
 |\varphi_S(U)|^2=
 \frac1{q^2D^2}\sum_{a,b}
 \Tr[(O_a\otimes O_b)U^{\otimes2}
     (O_a\otimes O_b)U^{\dagger\otimes2}].
\]
Its Haar mean is \(1/[q(D^2-1)]\); the design error changes it by at
most \(\varepsilon\).

To relate this estimate to other times, write the eigenvectors as
\(V^\dagger|\mathbf z\rangle\), define
\[
 w_{\mathbf z\mathbf w}=
 \frac1{qD}\sum_a
 |\langle\mathbf z|VO_aV^\dagger|\mathbf w\rangle|^2,
\]
and group these positive weights by the integer vector \(\ell\) in
\(E_{\mathbf z}-E_{\mathbf w}=\ell\cdot J\). The resulting
\(W_\ell\) sum to one. The quantity
\(\beta(V)=\sum_\ell W_\ell^2\) obeys
\[
 \mathbb E_V\beta(V)\le
 b_N:=\varepsilon+\frac1{q(D^2-1)}.
\]
Indeed, uniform energies at time \(\pi\) annihilate every nonzero
integer multiplicity. Independent absolutely continuous energies
also imply \(\pi_{0,S}=W_0\) almost surely, and
\(\mathbb E\pi_{0,S}\le\varepsilon\).

For uniform energies, any nonidentical pair of labels contributes
at least two nonzero sinc factors, giving
\[
 \left|\mathbb E_J|\varphi_S(t)|^2-\beta(V)\right|
 \le \min\{1,t^{-2}\}.
\]
More generally, assume independent energy densities \(p_{i,z}\) on
\([-1,1]\), absolutely continuous with integrable derivative and
\[
 |p_{i,z}(-1)|+|p_{i,z}(1)|
       +\int_{-1}^1|p'_{i,z}(x)|dx\le L_0
\]
for a common \(L_0\ge1\). The remainder becomes
\(\min\{1,L_0^2/t^2\}\), whose positive-half-line integral is \(2L_0\).
Equation~\eqref{eq:kernel-main} then gives
\begin{equation}
 \mathbb E\Ewin_{H_N}(T)
 \le\mathfrak c_d\left[b_N+\frac{4L_0}{T}\right].
 \label{eq:ensemble-exact}
\end{equation}
The reference design estimate is used only to bound \(\beta(V)\)
and \(W_0(V)\); these do not depend on the actual energy law.
The probability and longer-window consequences follow by Markov's
inequality applied to \(\Ewin_H(T)\), not by exchanging an expectation
and a supremum.

\subsection{Transient lower bound and perturbation control}

For a pure input on \(s\) qubits,
\(\Dt(\rho_S,\tau_S)=a_d:=1-1/d\).
At most \(5s\) norm-one terms meet \(S\), hence
\(\|[H,\rho_0]\|_1\le10s\) and
\(\Dt(\Phi_t(\rho_S),\rho_S)\le5st\).
Integrating over \(0\le t\le a_d/(10s)\) proves
\[
 \Ewin_H(T)\ge\frac{a_d^3}{40sT},\qquad T\ge1.
\]

For arbitrary finite Hamiltonians \(G,H\) on the same space with
\(\|H-G\|_\infty\le\delta\), Duhamel's identity gives
\(|\varphi_S(t;H)-\varphi_S(t;G)|\le2\delta|t|\).
The triangular weight has second moment \(T^2/6\), so
\[
 \Ewin_H(T)\le\mathfrak c_d
 \left[2\mathcal F_S(T;G)+\frac43\delta^2T^2\right].
\]
For \(G\) drawn from the patch ensemble this becomes
\[
 \mathbb E\Ewin_H(T)\le\mathfrak c_d
 \left[2b_N+\frac{8L_0}{T}+\frac43\delta^2T^2\right].
\]
The argument compares correlations only on \([-T,T]\). Uniformity
over arbitrarily late window origins comes from
Eq.~\eqref{eq:kernel-main}, not from applying Duhamel's estimate
at the origin of a late window. Supplementary Note~\hyperref[sec:supp6]{S6} contains the
complete proofs.

\subsection{Reproducibility and assisted preparation}

The results are analytic. 

ChatGPT (OpenAI) assisted with proof-checking derivations, preparing
consistency scripts, and editing the manuscript. The author is
responsible for the scientific content. The assumptions, external
inputs, and proofs are stated in the article and Supplementary
Information.

\section*{Data availability}

No experimental data were used. The derivations are contained in
this article and its Supplementary Information. 

\section*{Code availability}
The analytic proofs do not depend on numerical computations.

\section*{Author contributions}
P.\'{C}. formulated the research question and directed the theoretical
development and manuscript preparation.

\section*{Acknowledgements}
P.\'{C}. acknowledges discussions with Amit Vikram about the
results presented in this work, and especially thanks him for helpful
discussions of his contributions to finite-time quantum thermalisation,
for drawing attention to relevant earlier results, and for clarifying
their relation to the results presented here.
This work is supported by the IRAP/MAB program, project no.
FENG.02.01IP.05-0006/23, financed by the MAB FENG program 2021--2027,
Priority FENG.02, Measure FENG.02.01, with the support of the FNP
(Foundation for Polish Science).

\section*{Competing interests}
The author declares no competing interests.

\clearpage
\section*{References}
\bibliography{references}

\clearpage
\onecolumngrid

\setcounter{section}{0}
\setcounter{subsection}{0}
\setcounter{subsubsection}{0}
\setcounter{equation}{0}
\setcounter{theorem}{0}
\setcounter{figure}{0}
\setcounter{table}{0}
\renewcommand{\thesection}{S\arabic{section}}
\renewcommand{\theequation}{S\arabic{equation}}
\renewcommand{\thetheorem}{S\arabic{theorem}}
\renewcommand{\thefigure}{S\arabic{figure}}
\renewcommand{\thetable}{S\arabic{table}}
\providecommand{\theHsection}{}
\providecommand{\theHequation}{}
\providecommand{\theHtheorem}{}
\renewcommand{\theHsection}{supp.\arabic{section}}
\renewcommand{\theHequation}{supp.\arabic{equation}}
\renewcommand{\theHtheorem}{supp.\arabic{theorem}}
\pdfbookmark[0]{Supplementary Information}{supplementary-information}
\begin{center}
{\Large\bfseries Supplementary Information\par}
\vspace{0.7em}
{\large Finite-time quantum equilibration from local correlations\par}
\vspace{0.7em}
Piotr {\'C}wikli{\'n}ski
\end{center}
\vspace{0.7em}

\par\medskip

\noindent
In these Supplementary Notes we give the proofs of the results in
the main article, using the setup and model definitions from
Methods. Notes~S1--S2 establish the channel bounds and the
positive-kernel argument. We then bound the stationary memory in
Note~S3 and treat the periodic clocks and prime-power extension in
Notes~S4--S5. Note~S6 contains the random commuting-ensemble
result, its probability bounds, and the transient and perturbation
estimates. Each external theorem is identified where it is used.
No generic noncommuting-ensemble assumption is needed.

\section*{Contents}
\begingroup
\renewcommand{\arraystretch}{1.55}
\begin{tabular}{@{}p{0.07\textwidth}p{0.80\textwidth}r@{}}
S1 & \hyperref[sec:supp1]{Local channels and complete operator return} & \pageref{sec:supp1}\\
S2 & \hyperref[sec:supp2]{Finite-window bounds and the autocorrelation characterisation} & \pageref{sec:supp2}\\
S3 & \hyperref[sec:supp3]{Conserved local memory} & \pageref{sec:supp3}\\
S4 & \hyperref[sec:supp4]{The periodic LFSR clock} & \pageref{sec:supp4}\\
S5 & \hyperref[sec:supp5]{Structural clock extension, prime powers, and sector preparation} & \pageref{sec:supp5}\\
S6 & \hyperref[sec:supp6]{The structured random commuting ensemble} & \pageref{sec:supp6}\\
\end{tabular}
\endgroup
\clearpage

\section{Local channels and complete operator return}\label{sec:supp1}
\subsection{Finite-dimensional operator-space setup}

Let
\[
  \mathcal H=\mathcal H_S\otimes\mathcal H_B,
  \qquad
  d=\dim\mathcal H_S,
  \qquad
  m=\dim\mathcal H_B,
  \qquad
  D=dm.
\]
Let \(H=H^\dagger\) be an arbitrary Hamiltonian on \(\mathcal H\), and put
\[
  U_t=e^{-itH}.
\]
We assume \(d\ge2\); a one-dimensional subsystem has zero trace-distance
error and needs no traceless-algebra construction.
No locality assumption is needed for the operator-space results in the present part.

We equip \(\calB(\mathcal H)\) with the normalised Hilbert--Schmidt inner
product
\begin{equation}
  \langle X,Y\rangle
  :=
  \frac{1}{D}\Tr(X^\dagger Y),
  \qquad
  \norm{X}_{2,D}
  :=
  \sqrt{\langle X,X\rangle}.
  \label{eq:normalized-HS-global}
\end{equation}
On \(\calB(\mathcal H_S)\) we use
\begin{equation}
  \langle A,B\rangle_S
  :=
  \frac{1}{d}\Tr(A^\dagger B).
  \label{eq:normalized-HS-system}
\end{equation}

Choose a Hermitian operator basis
\[
  P_0=\Id_S,\quad P_1,\ldots,P_q,
  \qquad
  q=d^2-1,
\]
such that
\begin{equation}
  \Tr(P_aP_b)=d\,\delta_{ab},
  \qquad
  \Tr P_a=0
  \quad
  (a,b=1,\ldots,q).
  \label{eq:local-operator-basis}
\end{equation}
For qubits, these may be chosen as Pauli strings.

Define the embedded traceless local operators
\begin{equation}
  O_a=P_a\otimes\Id_B,
  \qquad
  a=1,\ldots,q.
  \label{eq:embedded-local-basis}
\end{equation}
By Eqs.~\eqref{eq:normalized-HS-global} and
\eqref{eq:local-operator-basis},
\[
  \langle O_a,O_b\rangle=\delta_{ab}.
\]
Thus \(\{O_a\}_{a=1}^q\) is an orthonormal basis of the traceless local
operator subspace
\begin{equation}
  \calV_S
  :=
  \left\{
  A_S\otimes\Id_B:
  \Tr A_S=0
  \right\}
  \subset\calB(\mathcal H).
  \label{eq:local-operator-subspace}
\end{equation}
Let \(\mathsf P_S\) denote the orthogonal projection onto \(\calV_S\) with
respect to Eq.~\eqref{eq:normalized-HS-global}.

\begin{lemma}[Explicit formula for the local operator projection]
\label{lem:local-projection}
For every \(X\in\calB(\mathcal H)\),
\begin{equation}
  \mathsf P_S(X)
  =
  \left[
  \frac{1}{m}\Tr_B X
  -
  \frac{\Tr X}{D}\Id_S
  \right]\otimes\Id_B .
  \label{eq:local-projection-formula}
\end{equation}
\end{lemma}

\begin{proof}
Define the right-hand side of Eq.~\eqref{eq:local-projection-formula} as
\(Y\).  First,
\[
  \Tr_S\left[
  \frac1m\Tr_B X-\frac{\Tr X}{D}\Id_S
  \right]
  =
  \frac{\Tr X}{m}
  -
  \frac{\Tr X}{dm}d
  =0.
\]
Therefore \(Y\in\calV_S\).

Let \(A_S\otimes\Id_B\in\calV_S\), so that \(\Tr A_S=0\).  Then
\begin{align}
  \langle A_S\otimes\Id_B,X-Y\rangle
  &=
  \frac1D
  \Tr\left[
  (A_S^\dagger\otimes\Id_B)(X-Y)
  \right]
  \nonumber\\
  &=
  \frac1D
  \Tr_S\left[
  A_S^\dagger\Tr_B(X-Y)
  \right].
  \label{eq:projection-orthogonality-step1}
\end{align}
Using Eq.~\eqref{eq:local-projection-formula},
\[
  \Tr_B Y
  =
  m\left[
  \frac1m\Tr_B X-\frac{\Tr X}{D}\Id_S
  \right]
  =
  \Tr_B X-\frac{\Tr X}{d}\Id_S.
\]
Hence
\[
  \Tr_B(X-Y)
  =
  \frac{\Tr X}{d}\Id_S.
\]
Substitution into Eq.~\eqref{eq:projection-orthogonality-step1} gives
\[
  \langle A_S\otimes\Id_B,X-Y\rangle
  =
  \frac{\Tr X}{dD}\Tr A_S^\dagger
  =0.
\]
Thus \(X-Y\) is orthogonal to \(\calV_S\), while \(Y\in\calV_S\).
Therefore \(Y=\mathsf P_S(X)\).
\end{proof}

\subsection{Infinite-temperature reduced dynamics}

Let the bath start in the maximally mixed state
\[
  \tau_B=\frac{\Id_B}{m}.
\]
The reduced channel is
\begin{equation}
  \Phi_t(X)
  :=
  \Tr_B\left[
  U_t
  \left(X\otimes\frac{\Id_B}{m}\right)
  U_t^\dagger
  \right].
  \label{eq:infinite-temp-channel}
\end{equation}
Because \(U_t\) is unitary,
\begin{equation}
  \Phi_t\left(\frac{\Id_S}{d}\right)
  =
  \frac{\Id_S}{d}.
  \label{eq:maximally-mixed-fixed}
\end{equation}

\begin{definition}[Local return matrix]
\label{def:return-matrix}
The \(q\times q\) local return matrix \(C(t)\) is defined by
\begin{equation}
  C_{ba}(t)
  :=
  \frac1d\Tr\left[
  P_b\Phi_t(P_a)
  \right],
  \qquad
  a,b=1,\ldots,q.
  \label{eq:return-matrix-def}
\end{equation}
\end{definition}

The same matrix can be written as an overlap of evolved and local
operators.

\begin{lemma}[Return matrix as a local operator overlap]
\label{lem:return-overlap}
For every \(a,b\),
\begin{equation}
  C_{ba}(t)
  =
  \frac1D
  \Tr\left[
  O_b\,U_t O_aU_t^\dagger
  \right]
  =
  \langle O_b,\mathcal U_t(O_a)\rangle,
  \label{eq:return-matrix-overlap}
\end{equation}
where
\[
  \mathcal U_t(X):=U_tXU_t^\dagger.
\]
\end{lemma}

\begin{proof}
By Eq.~\eqref{eq:infinite-temp-channel},
\begin{align}
  \frac1d\Tr_S[P_b\Phi_t(P_a)]
  &=
  \frac1d
  \Tr_{SB}\left[
  (P_b\otimes\Id_B)
  U_t
  \left(P_a\otimes\frac{\Id_B}{m}\right)
  U_t^\dagger
  \right]
  \nonumber\\
  &=
  \frac1{dm}
  \Tr_{SB}\left[
  O_bU_tO_aU_t^\dagger
  \right].
\end{align}
Since \(D=dm\), this is Eq.~\eqref{eq:return-matrix-overlap}.
\end{proof}

\begin{lemma}[Expansion of a traceless operator]
\label{lem:traceless-expansion}
For every Hermitian \(X\in\calB(\mathcal H_S)\) with \(\Tr X=0\),
\begin{equation}
  X
  =
  \frac1d\sum_{a=1}^q x_aP_a,
  \qquad
  x_a=\Tr(P_aX),
  \label{eq:traceless-expansion}
\end{equation}
and
\begin{equation}
  \sum_{a=1}^q x_a^2
  =
  d\,\Tr(X^2).
  \label{eq:parseval-system}
\end{equation}
\end{lemma}

\begin{proof}
Because
\(\{\Id_S/\sqrt d,P_1/\sqrt d,\ldots,P_q/\sqrt d\}\)
is an orthonormal basis with respect to the unnormalised Hilbert--Schmidt
inner product, every \(X\) has the expansion
\[
  X
  =
  \frac{\Tr X}{d}\Id_S
  +
  \frac1d\sum_{a=1}^q\Tr(P_aX)P_a.
\]
The first term vanishes because \(\Tr X=0\), proving
Eq.~\eqref{eq:traceless-expansion}.  Squaring the expansion and using
\(\Tr(P_aP_b)=d\delta_{ab}\) gives
\begin{align}
  \Tr(X^2)
  &=
  \frac1{d^2}
  \sum_{a,b}x_ax_b\Tr(P_aP_b)
  \nonumber\\
  &=
  \frac1d\sum_a x_a^2.
\end{align}
This is Eq.~\eqref{eq:parseval-system}.
\end{proof}

\begin{lemma}[Trace norm versus Hilbert--Schmidt norm]
\label{lem:trace-HS}
For every operator \(X\) on a \(d\)-dimensional Hilbert space,
\begin{equation}
  \norm{X}_1
  \le
  \sqrt d\,\norm{X}_2.
  \label{eq:trace-HS}
\end{equation}
\end{lemma}

\begin{proof}
Let \(s_1,\ldots,s_d\ge0\) be the singular values of \(X\).  Then
\[
  \norm{X}_1=\sum_{j=1}^d s_j,
  \qquad
  \norm{X}_2=\left(\sum_{j=1}^d s_j^2\right)^{1/2}.
\]
The scalar Cauchy--Schwarz inequality gives
\[
  \left(\sum_{j=1}^d s_j\right)^2
  \le
  \left(\sum_{j=1}^d1^2\right)
  \left(\sum_{j=1}^ds_j^2\right)
  =
  d\,\norm{X}_2^2.
\]
Taking square roots proves Eq.~\eqref{eq:trace-HS}.
\end{proof}

\subsection{From local return to trace distance}

\begin{theorem}[Local operator return controls infinite-temperature thermalisation]
\label{thm:return-controls-thermalisation}
For every Hamiltonian \(H\), every subsystem input state \(\rho_S\), and
every \(t\),
\begin{equation}
  \Dt\left(
  \Phi_t(\rho_S),
  \frac{\Id_S}{d}
  \right)^2
  \le
  \frac{d-1}{4}
  \norm{C(t)}_{\mathrm{op}}^2
  \le
  \frac{d-1}{4}
  \norm{C(t)}_{\mathrm F}^2.
  \label{eq:return-controls-thermalisation}
\end{equation}
Here \(\norm{\cdot}_{\mathrm{op}}\) and
\(\norm{\cdot}_{\mathrm F}\) are respectively the Euclidean operator and
Frobenius norms of the \(q\times q\) matrix \(C(t)\).
\end{theorem}

\begin{proof}
Put
\[
  \Delta
  =
  \rho_S-\frac{\Id_S}{d}.
\]
Then \(\Delta=\Delta^\dagger\) and \(\Tr\Delta=0\).  By
Lemma~\ref{lem:traceless-expansion},
\begin{equation}
  \Delta
  =
  \frac1d\sum_{a=1}^q x_aP_a,
  \qquad
  x_a=\Tr(P_a\Delta).
  \label{eq:delta-expansion-proof}
\end{equation}
Since \(\Tr P_a=0\),
\[
  x_a=\Tr(P_a\rho_S).
\]
Again by Lemma~\ref{lem:traceless-expansion},
\begin{align}
  \norm{x}_2^2
  &=
  d\,\Tr(\Delta^2)
  \nonumber\\
  &=
  d\Tr\left[
  \rho_S^2-\frac{2}{d}\rho_S+\frac{\Id_S}{d^2}
  \right]
  \nonumber\\
  &=
  d\Tr(\rho_S^2)-2+\frac{\Tr\Id_S}{d}
  \nonumber\\
  &=
  d\Tr(\rho_S^2)-1
  \nonumber\\
  &\le d-1,
  \label{eq:x-bound-proof}
\end{align}
because every density operator satisfies
\(\Tr(\rho_S^2)\le1\).

By linearity and Eq.~\eqref{eq:maximally-mixed-fixed},
\begin{equation}
  \Phi_t(\rho_S)-\frac{\Id_S}{d}
  =
  \Phi_t(\Delta).
  \label{eq:output-difference-proof}
\end{equation}
Expand the output as
\[
  \Phi_t(\Delta)
  =
  \frac1d\sum_{b=1}^q y_bP_b.
\]
Its coefficients are
\begin{align}
  y_b
  &=
  \Tr[P_b\Phi_t(\Delta)]
  \nonumber\\
  &=
  \frac1d\sum_a x_a
  \Tr[P_b\Phi_t(P_a)]
  \nonumber\\
  &=
  \sum_aC_{ba}(t)x_a.
\end{align}
Thus
\begin{equation}
  y=C(t)x.
  \label{eq:y-Cx}
\end{equation}

Applying Lemma~\ref{lem:trace-HS} and then
Lemma~\ref{lem:traceless-expansion},
\begin{align}
  \Dt\left(
  \Phi_t(\rho_S),\frac{\Id_S}{d}
  \right)^2
  &=
  \frac14\norm{\Phi_t(\Delta)}_1^2
  \nonumber\\
  &\le
  \frac d4\norm{\Phi_t(\Delta)}_2^2
  \nonumber\\
  &=
  \frac d4
  \left(
  \frac1d\sum_b y_b^2
  \right)
  \nonumber\\
  &=
  \frac14\norm{y}_2^2
  \nonumber\\
  &=
  \frac14\norm{C(t)x}_2^2
  \nonumber\\
  &\le
  \frac14\norm{C(t)}_{\mathrm{op}}^2\norm{x}_2^2
  \nonumber\\
  &\le
  \frac{d-1}{4}
  \norm{C(t)}_{\mathrm{op}}^2.
\end{align}
Finally,
\(\norm{C}_{\mathrm{op}}\le\norm{C}_{\mathrm F}\),
because the largest singular value is no larger than the square root of
the sum of squared singular values.  This proves
Eq.~\eqref{eq:return-controls-thermalisation}.
\end{proof}

\begin{remark}
The quantity
\[
  \frac{1}{d^2-1}\norm{C(t)}_{\mathrm F}^2
\]
is the unitarity of the traceless block of the reduced channel in the
standard channel-transfer-matrix terminology; see
Ref.~\cite{WallmanGranadeHarperFlammia2015}.  Here it is interpreted
dynamically as the mean return of a complete local operator basis.
\end{remark}

\subsection{Coordinate-free local return}

Define the unitary conjugation superoperator
\begin{equation}
  \mathcal U_t(X)=U_tXU_t^\dagger .
  \label{eq:heisenberg-superoperator}
\end{equation}
It is unitary on the operator Hilbert space
\((\calB(\mathcal H),\langle\cdot,\cdot\rangle)\), because
\begin{align}
  \langle\mathcal U_t(X),\mathcal U_t(Y)\rangle
  &=
  \frac1D
  \Tr\left[
  U_tX^\dagger U_t^\dagger
  U_tYU_t^\dagger
  \right]
  \nonumber\\
  &=
  \frac1D\Tr(X^\dagger Y)
  =
  \langle X,Y\rangle.
\end{align}

\begin{definition}[Many-body local return superoperator]
\label{def:manybody-return}
Define
\begin{equation}
  \mathsf K_S(t)
  :=
  \mathsf P_S\mathcal U_t\mathsf P_S
  :
  \calV_S\longrightarrow\calV_S.
  \label{eq:manybody-return}
\end{equation}
Its squared Hilbert--Schmidt norm, which sums the return over all
\(q\) basis directions, is
\begin{equation}
  R_S(t)
  :=
  \norm{\mathsf K_S(t)}_{\mathrm F}^2.
  \label{eq:return-probability}
\end{equation}
\end{definition}

\begin{proposition}
\label{prop:return-equals-C}
For every \(t\),
\begin{equation}
  R_S(t)=\norm{C(t)}_{\mathrm F}^2.
  \label{eq:return-equals-C}
\end{equation}
\end{proposition}

\begin{proof}
In the orthonormal basis \(\{O_a\}_{a=1}^q\) of \(\calV_S\), the matrix
entries of \(\mathsf K_S(t)\) are
\begin{align}
  \langle O_b,\mathsf K_S(t)O_a\rangle
  &=
  \langle O_b,\mathsf P_S\mathcal U_tO_a\rangle
  \nonumber\\
  &=
  \langle \mathsf P_SO_b,\mathcal U_tO_a\rangle
  \nonumber\\
  &=
  \langle O_b,\mathcal U_tO_a\rangle
  =
  C_{ba}(t),
\end{align}
where we used \(\mathsf P_SO_b=O_b\) and
Lemma~\ref{lem:return-overlap}.  The squared Frobenius norm is the sum of
the squared moduli of the matrix entries, hence
\[
  R_S(t)
  =
  \sum_{a,b=1}^q\abs{C_{ba}(t)}^2
  =
  \norm{C(t)}_{\mathrm F}^2.
\]
\end{proof}

Combining Theorem~\ref{thm:return-controls-thermalisation} and
Proposition~\ref{prop:return-equals-C} gives
\begin{equation}
  \Dt\left(
  \Phi_t(\rho_S),\frac{\Id_S}{d}
  \right)^2
  \le
  \frac{d-1}{4}R_S(t).
  \label{eq:trace-distance-by-return-probability}
\end{equation}

\subsection{Exact Liouvillian spectral representation}

Define the Liouvillian
\begin{equation}
  \calL(X)=[H,X].
  \label{eq:Liouvillian}
\end{equation}
Because \(H=H^\dagger\), the Liouvillian is self-adjoint with respect to
Eq.~\eqref{eq:normalized-HS-global}:
\begin{align}
  \langle X,\calL(Y)\rangle
  &=
  \frac1D\Tr[X^\dagger(HY-YH)]
  \nonumber\\
  &=
  \frac1D\Tr[(X^\dagger H-HX^\dagger)Y]
  \nonumber\\
  &=
  \frac1D\Tr([H,X]^\dagger Y)
  =
  \langle\calL(X),Y\rangle.
\end{align}
Therefore \(\calL\) has an orthogonal spectral decomposition
\begin{equation}
  \calL
  =
  \sum_{\omega}\omega\,\mathsf Q_\omega,
  \qquad
  \sum_\omega\mathsf Q_\omega=\Id_{\calB(\mathcal H)},
  \label{eq:Liouvillian-spectral-decomp}
\end{equation}
where distinct values of \(\omega\) are the distinct Bohr frequencies,
including \(\omega=0\), and \(\mathsf Q_\omega\) are the corresponding
orthogonal projectors.  Consequently,
\begin{equation}
  \mathcal U_t=e^{-it\calL}
  =
  \sum_\omega e^{-i\omega t}\mathsf Q_\omega.
  \label{eq:Liouvillian-time-evolution}
\end{equation}

For a bounded scalar function \(f\), write
\begin{equation}
  \av{f(t)}
  :=
  \lim_{T\to\infty}
  \frac1T\int_0^T f(t)\,dt
  \label{eq:Cesaro-average}
\end{equation}
whenever the limit exists.  Every function below is a finite sum of
exponentials, so the limit exists.

\begin{lemma}[Cesaro orthogonality]
\label{lem:Cesaro-orthogonality}
For real numbers \(\omega,\nu\),
\begin{equation}
  \lim_{T\to\infty}
  \frac1T\int_0^T
  e^{it(\omega-\nu)}\,dt
  =
  \begin{cases}
  1,&\omega=\nu,\\
  0,&\omega\ne\nu.
  \end{cases}
  \label{eq:Cesaro-orthogonality}
\end{equation}
\end{lemma}

\begin{proof}
If \(\omega=\nu\), the integrand is \(1\), and the average equals \(1\).
If \(\omega\ne\nu\), direct integration gives
\[
  \frac1T\int_0^T
  e^{it(\omega-\nu)}dt
  =
  \frac{e^{iT(\omega-\nu)}-1}
  {iT(\omega-\nu)}.
\]
The numerator has absolute value at most \(2\), while the denominator has
absolute value \(T\abs{\omega-\nu}\).  The ratio therefore converges to
zero.
\end{proof}

\begin{theorem}[Exact infinite-time local return formula]
\label{thm:exact-return-spectral}
The infinite-time average of the total local return is
\begin{equation}
  \av{R_S(t)}
  =
  \sum_\omega
  \norm{
  \mathsf P_S\mathsf Q_\omega\mathsf P_S
  }_{\mathrm F}^2.
  \label{eq:exact-return-spectral}
\end{equation}
Equivalently, in the basis \(\{O_a\}\),
\begin{equation}
  \av{R_S(t)}
  =
  \sum_{a,b=1}^q
  \sum_\omega
  \abs{
  \langle O_b,\mathsf Q_\omega O_a\rangle
  }^2.
  \label{eq:exact-return-spectral-components}
\end{equation}
\end{theorem}

\begin{proof}
By Definition~\ref{def:manybody-return},
\[
  R_S(t)
  =
  \norm{\mathsf P_S\mathcal U_t\mathsf P_S}_{\mathrm F}^2.
\]
For a finite-dimensional operator \(A\),
\(\norm A_{\mathrm F}^2=\Tr_{\calK}(A^\dagger A)\), where the trace is
over operator Hilbert space.  Hence
\begin{align}
  R_S(t)
  &=
  \Tr_{\calK}\left[
  \mathsf P_S
  \mathcal U_t^\dagger
  \mathsf P_S
  \mathcal U_t
  \mathsf P_S
  \right]
  \nonumber\\
  &=
  \sum_{\omega,\nu}
  e^{it(\omega-\nu)}
  \Tr_{\calK}\left[
  \mathsf P_S\mathsf Q_\omega
  \mathsf P_S\mathsf Q_\nu
  \mathsf P_S
  \right],
  \label{eq:return-before-Cesaro}
\end{align}
where Eq.~\eqref{eq:Liouvillian-time-evolution} was used.

Applying Lemma~\ref{lem:Cesaro-orthogonality} term by term leaves only
\(\omega=\nu\):
\begin{align}
  \av{R_S(t)}
  &=
  \sum_\omega
  \Tr_{\calK}\left[
  \mathsf P_S\mathsf Q_\omega
  \mathsf P_S\mathsf Q_\omega
  \mathsf P_S
  \right]
  \nonumber\\
  &=
  \sum_\omega
  \Tr_{\calK}\left[
  (\mathsf P_S\mathsf Q_\omega\mathsf P_S)^2
  \right].
\end{align}
The operator
\(\mathsf P_S\mathsf Q_\omega\mathsf P_S\) is self-adjoint, so its squared
Frobenius norm equals the last trace.  This proves
Eq.~\eqref{eq:exact-return-spectral}.

To obtain Eq.~\eqref{eq:exact-return-spectral-components}, use the
orthonormal basis \(\{O_a\}\) of \(\calV_S\).  For each \(\omega\),
\[
  \norm{
  \mathsf P_S\mathsf Q_\omega\mathsf P_S
  }_{\mathrm F}^2
  =
  \sum_{a,b}
  \abs{
  \langle O_b,\mathsf Q_\omega O_a\rangle
  }^2.
\]
Summing over \(\omega\) proves the result.
\end{proof}

\section{Finite-window bounds and the autocorrelation characterisation}\label{sec:supp2}
\subsection{A positive majorant of an observation window}
\begin{lemma}[Tent domination]
\label{lem:tent-domination}
Fix \(T>0\), \(t_0\in\mathbb R\), and put
\[
  c=t_0+\frac T2,
  \qquad
  w_T(t)
  =
  \left(
  1-\frac{|t-c|}{T}
  \right)_+.
\]
For every nonnegative measurable \(f\),
\begin{equation}
  \frac1T\int_{t_0}^{t_0+T}f(t)\,dt
  \le
  \frac2T\int_{\mathbb R}w_T(t)f(t)\,dt.
  \label{eq:tent-domination}
\end{equation}
Moreover, for every real \(\Delta\),
\begin{equation}
  \frac1T
  \int_{\mathbb R}
  w_T(t)e^{it\Delta}\,dt
  =
  e^{ic\Delta}
  \operatorname{sinc}^2
  \left(
  \frac{T\Delta}{2}
  \right),
  \label{eq:tent-fourier}
\end{equation}
where \(\operatorname{sinc}(0)=1\) and
\(\operatorname{sinc}(x)=\sin x/x\) for \(x\ne0\).
\end{lemma}

\begin{proof}
For \(t\in[t_0,t_0+T]\),
\[
  |t-c|\le\frac T2,
\]
so \(w_T(t)\ge1/2\).  Hence
\[
  \mathbbm 1_{[t_0,t_0+T]}(t)
  \le2w_T(t),
\]
which proves Eq.~\eqref{eq:tent-domination} after multiplication by
\(f(t)\ge0\) and integration.

For the Fourier transform, substitute \(s=t-c\).  Then
\begin{align}
  &\frac1T
  \int_{\mathbb R}
  w_T(t)e^{it\Delta}\,dt
  \nonumber\\
  &\quad=
  \frac{e^{ic\Delta}}{T}
  \int_{-T}^{T}
  \left(
  1-\frac{|s|}{T}
  \right)e^{is\Delta}\,ds
  \nonumber\\
  &\quad=
  \frac{2e^{ic\Delta}}{T}
  \int_0^T
  \left(
  1-\frac{s}{T}
  \right)\cos(s\Delta)\,ds.
  \label{eq:tent-transform-step}
\end{align}
If \(\Delta\ne0\), direct integration gives
\[
  2\int_0^T
  \left(
  1-\frac{s}{T}
  \right)\cos(s\Delta)\,ds
  =
  T\left[
  \frac{\sin(T\Delta/2)}{T\Delta/2}
  \right]^2.
\]
The identity extends continuously to \(\Delta=0\), proving
Eq.~\eqref{eq:tent-fourier}.
\end{proof}
For every real \(x\), with \(\operatorname{sinc}(0)=1\),
\begin{equation}
 \operatorname{sinc}^2x\le\min\{1,x^{-2}\}.
 \label{eq:sinc-square-bound}
\end{equation}
At \(x=0\) the second bound is omitted.  For \(x\ne0\), the two
bounds follow from \(|\sin x|\le |x|\) and \(|\sin x|\le1\),
respectively.

\subsection{Hamiltonian-only local Liouvillian measure}

Let
\[
  q=d^2-1.
\]
For each distinct Liouvillian frequency define the positive operator on
\(\calV_S\)
\begin{equation}
  A_\omega
  :=
  \mathsf P_S\mathsf Q_\omega\mathsf P_S,
  \qquad
  \pi_\omega
  :=
  \frac1q\Tr_{\calK}A_\omega.
  \label{eq:local-A-pi-new}
\end{equation}
For a Borel set \(B\subset\mathbb R\), define
\[
  \mathsf Q_B
  :=
  \sum_{\omega\in B}\mathsf Q_\omega
\]
and
\begin{equation}
  \mu_S(B)
  :=
  \frac1q
  \Tr_{\calK}
  \left[
  \mathsf P_S\mathsf Q_B\mathsf P_S
  \right]
  =
  \sum_{\omega\in B}\pi_\omega.
  \label{eq:local-mu-new}
\end{equation}
Define
\begin{equation}
  \chi_S(\varepsilon;H)
  :=
  \sup_{x\in\mathbb R}
  \mu_S([x,x+\varepsilon)).
  \label{eq:local-chi-new}
\end{equation}

\begin{lemma}[Local spectral weights]
\label{lem:local-spectral-weights-new}
The probability measure \(\mu_S\) is supported in
\([-2\norm H_\infty,2\norm H_\infty]\).  Moreover, for all
Liouvillian frequencies \(\omega,\nu\),
\begin{equation}
  0
  \le
  \Tr_{\calK}(A_\omega A_\nu)
  \le
  q^2\pi_\omega\pi_\nu.
  \label{eq:local-cross-new}
\end{equation}
\end{lemma}

\begin{proof}
For \(X\in\calV_S\),
\[
  \langle X,A_\omega X\rangle
  =
  \langle X,\mathsf Q_\omega X\rangle
  =
  \norm{\mathsf Q_\omega X}_{2,D}^2
  \ge0,
\]
so \(A_\omega\ge0\) and \(\pi_\omega\ge0\).  Since
\[
  \sum_\omega A_\omega=\mathsf P_S
\]
and \(\Tr_{\calK}\mathsf P_S=q\),
\[
  \sum_\omega\pi_\omega=1.
\]
Every Liouvillian frequency is an energy difference and therefore has
absolute value at most \(2\norm H_\infty\).

For positive operators \(A,B\), one has
\[
  A\le(\Tr A)\Id
\]
because every eigenvalue of \(A\) is nonnegative and no eigenvalue exceeds
their sum.  Thus
\[
  0\le\Tr(AB)\le(\Tr A)(\Tr B).
\]
Apply this to \(A_\omega,A_\nu\) and use
\(\Tr A_\omega=q\pi_\omega\).
\end{proof}

\begin{remark}[Relation to finite-time energy-band thermalisation]
The positive time weight used below is not a new ingredient.
Positive-Fourier weights are central to the finite-time autocorrelation
framework of Ref.~\cite{Vikram2026} and to the general inner-product
formulation in Appendix~A.1 and Theorem~A.1 of
Ref.~\cite{VikramNoAverages2026}; Appendix~C.2 of the latter gives the
proof. The tent function is one of the examples used there.

The theorem below applies this Fourier positivity to a different
spectral object: the probability measure obtained by compressing the
Liouvillian spectral projectors to the complete local operator algebra.
It keeps the full shell dependence of the sinc-squared kernel instead
of replacing all frequencies outside one band by a single worst-tail
parameter.
\end{remark}

\begin{theorem}[Local Liouvillian concentration theorem]
\label{thm:fejer-local-band}
For every Hamiltonian \(H\) and every
\(t_0,T,\varepsilon>0\),
\begin{align}
  \frac1T
  \int_{t_0}^{t_0+T}
  R_S(t)\,dt
 &\le
  4q^2\chi_S(\varepsilon;H)
  \left[
  1+
  \frac{2\pi^2}{3(T\varepsilon)^2}
  \right],
  \label{eq:fejer-local-return}\\
  \frac1T
  \int_{t_0}^{t_0+T}
  \sup_{\rho_S}
  \Dt\left(
  \Phi_t(\rho_S),\frac{\Id_S}{d}
  \right)^2dt
 &\le
  (d-1)q^2\chi_S(\varepsilon;H)
  \left[
  1+
  \frac{2\pi^2}{3(T\varepsilon)^2}
  \right].
  \label{eq:fejer-local-trace}
\end{align}
No nondegeneracy assumption on energies or gaps is required.
\end{theorem}

\begin{proof}
The exact Liouvillian expansion gives
\[
  \mathsf K_S(t)
  =
  \sum_\omega e^{-it\omega}A_\omega
\]
and hence
\begin{equation}
  R_S(t)
  =
  \sum_{\omega,\nu}
  e^{it(\omega-\nu)}
  \Tr_{\calK}(A_\omega A_\nu).
  \label{eq:local-return-spectrum-new}
\end{equation}
By Lemma~\ref{lem:tent-domination} and nonnegativity of \(R_S(t)\),
the rectangular average is no larger than twice its tent-weighted
average.  Using Eq.~\eqref{eq:tent-fourier},
Eq.~\eqref{eq:local-return-spectrum-new}, and
Eq.~\eqref{eq:local-cross-new} gives
\begin{align}
  \frac1T
  \int_{t_0}^{t_0+T}R_S(t)dt
  &\le
  2q^2
  \sum_{\omega,\nu}
  \pi_\omega\pi_\nu
  \operatorname{sinc}^2
  \left(
  \frac{T(\omega-\nu)}2
  \right).
  \label{eq:local-fejer-double}
\end{align}
Fix \(\omega\).  The region
\[
  |\omega-\nu|<\varepsilon
\]
is covered by two intervals of width \(\varepsilon\), so its
\(\mu_S\)-mass is at most \(2\chi_S(\varepsilon;H)\).  For each
integer \(\ell\ge1\), the shell
\[
  \ell\varepsilon
  \le
  |\omega-\nu|
  <
  (\ell+1)\varepsilon
\]
has mass at most \(2\chi_S(\varepsilon;H)\), while
Eq.~\eqref{eq:sinc-square-bound} gives
\[
  \operatorname{sinc}^2
  \left(
  \frac{T(\omega-\nu)}2
  \right)
  \le
  \frac4{T^2\ell^2\varepsilon^2}.
\]
Therefore
\begin{align}
  \sum_\nu
  \pi_\nu
  \operatorname{sinc}^2
  \left(
  \frac{T(\omega-\nu)}2
  \right)
  &\le
  2\chi_S(\varepsilon;H)
  \left[
  1+
  \frac{2\pi^2}{3(T\varepsilon)^2}
  \right].
  \label{eq:local-fejer-fixedomega}
\end{align}
Average over \(\omega\) with total weight one and substitute into
Eq.~\eqref{eq:local-fejer-double}.  This proves
Eq.~\eqref{eq:fejer-local-return}.  The pointwise estimate
\eqref{eq:trace-distance-by-return-probability} is uniform in
\(\rho_S\).  Taking the input supremum before the time integral therefore
gives Eq.~\eqref{eq:fejer-local-trace}.
\end{proof}

\begin{corollary}[Same-scale local concentration bound]
\label{cor:no-time-resolution-overhead}
For every \(T>0\),
\begin{equation}
  \sup_{t_0\ge0}
  \frac1T
  \int_{t_0}^{t_0+T}
  \sup_{\rho_S}
  \Dt\left(
  \Phi_t(\rho_S),\frac{\Id_S}{d}
  \right)^2dt
  \le
  C_d\,
  \chi_S(T^{-1};H),
  \label{eq:no-time-resolution-overhead}
\end{equation}
where
\[
  C_d
  =
  (d-1)(d^2-1)^2
  \left(
  1+\frac{2\pi^2}{3}
  \right).
\]
Consequently, if a fixed subsystem satisfies
\[
  \chi_S(N^{-c};H_N)=O(N^{-\kappa}),
\]
then windows \(T_N=N^c\) already give
\[
  \sup_{t_0\ge0}
  \frac1{T_N}\int_{t_0}^{t_0+T_N}
  \sup_{\rho_S}
  \Dt\left(
  \Phi_t(\rho_S),\frac{\Id_S}{d}
  \right)^2dt
  =
  O_S(N^{-\kappa}).
\]
\end{corollary}

\begin{proof}
Set \(\varepsilon=T^{-1}\) in
Theorem~\ref{thm:fejer-local-band}.  Then \(T\varepsilon=1\).
\end{proof}

\begin{remark}[Why the same scale is available here]
The absence of an additional time-scale separation follows from the
shell-by-shell local concentration estimate, not from the Fej\'er
kernel alone.  In the forecasting specialisation of Theorem~A.1 of
Ref.~\cite{VikramNoAverages2026}, a rectangular output window of
length \(T\) and an input tent of width \(T_{\rm obs}\) give
\(w_0=T_{\rm obs}/(\xi T)\).  Here each successive local Liouvillian
shell is separately bounded by \(\chi_S(\varepsilon)\), and the
sinc-squared tail is summable.  This permits the direct choice
\(\varepsilon=T^{-1}\).
\end{remark}

\subsection{Exact separation of stationary memory and dynamical return}

Define
\begin{equation}
  \pi_{0,S}
  :=
  \mu_S(\{0\}),
  \qquad
  \widetilde\mu_S
  :=
  \mu_S-\pi_{0,S}\delta_0,
  \qquad
  m_S
  :=
  1-\pi_{0,S}.
  \label{eq:mu-split-new}
\end{equation}
For \(\varepsilon>0\), define the nonzero-frequency concentration
\begin{equation}
  \widetilde\chi_S(\varepsilon;H)
  :=
  \sup_x
  \widetilde\mu_S([x,x+\varepsilon)).
  \label{eq:tilde-chi-new}
\end{equation}
Let
\[
  A_0
  =
  \mathsf P_S\mathsf Q_0\mathsf P_S
\]
and
\begin{equation}
  \widetilde{\mathsf K}_S(t)
  :=
  \sum_{\omega\ne0}e^{-it\omega}A_\omega.
  \label{eq:K-tilde-new}
\end{equation}
Then
\[
  \mathsf K_S(t)
  =
  A_0+\widetilde{\mathsf K}_S(t).
\]

Let
\begin{equation}
  \overline{\Phi}(X)
  :=
  \lim_{\mathcal T\to\infty}
  \frac1{\mathcal T}
  \int_0^{\mathcal T}\Phi_t(X)\,dt.
  \label{eq:temporal-channel-new}
\end{equation}
The limit exists because all matrix elements are finite sums of
exponentials.

\begin{lemma}[Transfer block of the temporal channel]
\label{lem:temporal-transfer}
On the traceless local operator space, the transfer matrix of
\(\overline{\Phi}\) is \(A_0\).  The transfer matrix of
\(\Phi_t-\overline{\Phi}\) is
\(\widetilde{\mathsf K}_S(t)\).
\end{lemma}

\begin{proof}
In the orthonormal basis \(\{O_a\}\) of \(\calV_S\), the transfer matrix
of \(\Phi_t\) is the matrix of
\[
  \mathsf K_S(t)
  =
  \sum_\omega e^{-it\omega}A_\omega
\]
by Proposition~\ref{prop:return-equals-C}.  Ces\`aro orthogonality,
Lemma~\ref{lem:Cesaro-orthogonality}, removes every term with
\(\omega\ne0\), leaving \(A_0\).  Subtracting the temporal transfer
matrix from the time-\(t\) transfer matrix leaves
\(\widetilde{\mathsf K}_S(t)\).
\end{proof}

\begin{theorem}[Separated stationary-memory theorem]
\label{thm:separated-stationary}
For every local input state,
\begin{equation}
  \Dt\left(
  \overline{\Phi}(\rho_S),\frac{\Id_S}{d}
  \right)^2
  \le
  \frac{d-1}{4}q^2\pi_{0,S}^2.
  \label{eq:separated-stationary}
\end{equation}
Thus exact conserved local memory enters quadratically in the
mean-square trace-distance error.
\end{theorem}

\begin{proof}
Repeat the proof of
Theorem~\ref{thm:return-controls-thermalisation}, replacing the transfer
matrix \(C(t)\) by the temporal transfer matrix \(A_0\) from
Lemma~\ref{lem:temporal-transfer}.  This gives
\[
  \Dt\left(
  \overline{\Phi}(\rho_S),\frac{\Id_S}{d}
  \right)^2
  \le
  \frac{d-1}{4}\norm{A_0}_{\mathrm F}^2.
\]
Since \(A_0\ge0\),
\[
  \norm{A_0}_{\mathrm F}^2
  =
  \sum_j\lambda_j(A_0)^2
  \le
  \left(
  \sum_j\lambda_j(A_0)
  \right)^2
  =
  (\Tr_{\calK}A_0)^2
  =
  q^2\pi_{0,S}^2.
\]
Substitution proves the theorem.
\end{proof}

\begin{theorem}[Separated finite-window theorem]
\label{thm:separated-window}
For every \(t_0,T,\varepsilon>0\),
\begin{align}
  &\frac1T
  \int_{t_0}^{t_0+T}
  \sup_{\rho_S}
  \Dt\left(
  \Phi_t(\rho_S),
  \overline{\Phi}(\rho_S)
  \right)^2dt
  \nonumber\\
  &\quad\le
  (d-1)q^2
  m_S
  \widetilde\chi_S(\varepsilon;H)
  \left[
  1+
  \frac{2\pi^2}{3(T\varepsilon)^2}
  \right],
  \label{eq:separated-fluctuation}
\end{align}
and
\begin{align}
  &\frac1T
  \int_{t_0}^{t_0+T}
  \sup_{\rho_S}
  \Dt\left(
  \Phi_t(\rho_S),\frac{\Id_S}{d}
  \right)^2dt
  \nonumber\\
  &\quad\le
  (d-1)q^2
  \left[
  \sqrt{
  m_S\widetilde\chi_S(\varepsilon;H)
  \left(
  1+\frac{2\pi^2}{3(T\varepsilon)^2}
  \right)}
  +
  \frac{\pi_{0,S}}2
  \right]^2.
  \label{eq:separated-total}
\end{align}
\end{theorem}

\begin{proof}
By Lemma~\ref{lem:temporal-transfer}, the traceless transfer block of
\(\Phi_t-\overline{\Phi}\) is
\(\widetilde{\mathsf K}_S(t)\).  The norm conversion in the proof of
Theorem~\ref{thm:return-controls-thermalisation} therefore gives
\begin{equation}
  \Dt\left(
  \Phi_t(\rho_S),
  \overline{\Phi}(\rho_S)
  \right)^2
  \le
  \frac{d-1}{4}
  \norm{
  \widetilde{\mathsf K}_S(t)
  }_{\mathrm F}^2.
  \label{eq:fluctuation-pointwise}
\end{equation}
Apply Lemma~\ref{lem:tent-domination} to the nonnegative right-hand side.
The same spectral calculation as in
Theorem~\ref{thm:fejer-local-band}, but summing only over
\(\omega,\nu\ne0\), yields
\begin{align}
  &\frac1T
  \int_{t_0}^{t_0+T}
  \norm{
  \widetilde{\mathsf K}_S(t)
  }_{\mathrm F}^2dt
  \nonumber\\
  &\quad\le
  4q^2m_S\widetilde\chi_S(\varepsilon;H)
  \left[
  1+
  \frac{2\pi^2}{3(T\varepsilon)^2}
  \right].
  \label{eq:tilde-K-window}
\end{align}
Indeed, for fixed nonzero \(\omega\), every central region or shell has
\(\widetilde\mu_S\)-mass at most
\(2\widetilde\chi_S(\varepsilon;H)\), and integration over the first
frequency has total mass \(m_S\).  The pointwise estimate
\eqref{eq:fluctuation-pointwise} is uniform in \(\rho_S\).  Taking the
input supremum before integrating, and then using
Eq.~\eqref{eq:tilde-K-window}, proves
Eq.~\eqref{eq:separated-fluctuation}.

For the total error, the trace-distance triangle inequality and then
the input supremum give
\[
  \sup_{\rho_S}
  \Dt\left(
  \Phi_t(\rho_S),\frac{\Id_S}{d}
  \right)
  \le
  \sup_{\rho_S}
  \Dt\left(
  \Phi_t(\rho_S),
  \overline{\Phi}(\rho_S)
  \right)
  +s_0,
\]
where
\[
  s_0=\sup_{\rho_S}
  \Dt\left(
  \overline\Phi(\rho_S),\frac{\Id_S}{d}
  \right).
\]
Minkowski's inequality for the normalised \(L^2\) norm on
\([t_0,t_0+T]\) therefore gives
\begin{align}
 &\left[
 \frac1T\int_{t_0}^{t_0+T}
 \sup_{\rho_S}
 \Dt\left(
 \Phi_t(\rho_S),\frac{\Id_S}{d}
 \right)^2dt
 \right]^{1/2}
 \nonumber\\
 &\quad\le
 \left[
 \frac1T\int_{t_0}^{t_0+T}
 \sup_{\rho_S}
 \Dt\left(
 \Phi_t(\rho_S),\overline\Phi(\rho_S)
 \right)^2dt
 \right]^{1/2}
 +s_0.
\end{align}
Apply Eq.~\eqref{eq:separated-fluctuation} to the first term and
Theorem~\ref{thm:separated-stationary} to \(s_0\), then square. This
proves Eq.~\eqref{eq:separated-total}.
\end{proof}

\begin{corollary}[Same-scale separated form]
\label{cor:separated-same-scale}
Setting \(\varepsilon=T^{-1}\) gives
\begin{equation}
  \sup_{t_0\ge0}
  \frac1T
  \int_{t_0}^{t_0+T}
  \sup_{\rho_S}
  \Dt\left(
  \Phi_t(\rho_S),\frac{\Id_S}{d}
  \right)^2dt
  \le
  (d-1)q^2
  \left[
  \sqrt{
  \left(1+\frac{2\pi^2}{3}\right)
  \widetilde\chi_S(T^{-1};H)}
  +
  \frac{\pi_{0,S}}2
  \right]^2.
  \label{eq:separated-same-scale}
\end{equation}
\end{corollary}

\begin{proof}
Use \(m_S\le1\) and \(T\varepsilon=1\) in
Eq.~\eqref{eq:separated-total}.
\end{proof}

\subsection{Same-scale autocorrelation--band equivalence}
\label{sec:same-scale-equivalence}

The local spectral criterion has an exact finite-time dynamical
interpretation.  Define
\begin{equation}
  \varphi_S(t)
  :=
  \int e^{-it\omega}\,\mu_S(d\omega)
  =
  \frac1q
  \Tr_{\calK}
  \left[
  \mathsf P_Se^{-it\calL}\mathsf P_S
  \right].
  \label{eq:phi-local-new}
\end{equation}
If \(O_1,\ldots,O_q\) is any orthonormal basis of \(\calV_S\), then
\begin{equation}
  \varphi_S(t)
  =
  \frac1q
  \sum_{a=1}^q
  \langle O_a,e^{-it\calL}O_a\rangle.
  \label{eq:phi-basis-new}
\end{equation}
Thus \(\varphi_S\) is the mean infinite-temperature autocorrelation of the
complete traceless local operator algebra.  Remove the exact stationary
atom:
\begin{equation}
  \widetilde\varphi_S(t)
  :=
  \varphi_S(t)-\pi_{0,S}
  =
  \int e^{-it\omega}\,\widetilde\mu_S(d\omega).
  \label{eq:phi-tilde-new}
\end{equation}
Define the rectangular mean square
\begin{equation}
  A_S(T)
  :=
  \frac1{2T}
  \int_{-T}^{T}
  |\widetilde\varphi_S(t)|^2dt
  \label{eq:AS-new}
\end{equation}
and the Fej\'er-weighted mean square
\begin{equation}
  A_S^{\rm F}(T)
  :=
  \frac1T
  \int_{-T}^{T}
  \left(
  1-\frac{|t|}{T}
  \right)
  |\widetilde\varphi_S(t)|^2dt.
  \label{eq:AF-new}
\end{equation}

\begin{lemma}[Autocorrelation representation]
\label{lem:phi-autocorr-new}
Equations~\eqref{eq:phi-local-new} and
\eqref{eq:phi-basis-new} hold, and
\begin{equation}
  A_S^{\rm F}(T)
  =
  \iint
  \operatorname{sinc}^2
  \left(
  \frac{T(\omega-\nu)}2
  \right)
  \widetilde\mu_S(d\omega)
  \widetilde\mu_S(d\nu).
  \label{eq:AF-convolution-new}
\end{equation}
\end{lemma}

\begin{proof}
The spectral theorem gives
\[
  e^{-it\calL}
  =
  \sum_\omega e^{-it\omega}\mathsf Q_\omega.
\]
Taking the operator-space trace after compression by \(\mathsf P_S\)
gives Eq.~\eqref{eq:phi-local-new}; evaluating that trace in the
orthonormal basis \(\{O_a\}\) gives
Eq.~\eqref{eq:phi-basis-new}.

Next,
\[
  |\widetilde\varphi_S(t)|^2
  =
  \iint
  e^{-it(\omega-\nu)}
  \widetilde\mu_S(d\omega)
  \widetilde\mu_S(d\nu).
\]
The measure is finite and the integrand is bounded, so Fubini's theorem
applies.  Equation~\eqref{eq:tent-fourier} with centre zero gives
Eq.~\eqref{eq:AF-convolution-new}.
\end{proof}

\begin{theorem}[Same-scale two-sided equivalence]
\label{thm:same-scale-equivalence-new}
Let
\[
  c_*
  =
  \frac{\sin(1/2)}{1/2}.
\]
For every \(T>0\),
\begin{equation}
  \frac{c_*^2}{2}
  \widetilde\chi_S(T^{-1};H)^2
  \le
  A_S(T)
  \le
  4m_S
  \left(
  1+\frac{2\pi^2}{3}
  \right)
  \widetilde\chi_S(T^{-1};H).
  \label{eq:same-scale-equivalence-new}
\end{equation}
There is no logarithmic loss and no change of time scale.
\end{theorem}

\begin{proof}
Let \(I\) be an interval of length \(1/T\).  If
\(\omega,\nu\in I\), then
\[
  \left|
  \frac{T(\omega-\nu)}2
  \right|
  \le\frac12.
\]
The function \(\operatorname{sinc}(x)\) is even and decreasing on
\([0,1/2]\), so
\[
  \operatorname{sinc}^2
  \left(
  \frac{T(\omega-\nu)}2
  \right)
  \ge c_*^2.
\]
Equation~\eqref{eq:AF-convolution-new} therefore gives
\[
  A_S^{\rm F}(T)
  \ge
  c_*^2\,\widetilde\mu_S(I)^2.
\]
Taking the supremum over \(I\),
\begin{equation}
  A_S^{\rm F}(T)
  \ge
  c_*^2
  \widetilde\chi_S(T^{-1};H)^2.
  \label{eq:AF-lower-new}
\end{equation}
Since the triangular weight is at most one,
\[
  A_S^{\rm F}(T)
  \le
  \frac1T
  \int_{-T}^{T}
  |\widetilde\varphi_S(t)|^2dt
  =
  2A_S(T).
\]
This proves the lower inequality.

For the upper inequality, fix \(\omega\).  The same shell decomposition as
in Theorem~\ref{thm:fejer-local-band}, now for
\(\widetilde\mu_S\) and with \(\varepsilon=T^{-1}\), gives
\begin{align}
  &\int
  \operatorname{sinc}^2
  \left(
  \frac{T(\omega-\nu)}2
  \right)
  \widetilde\mu_S(d\nu)
  \nonumber\\
  &\qquad\le
  2\widetilde\chi_S(T^{-1};H)
  \left(
  1+\frac{2\pi^2}{3}
  \right).
  \label{eq:AF-upper-fixed-new}
\end{align}
Integrating over \(\omega\), whose total
\(\widetilde\mu_S\)-mass is \(m_S\), gives
\begin{equation}
  A_S^{\rm F}(T)
  \le
  2m_S
  \left(
  1+\frac{2\pi^2}{3}
  \right)
  \widetilde\chi_S(T^{-1};H).
  \label{eq:AF-upper-new}
\end{equation}
Apply this at time \(2T\).  For \(|t|\le T\),
\[
  1-\frac{|t|}{2T}\ge\frac12,
\]
so
\[
  A_S^{\rm F}(2T)
  \ge
  \frac1{4T}
  \int_{-T}^{T}
  |\widetilde\varphi_S(t)|^2dt
  =
  \frac12A_S(T).
\]
Consequently
\[
  A_S(T)
  \le
  2A_S^{\rm F}(2T)
  \le
  4m_S
  \left(
  1+\frac{2\pi^2}{3}
  \right)
  \widetilde\chi_S((2T)^{-1};H).
\]
The concentration function is nondecreasing in interval width, hence
\[
  \widetilde\chi_S((2T)^{-1};H)
  \le
  \widetilde\chi_S(T^{-1};H).
\]
This proves the upper inequality.
\end{proof}

\begin{remark}[Two different square-root conversions]
Equation~\eqref{eq:same-scale-equivalence-new} gives
\(\widetilde\chi_S(T^{-1};H)\le\sqrt{2A_S(T)}/c_*\).
This square root is a measure-theoretic cost of reconstructing a
\emph{maximum frequency-cell mass} from a mean-square characteristic
function.  A linear bound of a maximum cell mass by this mean square is
false for general positive measures.  For example, fix \(0<p<1/2\),
put \(M=\lceil p^{-2}\rceil\), and use
\[
 \mu_p=p\delta_{2\pi}+\frac{1-p}{M}
                   \sum_{j=2}^{M+1}\delta_{2\pi j}.
\]
All atoms are nonzero, every unit-width interval contains at most one,
and \((1-p)/M\le p^2<p\); thus its unit-width concentration equals
\(p\).  At \(T=1\), exact integration of the distinct integer-frequency
exponentials gives
\[
 \frac12\int_{-1}^{1}|\widehat\mu_p(t)|^2dt
 =p^2+\frac{(1-p)^2}{M}\le2p^2.
\]
The ratio of concentration to mean square is therefore at least
\(1/(2p)\), which is unbounded as \(p\downarrow0\).

This loss is different from the Cauchy--Schwarz square root in the
generic autocorrelation-to-cross-correlator bound of
Ref.~\cite{VikramNoAverages2026}; see Theorem~A.1 and its proof in
Appendix~C.2. The first arises when a scalar correlation is converted
into a maximum spectral-cell mass. The second arises when an
autocorrelator controls a different operator direction.  The direct
complete-local theorem below avoids both intermediate conversions by
keeping the positive operator-valued collision kernel and bounding the
physical local return directly.
\end{remark}

\subsection{Direct control by the local autocorrelation}
\label{sec:direct-positive-kernel}

The general principle that finite-time autocorrelation data can
constrain later thermalisation through positive Fourier weights is
developed in Ref.~\cite{Vikram2026} and formulated abstractly in
Theorem~A.1 of Ref.~\cite{VikramNoAverages2026}.  In that setting the
target is a generic cross-correlator, and Cauchy--Schwarz compares the
target operator direction with the spectral components of the
autocorrelated operator.

For the subsystem problem considered here, the complete local operator algebra defines positive operator-valued
spectral blocks \(A_\omega=\mathsf P_S\mathsf Q_\omega\mathsf P_S\),
while the physical quantity controlling the reduced-state error is the
quadratic return probability of that same local subspace.  This permits
a direct comparison of two quadratic quantities.

Besides \(A_S(T)\) and \(A_S^{\rm F}(T)\), which exclude the zero mode,
define their full-measure versions
\begin{align}
 \mathcal A_S(T;H)
 &=\frac1{2T}\int_{-T}^{T}|\varphi_S(t)|^2dt,\\
 \mathcal F_S(T;H)
 &=\frac1T\int_{-T}^{T}
       \left(1-\frac{|t|}{T}\right)|\varphi_S(t)|^2dt.
 \label{eq:full-correlation-functionals}
\end{align}
Set
\[
 \mathfrak c_d=\frac{(d-1)q^2}{2}.
\]
Define the pointwise worst-input errors and their worst-window averages by
\begin{align}
 e_H(t)&=\sup_{\rho_S}
     \Dt\!\left(\Phi_t(\rho_S),\Id_S/d\right)^2,\\
 f_H(t)&=\sup_{\rho_S}
     \Dt\!\left(\Phi_t(\rho_S),\overline\Phi(\rho_S)\right)^2,\\
 \mathscr E_H(T)&=\sup_{t_0\ge0}
       \frac1T\int_{t_0}^{t_0+T}e_H(t)\,dt,\\
 \mathscr G_H(T)&=\sup_{t_0\ge0}
       \frac1T\int_{t_0}^{t_0+T}f_H(t)\,dt.
 \label{eq:worst-input-window-functionals}
\end{align}
The input supremum is inside the time integral.

\begin{lemma}[Measurability of the every-window functionals]
\label{lem:window-measurability}
For fixed \(T>0\), \(\mathscr E_H(T)\) and \(\mathscr G_H(T)\) are Borel
functions of a finite-dimensional Hermitian matrix \(H\).  They take
values in \([0,1]\).
\end{lemma}
\begin{proof}
The exponential \(e^{-itH}\) is continuous in \(t\) and the matrix
entries of \(H\).  The density-operator set is compact.  Consequently
the supremum defining \(e_H(t)\) is continuous in \((H,t)\), and its
average over a finite window is continuous in \((H,t_0)\).  The
supremum over \(t_0\ge0\) equals the supremum over nonnegative rational
\(t_0\), by continuity in \(t_0\); a countable supremum of continuous
functions is Borel.

Every matrix entry of \(\overline\Phi\) is a pointwise limit of
finite-time averages of continuous functions of \(H\), hence is Borel.
For fixed \(H\), both the input and time dependence of the expression
defining \(f_H(t)\) are continuous.  Its input supremum can be taken
over a countable dense subset of density operators.  The resulting
function is Borel in \((H,t)\), and the bounded time integral is Borel
in \((H,t_0)\).  Taking the supremum over rational starting times proves
the second assertion.  Trace distance between states is at most one.
\end{proof}

\begin{definition}[Positive-Fourier observation kernel]
\label{def:positive-fourier-kernel}
A function \(w\in L^1(\mathbb R)\) is called an admissible
positive-Fourier observation kernel if
\[
 w(t)\ge0,\qquad
 \int_{\mathbb R}w(t)dt=1,\qquad
 \widehat w(x):=\int_{\mathbb R}w(t)e^{itx}dt\in[0,\infty)
\]
for every real \(x\).
\end{definition}

\begin{lemma}[Positive-block subspace comparison]
\label{lem:positive-block-subspace}
Let \(A_\alpha\ge0\) be a finite family of operators on an
\(r\)-dimensional Hilbert space, let \(\lambda_\alpha\in\mathbb R\), and
put
\[
 K(t)=\sum_\alpha e^{-it\lambda_\alpha}A_\alpha,\qquad
 a(t)=\frac1r\Tr K(t),\qquad
 p(t)=\frac1r\|K(t)\|_{\rm F}^2.
\]
For every admissible positive-Fourier kernel \(w\),
\begin{align}
 |a(t)|^2&\le p(t)\qquad\text{for every }t,
 \label{eq:abstract-pointwise}\\
 \int_{\mathbb R}w(t-c)p(t)dt
 &\le
 r\int_{\mathbb R}w(t)|a(t)|^2dt
 \qquad\text{for every }c\in\mathbb R.
 \label{eq:abstract-translated}
\end{align}
No normalisation assumption on \(\sum_\alpha A_\alpha\) is required.
\end{lemma}

\begin{proof}
Hilbert--Schmidt Cauchy--Schwarz gives
\[
 |\Tr K(t)|^2
 \le
 \|\Id_r\|_{\rm F}^2\|K(t)\|_{\rm F}^2
 =
 r\|K(t)\|_{\rm F}^2,
\]
which proves Eq.~\eqref{eq:abstract-pointwise}.

For
\[
 a_{\alpha\beta}=\Tr(A_\alpha A_\beta),
\]
positivity gives
\[
 0\le a_{\alpha\beta}
 \le(\Tr A_\alpha)(\Tr A_\beta).
\]
Indeed,
\(A_\alpha\le(\Tr A_\alpha)\Id_r\), and multiplication by
\(A_\beta\ge0\) followed by the trace proves the upper bound.
Expanding,
\[
 p(t)=\frac1r\sum_{\alpha,\beta}
 e^{it(\lambda_\alpha-\lambda_\beta)}
 a_{\alpha\beta}.
\]
Therefore
\[
 \int w(t-c)p(t)dt
 =
 \frac1r\sum_{\alpha,\beta}
 e^{ic(\lambda_\alpha-\lambda_\beta)}
 \widehat w(\lambda_\alpha-\lambda_\beta)
 a_{\alpha\beta}.
\]
The left-hand side is nonnegative.  Taking the absolute value of the
finite sum and using
\(\widehat w\ge0\), \(a_{\alpha\beta}\ge0\), and the upper bound above
gives
\[
 \int w(t-c)p(t)dt
 \le
 \frac1r\sum_{\alpha,\beta}
 \widehat w(\lambda_\alpha-\lambda_\beta)
 (\Tr A_\alpha)(\Tr A_\beta).
\]
Writing \(\pi_\alpha=r^{-1}\Tr A_\alpha\), the last expression is
\[
 r\sum_{\alpha,\beta}
 \widehat w(\lambda_\alpha-\lambda_\beta)
 \pi_\alpha\pi_\beta
 =
 r\int w(t)|a(t)|^2dt.
\]
\end{proof}

\begin{lemma}[Haar mean-square subsystem identity]
\label{lem:haar-mean-square-subsystem}
Let \(\Lambda\) be a Hermiticity-preserving linear map on
\(\mathcal B(\mathcal H_S)\) such that
\[
 \Lambda(\Id_S)=0,\qquad \Tr\Lambda(X)=0
\]
for every \(X\), and let \(C\) be its transfer matrix on the traceless
basis \(P_a\).  For a Haar-random pure state
\(\psi=\ket\psi\!\bra\psi\),
\begin{equation}
 \mathbb E_\psi\|\Lambda(\psi)\|_2^2
 =
 \frac{\|C\|_{\rm F}^2}{d(d+1)}.
 \label{eq:haar-mean-square-identity}
\end{equation}
Moreover, if \(X=X^\dagger\) and \(\Tr X=0\), then
\begin{equation}
 \|X\|_1^2\ge2\|X\|_2^2.
 \label{eq:traceless-1-2}
\end{equation}
\end{lemma}

\begin{proof}
Since \(\Lambda(\Id_S)=0\),
\[
 \Lambda(\psi)
 =
 \frac1d\sum_b(Cx)_bP_b,\qquad
 x_a=\Tr(P_a\psi).
\]
The pure-state second moment
\[
 \mathbb E_\psi[\psi^{\otimes2}]
 =
 \frac{\Id+F}{d(d+1)}
\]
gives
\[
 \mathbb E_\psi[x_ax_b]
 =
 \frac{\Tr(P_aP_b)}{d(d+1)}
 =
 \frac{\delta_{ab}}{d+1},
\]
because each \(P_a\) is traceless and
\(\Tr(P_aP_b)=d\delta_{ab}\).  Hence
\[
 \mathbb E_\psi\|\Lambda(\psi)\|_2^2
 =
 \frac1d
 \mathbb E_\psi\|Cx\|_2^2
 =
 \frac{\|C\|_{\rm F}^2}{d(d+1)}.
\]

For Eq.~\eqref{eq:traceless-1-2}, write
\(X=X_+-X_-\) with orthogonal positive and negative parts.
Trace zero gives
\(\Tr X_+=\Tr X_-=:s\).  Thus
\[
 \|X\|_1=2s,
\]
while positivity implies
\[
 \|X\|_2^2
 =
 \Tr X_+^2+\Tr X_-^2
 \le
 (\Tr X_+)^2+(\Tr X_-)^2
 =
 2s^2.
\]
Therefore \(\|X\|_1^2\ge2\|X\|_2^2\).
\end{proof}

\begin{proposition}[Pointwise local return and worst-input error]
\label{prop:pointwise-two-sided-local-return}
Let \(K_*(t)\) be either the full local return block
\(\mathsf K_S(t)\) or the nonzero-frequency block
\(\widetilde{\mathsf K}_S(t)\), and put
\[
 p_*(t)=\frac1q\|K_*(t)\|_{\rm F}^2.
\]
Let \(g_*(t)\) be respectively \(e_H(t)\) or \(f_H(t)\).  Then
\begin{equation}
 \frac{d-1}{2d}\,p_*(t)
 \le
 g_*(t)
 \le
 \frac{(d-1)q}{4}\,p_*(t).
 \label{eq:pointwise-two-sided-local-return}
\end{equation}
\end{proposition}

\begin{proof}
The upper bound is the local-return estimate of
Theorem~\ref{thm:return-controls-thermalisation}, applied either to the
full transfer matrix or to the difference
\(\Phi_t-\overline\Phi\).

For the lower bound, let \(\Lambda_t\) be the corresponding
trace-annihilating map.  Its traceless transfer matrix is \(K_*(t)\).
The worst-input error is at least the Haar average over pure states.
Each \(\Lambda_t(\psi)\) is Hermitian and traceless, so by
Lemma~\ref{lem:haar-mean-square-subsystem},
\begin{align}
 g_*(t)
 &\ge
 \mathbb E_\psi\frac14\|\Lambda_t(\psi)\|_1^2\\
 &\ge
 \frac12\mathbb E_\psi\|\Lambda_t(\psi)\|_2^2\\
 &=
 \frac{\|K_*(t)\|_{\rm F}^2}{2d(d+1)}.
\end{align}
Since \(q=(d-1)(d+1)\), this equals
\[
 \frac{d-1}{2d}\,p_*(t).
\]
\end{proof}

\begin{theorem}[Complete-local equilibration for positive observation weights]
\label{thm:positive-kernel-general}
For an admissible positive-Fourier kernel \(w\), define
\begin{align}
 \mathcal C_w^{\rm full}
 &=\int_{\mathbb R}w(t)|\varphi_S(t)|^2dt,\\
 \mathcal C_w^{\rm fl}
 &=\int_{\mathbb R}w(t)|\widetilde\varphi_S(t)|^2dt,\\
 \mathcal E_{w,H}
 &=\sup_{c\in\mathbb R}
   \int_{\mathbb R}w(t-c)e_H(t)dt,\\
 \mathcal G_{w,H}
 &=\sup_{c\in\mathbb R}
   \int_{\mathbb R}w(t-c)f_H(t)dt.
\end{align}
Then
\begin{align}
 \frac{d-1}{2d}\mathcal C_w^{\rm full}
 \le
 \mathcal E_{w,H}
 \le
 \frac{(d-1)q^2}{4}\mathcal C_w^{\rm full},
 \label{eq:positive-kernel-general-full}\\
 \frac{d-1}{2d}\mathcal C_w^{\rm fl}
 \le
 \mathcal G_{w,H}
 \le
 \frac{(d-1)q^2}{4}\mathcal C_w^{\rm fl}.
 \label{eq:positive-kernel-general-fl}
\end{align}
No spectral nondegeneracy, locality, randomness, or Markovianity
assumption is required.
\end{theorem}

\begin{proof}
Apply Lemma~\ref{lem:positive-block-subspace} with \(r=q\) to either
all Liouvillian frequencies or the nonzero-frequency set.  In the full
case,
\[
 a(t)=\varphi_S(t),\qquad
 p(t)=q^{-1}\|\mathsf K_S(t)\|_{\rm F}^2;
\]
in the fluctuation case,
\[
 a(t)=\widetilde\varphi_S(t),\qquad
 p(t)=q^{-1}\|\widetilde{\mathsf K}_S(t)\|_{\rm F}^2.
\]
For the lower bounds choose \(c=0\) in the supremum, use the lower
pointwise inequality in
Proposition~\ref{prop:pointwise-two-sided-local-return}, and then
Eq.~\eqref{eq:abstract-pointwise}.  For the upper bounds use the upper
pointwise inequality and then
Eq.~\eqref{eq:abstract-translated}.  Combining the pointwise factor \((d-1)q/4\) with the
factor \(q\) in Eq.~\eqref{eq:abstract-translated} gives
\((d-1)q^2/4\), as claimed.
\end{proof}

\begin{remark}[Relation to Refs.~\cite{Vikram2026,VikramNoAverages2026}]
Positive Fourier weights are already central to the finite-time
autocorrelation results of Refs.~\cite{Vikram2026,VikramNoAverages2026}.
The step used here is specific to the complete local operator space:
the physical target is the quadratic return of this same subspace, with
coefficients
\[
 a_{\omega\nu}=\Tr(A_\omega A_\nu)
\]
satisfying
\[
 0\le a_{\omega\nu}
 \le\Tr(A_\omega)\Tr(A_\nu).
\]
This coefficientwise comparison avoids the Cauchy--Schwarz square
root that appears when an autocorrelator is used to control a different
operator direction. Positive Fourier weights themselves are not new.
\end{remark}

\begin{theorem}[Positive-kernel return theorem]
\label{thm:positive-kernel-return}
For every finite Hamiltonian, every subsystem of dimension \(d\ge2\),
and every \(T>0\),
\begin{align}
 \mathscr E_H(T)
     \le \mathfrak c_d\,\mathcal F_S(T;H)
     \le 2\mathfrak c_d\,\mathcal A_S(T;H),
 \label{eq:positive-kernel-full}\\
 \mathscr G_H(T)
     \le \mathfrak c_d\,A_S^{\rm F}(T)
     \le 2\mathfrak c_d\,A_S(T).
 \label{eq:positive-kernel-fluct}
\end{align}
The right-hand sides involve correlations only on \([-T,T]\), although
the bounds hold simultaneously in every translated window.
\end{theorem}

\begin{proof}
Use the triangular kernel
\[
 w_T(t)=\frac1T\left(1-\frac{|t|}{T}\right)_+.
\]
It is admissible and
\[
 \widehat w_T(x)=
 \operatorname{sinc}^2\!\left(\frac{Tx}{2}\right)\ge0.
\]
By Lemma~\ref{lem:tent-domination}, every rectangular window of length
\(T\) is bounded by twice a translate of \(w_T\).  Apply the upper
bounds of Theorem~\ref{thm:positive-kernel-general}.  Since
\[
 \mathfrak c_d=\frac{(d-1)q^2}{2},
\]
this gives the first inequalities above.  The triangular weight is at
most \(1/T\) on \([-T,T]\), which yields
\(\mathcal F_S\le2\mathcal A_S\) and
\(A_S^{\rm F}\le2A_S\).
\end{proof}

\begin{theorem}[Improved linear same-scale comparison]
\label{thm:linear-dynamical-equivalence}
For every \(T>0\),
\begin{align}
 \frac{d-1}{2d}\,\mathcal A_S(T;H)
 &\le \mathscr E_H(T)
 \le (d-1)q^2\,\mathcal A_S(T;H),\\
 \frac{d-1}{2d}\,A_S(T)
 &\le \mathscr G_H(T)
 \le (d-1)q^2\,A_S(T).
 \label{eq:linear-dynamical-equivalence}
\end{align}
Thus at fixed subsystem dimension the complete-local correlation mean
square and the every-window worst-input error are comparable linearly.
\end{theorem}

\begin{proof}
For the lower bounds use the window with \(t_0=0\) and
Proposition~\ref{prop:pointwise-two-sided-local-return}.  In the full
case,
\[
 \frac1T\int_0^T e_H(t)dt
 \ge
 \frac{d-1}{2d}\frac1T\int_0^T
 \frac{\|\mathsf K_S(t)\|_{\rm F}^2}{q}dt.
\]
Equation~\eqref{eq:abstract-pointwise}, applied to the full positive
blocks, gives
\[
 |\varphi_S(t)|^2
 \le
 \frac{\|\mathsf K_S(t)\|_{\rm F}^2}{q}.
\]
Since the squared correlation is even, its average on \([0,T]\) is
\(\mathcal A_S(T;H)\).  This proves the first lower bound.
The fluctuation case is identical with the zero frequency removed.

The upper bounds follow from
Theorem~\ref{thm:positive-kernel-return} and
\(\mathcal F_S\le2\mathcal A_S\),
\(A_S^{\rm F}\le2A_S\).
\end{proof}

\begin{corollary}[Reference-assisted complete-local equilibration]
\label{cor:diamond-complete-local}
Let
\[
 \mathcal D(X)=\Tr(X)\frac{\Id_S}{d},
 \qquad
 \Delta_\diamond(\Phi,\Psi)
 =\frac12\|\Phi-\Psi\|_\diamond,
\]
and define the every-window channel errors
\begin{align}
 \mathscr G_H^\diamond(T)
 &=
 \sup_{t_0\ge0}\frac1T\int_{t_0}^{t_0+T}
 \Delta_\diamond(\Phi_t,\overline\Phi)^2dt,\\
 \mathscr E_H^\diamond(T)
 &=
 \sup_{t_0\ge0}\frac1T\int_{t_0}^{t_0+T}
 \Delta_\diamond(\Phi_t,\mathcal D)^2dt.
\end{align}
Then
\begin{align}
 \frac{d-1}{2d}A_S(T)
 &\le
 \mathscr G_H^\diamond(T)
 \le
 d^2q^2A_S(T),
 \label{eq:diamond-G-supp}\\
 \frac{d-1}{2d}\mathcal A_S(T;H)
 &\le
 \mathscr E_H^\diamond(T)
 \le
 d^2q^2\mathcal A_S(T;H).
 \label{eq:diamond-E-supp}
\end{align}
\end{corollary}

\begin{proof}
For a Hermiticity-preserving linear map on a \(d\)-dimensional system,
the standard norm comparison
\[
 \|\Lambda\|_\diamond
 \le d\,\|\Lambda\|_{2\to2}
\]
is the upper inequality in Eq.~(2.2) of
Ref.~\cite{ChenDingHuang2022}; see also
Ref.~\cite{Watrous2018} for the diamond norm and its channel
distinguishability interpretation.

Take first
\(\Lambda_t=\Phi_t-\overline\Phi\).  Since both channels are
trace preserving and unital,
\[
 \Lambda_t(\Id_S)=0,\qquad
 \Tr\Lambda_t(X)=0.
\]
The Hilbert--Schmidt identity/traceless decomposition therefore shows
that
\[
 \|\Lambda_t\|_{2\to2}
 =
 \|\widetilde{\mathsf K}_S(t)\|_{\rm op}.
\]
Consequently
\[
 \Delta_\diamond(\Phi_t,\overline\Phi)^2
 \le
 \frac{d^2}{4}
 \|\widetilde{\mathsf K}_S(t)\|_{\rm op}^2
 \le
 \frac{d^2}{4}
 \|\widetilde{\mathsf K}_S(t)\|_{\rm F}^2.
\]
The positive-kernel proof of
Theorem~\ref{thm:positive-kernel-return}, before the final
trace-distance conversion, gives
\[
 \sup_{t_0\ge0}\frac1T\int_{t_0}^{t_0+T}
 \|\widetilde{\mathsf K}_S(t)\|_{\rm F}^2dt
 \le
 4q^2A_S(T).
\]
This proves the upper bound in
Eq.~\eqref{eq:diamond-G-supp}.  The lower bound follows from
\[
 f_H(t)
 \le
 \Delta_\diamond(\Phi_t,\overline\Phi)^2
\]
and Theorem~\ref{thm:linear-dynamical-equivalence}.

For the full result replace
\(\Lambda_t\) by \(\Phi_t-\mathcal D\) and
\(\widetilde{\mathsf K}_S(t)\) by \(\mathsf K_S(t)\).  The same
argument gives Eq.~\eqref{eq:diamond-E-supp}.
\end{proof}

\begin{remark}
The diamond-norm conversion and its operational meaning are standard;
the corollary above is an immediate strengthening supplied by the
complete-local return theorem.  It says that the same scalar
autocorrelation controls distinguishability even when the subsystem is
part of an arbitrarily entangled state with an external reference.
For the full thermalisation statement,
\[
 (\mathcal D\otimes\mathrm{id}_R)(\rho_{SR})
 =
 \frac{\Id_S}{d}\otimes\rho_R,
\]
so small \(\mathscr E_H^\diamond(T)\) is an every-window decoupling
statement.

This is a stronger operational corollary of the complete-local theorem,
not a different autocorrelation mechanism.
\end{remark}

\begin{remark}[What is meant by the converse]
Ref.~\cite{Vikram2026} also contains converse-type statements relating
thermalisation of a sufficiently large fraction of a projector
eigenbasis to the corresponding autocorrelator. We therefore do not
claim converse reasoning itself as new. The statement above is more
specific: the same two complete-subsystem functionals are bounded
against one another with constants depending only on the fixed
subsystem dimension.
\end{remark}

\begin{remark}[Connection to subspace aperiodicity]
For an ordinary rank-\(r\) subspace with projector \(P\) and unitary
dynamics \(U_t\), define
\[
 a_P(t)=r^{-1}\Tr(PU_tP),\qquad
 p_P(t)=r^{-1}\|PU_tP\|_{\rm F}^2.
\]
The double-time average of \(p_P(t-t')\) over \([0,T]^2\) is exactly a
Fej\'er-weighted return probability.  Following the aperiodicity
language of Ref.~\cite{VikramAperiodicity2026}, define its reciprocal by
\[
 \mathfrak A_T[P]^{-1}
 :=
 \int_{\mathbb R}w_T(t)p_P(t)dt.
\]
Lemma~\ref{lem:positive-block-subspace} therefore gives
\[
 \int w_T(t)|a_P(t)|^2dt
 \le
 \mathfrak A_T[P]^{-1}
 \le
 r\int w_T(t)|a_P(t)|^2dt.
\]  The present subsystem theorem is
the Liouville-space specialisation \(P=\mathsf P_S\), supplemented by
the operational conversion from local-operator return to worst-input
reduced-state error.  We use this as a conceptual bridge, not as a
claim that subspace recurrence or projected operator-valued spectral
measures themselves are new; see, for example,
Ref.~\cite{BourgainGrunbaumVelazquezWilkening2014}.
\end{remark}

\begin{corollary}[Thermalisation directly from the complete local autocorrelation]
\label{cor:direct-from-autocorrelation-new}
For every \(T>0\),
\begin{equation}
 \mathscr E_H(T)
 \le
 \mathfrak c_d
 \left[
 \sqrt{A_S^{\rm F}(T)}
 +\frac{\pi_{0,S}}{\sqrt2}
 \right]^2
 \le
 (d-1)q^2
 \left[
 \sqrt{A_S(T)}
 +\frac{\pi_{0,S}}2
 \right]^2.
 \label{eq:direct-from-autocorrelation-new}
\end{equation}
\end{corollary}

\begin{proof}
Put
\[
 s_0=
 \sup_{\rho_S}
 \Dt\left(
 \overline\Phi(\rho_S),\frac{\Id_S}{d}
 \right).
\]
For every \(t\),
\[
 \sqrt{e_H(t)}
 \le
 \sqrt{f_H(t)}+s_0.
\]
Applying Minkowski's inequality in every length-\(T\) window and then
taking the supremum over its starting point gives
\[
 \sqrt{\mathscr E_H(T)}
 \le
 \sqrt{\mathscr G_H(T)}+s_0.
\]
Theorem~\ref{thm:positive-kernel-return} gives
\[
 \mathscr G_H(T)\le\mathfrak c_d A_S^{\rm F}(T),
\]
while Theorem~\ref{thm:separated-stationary} gives
\[
 s_0\le\frac{q\sqrt{d-1}}2\,\pi_{0,S}.
\]
Since
\(\mathfrak c_d=(d-1)q^2/2\), these inequalities give the first
bound.  Finally use \(A_S^{\rm F}(T)\le2A_S(T)\).
\end{proof}

\begin{corollary}[Ensemble means and probability bounds]
\label{cor:linear-annealed}
For any random Hamiltonian ensemble, if
\(\mathbb E A_S(T)\le a\) and
\(\mathbb E\pi_{0,S}^2\le p_2\), then
\begin{align}
 \mathbb E\mathscr E_H(T)
 &\le
 (d-1)q^2
 \left(
 \sqrt a+\frac{\sqrt{p_2}}2
 \right)^2,
 \label{eq:annealed-Minkowski}\\
 \mathbb P[\mathscr E_H(T)>u]
 &\le
 \frac{(d-1)q^2}{u}
 \left(
 \sqrt a+\frac{\sqrt{p_2}}2
 \right)^2
 \qquad (u>0).
\end{align}
The event concerns all starting times and all subsystem inputs
simultaneously.
\end{corollary}

\begin{proof}
Take expectations in the second inequality of
Eq.~\eqref{eq:direct-from-autocorrelation-new}.  The cross term obeys
\[
 \mathbb E\!\left[\pi_{0,S}\sqrt{A_S(T)}\right]
 \le
 \sqrt{
 \mathbb E\pi_{0,S}^2\,
 \mathbb E A_S(T)}
 \le
 \sqrt{p_2a}
\]
by Cauchy--Schwarz.  This proves
Eq.~\eqref{eq:annealed-Minkowski}.  Measurability follows from
Lemma~\ref{lem:window-measurability}, and Markov's inequality gives the
probability bound.
\end{proof}

\section{Conserved local memory}\label{sec:supp3}
\subsection{Energy-conservation floor in bounded-degree Haar-local ensembles}

This section proves the energy-conservation floor under explicit locality and independence assumptions.

Consider \(N\) sites of fixed local dimension \(d_0\).  Let
\(\mathcal E_N\) be a family of interaction supports such that
\begin{enumerate}
\item every \(X\in\mathcal E_N\) contains at most \(k\) sites;
\item every site belongs to at most \(\Delta\) supports;
\item \(|\mathcal E_N|\le c_E N\).
\end{enumerate}
The constants \(d_0,k,\Delta,c_E\) are independent of \(N\).

For each support \(X\), fix a Hermitian local spectrum
\(\Lambda_X\) and let
\begin{equation}
  h_X
  =
  U_X\Lambda_XU_X^\dagger,
  \qquad
  H_N
  =
  \sum_{X\in\mathcal E_N}h_X,
  \label{eq:Haar-local-floor-ensemble}
\end{equation}
where the \(U_X\) are independent Haar unitaries on
\(\mathcal H_X\).  Subtracting
\[
  \frac{\Tr_X\Lambda_X}{d_X}\Id_X
\]
from every \(\Lambda_X\) changes \(H_N\) only by a scalar multiple of
the global identity and therefore does not change the Liouvillian.
Hence, without loss of generality, assume
\begin{equation}
  \Tr_X\Lambda_X=0
  \quad\text{for all }X.
  \label{eq:local-traceless-floor}
\end{equation}
Assume also
\begin{equation}
  \norm{\Lambda_X}_\infty\le J
  \label{eq:local-norm-floor}
\end{equation}
uniformly in \(X,N\).

The next projection estimate is a finite-dimensional Mazur bound;
the relation between conserved quantities and long-time correlations
goes back to Mazur and Suzuki.\cite{Mazur1969,Suzuki1971}
We give the proof with the normalization used here.

\begin{lemma}[Conserved energy lower-bounds the zero-frequency trace weight]
\label{lem:pi0-energy-lower}
Let
\[
  \widetilde H
  =
  H-\frac{\Tr H}{D}\Id.
\]
If \(\widetilde H\ne0\), then
\begin{equation}
  \pi_{0,S}
  \ge
  \frac1q
  \frac{
  \norm{\mathsf P_S\widetilde H}_{2,D}^2
  }{
  \norm{\widetilde H}_{2,D}^2
  },
  \qquad
  q=d^2-1.
  \label{eq:pi0-energy-lower}
\end{equation}
If \(\widetilde H=0\), the right-hand side is defined as zero and the
inequality remains valid.
\end{lemma}

\begin{proof}
Suppose first that \(\widetilde H\ne0\) and set
\[
  X_H
  =
  \frac{\widetilde H}{\norm{\widetilde H}_{2,D}}.
\]
Since \([H,\widetilde H]=0\), one has
\(X_H\in\Ran\mathsf Q_0\), and
\(\norm{X_H}_{2,D}=1\).
Extend \(X_H\) to an orthonormal basis
\(\{X_j\}\) of \(\Ran\mathsf Q_0\).  Then
\begin{align}
  q\pi_{0,S}
  &=
  \Tr_{\calK}
  \left[
  \mathsf P_S\mathsf Q_0\mathsf P_S
  \right]
  \nonumber\\
  &=
  \Tr_{\calK}
  \left[
  \mathsf Q_0\mathsf P_S\mathsf Q_0
  \right]
  \nonumber\\
  &=
  \sum_j
  \norm{\mathsf P_SX_j}_{2,D}^2
  \nonumber\\
  &\ge
  \norm{\mathsf P_SX_H}_{2,D}^2.
\end{align}
Substituting the definition of \(X_H\) proves
Eq.~\eqref{eq:pi0-energy-lower}.  If
\(\widetilde H=0\), the declared right-hand side is zero while
\(\pi_{0,S}\ge0\).
\end{proof}

\begin{lemma}[Deterministic extensive Hilbert--Schmidt upper bound]
\label{lem:HN-HS-upper}
For the ensemble
Eq.~\eqref{eq:Haar-local-floor-ensemble} under
Eqs.~\eqref{eq:local-traceless-floor} and
\eqref{eq:local-norm-floor}, there is a constant
\[
  C_H=C_H(k,\Delta,c_E,J)
\]
independent of \(N\) and of the Haar realization such that
\begin{equation}
  \norm{H_N}_{2,D}^2
  \le
  C_HN.
  \label{eq:HN-HS-upper}
\end{equation}
\end{lemma}

\begin{proof}
Expand
\[
  \norm{H_N}_{2,D}^2
  =
  \sum_{X,Y\in\mathcal E_N}
  \langle h_X,h_Y\rangle.
\]
If \(X\cap Y=\varnothing\), then the normalized trace factorizes:
\[
  \langle h_X,h_Y\rangle
  =
  \frac{\Tr_Xh_X}{d_X}
  \frac{\Tr_Yh_Y}{d_Y}
  =
  0
\]
by Eq.~\eqref{eq:local-traceless-floor}.
For arbitrary \(X,Y\), Cauchy--Schwarz gives
\[
  |\langle h_X,h_Y\rangle|
  \le
  \norm{h_X}_{2,D}
  \norm{h_Y}_{2,D}.
\]
Because \(h_X\) acts nontrivially only on \(X\),
\[
  \norm{h_X}_{2,D}^2
  =
  \frac1{d_X}\Tr_X(h_X^2)
  \le
  \norm{h_X}_\infty^2
  =
  \norm{\Lambda_X}_\infty^2
  \le
  J^2.
\]
Thus every nonzero summand has modulus at most \(J^2\).

Fix \(X\).  Since \(|X|\le k\) and every site belongs to at most
\(\Delta\) interaction supports, the number of \(Y\) with
\(Y\cap X\ne\varnothing\) is at most \(k\Delta\).  Hence the number of
ordered pairs \((X,Y)\) that can contribute is at most
\[
  |\mathcal E_N|\,k\Delta
  \le
  c_Ek\Delta\,N.
\]
Therefore
\[
  \norm{H_N}_{2,D}^2
  \le
  c_Ek\Delta J^2\,N.
\]
One may take
\[
  C_H=c_Ek\Delta J^2.
\]
\end{proof}

\begin{theorem}[Conservation-law floor for Haar-local ensembles]
\label{thm:Haar-local-pi0-floor}
Fix a subsystem \(S\), independent of \(N\), and suppose it contains one
interaction support \(X_0\in\mathcal E_N\).  Assume
\begin{equation}
  v_{X_0}
  :=
  \frac1{d_{X_0}}
  \Tr_{X_0}(\Lambda_{X_0}^2)
  \ge v_*>0
  \label{eq:vstar-floor}
\end{equation}
uniformly in \(N\).  Then
\begin{equation}
  \mathbb E\,\pi_{0,S}
  \ge
  \frac{v_*}{qC_H}\frac1N.
  \label{eq:Haar-local-pi0-floor}
\end{equation}
Thus, for the unseparated local concentration
\(\chi_S\) defined in Eq.~\eqref{eq:local-chi-new},
\[
  \mathbb E\,\chi_S(\varepsilon;H_N)
  \ge
  \frac{v_*}{qC_H}\frac1N
\]
for every \(\varepsilon>0\).
\end{theorem}

\begin{proof}
By Eq.~\eqref{eq:local-traceless-floor},
\[
  \Tr H_N=0,
\]
so \(\widetilde H_N=H_N\).
Lemma~\ref{lem:pi0-energy-lower} and
Lemma~\ref{lem:HN-HS-upper} give pointwise
\begin{equation}
  \pi_{0,S}
  \ge
  \frac{
  \norm{\mathsf P_SH_N}_{2,D}^2
  }{
  qC_HN
  }.
  \label{eq:pi0-floor-pointwise}
\end{equation}
It remains to lower-bound the expectation of the numerator.

Linearity of \(\mathsf P_S\) gives
\[
  \mathsf P_SH_N
  =
  \sum_X\mathsf P_Sh_X.
\]
For each traceless \(\Lambda_X\), Haar invariance implies
\[
  \mathbb E\,h_X
  =
  \frac{\Tr_X\Lambda_X}{d_X}\Id_X
  =
  0.
\]
The \(U_X\) are independent.  Hence for \(X\ne Y\),
conditioning on all variables except \(U_X\) gives
\[
  \mathbb E
  \langle
  \mathsf P_Sh_X,\mathsf P_Sh_Y
  \rangle
  =
  0.
\]
Therefore
\begin{align}
  \mathbb E
  \norm{\mathsf P_SH_N}_{2,D}^2
  &=
  \sum_X
  \mathbb E
  \norm{\mathsf P_Sh_X}_{2,D}^2
  \nonumber\\
  &\ge
  \mathbb E
  \norm{\mathsf P_Sh_{X_0}}_{2,D}^2.
  \label{eq:local-projection-expectation}
\end{align}
Because \(X_0\subseteq S\) and \(h_{X_0}\) is traceless,
\(h_{X_0}\), embedded by identities outside \(X_0\), belongs to
\(\calV_S\).  Hence
\[
  \mathsf P_Sh_{X_0}=h_{X_0}
\]
and
\[
  \norm{h_{X_0}}_{2,D}^2
  =
  \frac1{d_{X_0}}
  \Tr_{X_0}(h_{X_0}^2)
  =
  \frac1{d_{X_0}}
  \Tr_{X_0}(\Lambda_{X_0}^2)
  =
  v_{X_0}
  \ge
  v_*.
\]
Taking expectations in
Eq.~\eqref{eq:pi0-floor-pointwise} and using
Eq.~\eqref{eq:local-projection-expectation} proves
Eq.~\eqref{eq:Haar-local-pi0-floor}.

Finally every interval containing zero has
\[
  \chi_S(\varepsilon;H_N)\ge\pi_{0,S},
\]
so the same expectation lower bound holds for \(\chi_S\).
\end{proof}

\begin{remark}
The theorem explains why the zero mode must be separated before asking
for fast dynamical spectral spreading.  The \(N^{-1}\) floor is forced by
the conserved extensive energy in this ensemble, but
Theorem~\ref{thm:separated-stationary} charges stationary memory as
\(\pi_{0,S}^2\), so an \(O(N^{-1})\) zero-mode weight contributes only
\(O(N^{-2})\) to the mean-square stationary trace-distance error.
\end{remark}

\subsection{Control of the zero-frequency atom by weak ETH}
\label{sec:weak-ETH-zero-new}

Assume in this section that the energy eigenvalues are nondegenerate.
Define the infinite-temperature weak-ETH defect
\begin{equation}
  \delta_S(H)
  =
  \frac1D
  \sum_{n=1}^{D}
  \Dt\left(
  \rho_S^{(n)},\frac{\Id_S}{d}
  \right),
  \qquad
  \rho_S^{(n)}
  =
  \Tr_B\ket n\!\bra n.
  \label{eq:weak-ETH-defect}
\end{equation}

\begin{lemma}[Weak ETH controls exact local memory]
\label{lem:weak-ETH-zero-new}
Under nondegenerate energies,
\begin{equation}
  \pi_{0,S}
  \le
  \frac{2d}{d^2-1}
  \delta_S(H).
  \label{eq:weak-ETH-zero-new}
\end{equation}
\end{lemma}

\begin{proof}
Let \(P_1,\ldots,P_q\) be the Hermitian subsystem basis from
Eq.~\eqref{eq:local-operator-basis}, and
\(O_a=P_a\otimes\Id_B\).  Since the energies are nondegenerate,
\(\mathsf Q_0\) projects onto the span of
\(\{\ket n\!\bra n\}_n\).  Therefore
\[
  \norm{\mathsf Q_0O_a}_{2,D}^2
  =
  \frac1D
  \sum_n
  |\bra nO_a\ket n|^2.
\]
Set
\[
  X_n
  =
  \rho_S^{(n)}-\frac{\Id_S}{d}.
\]
Since \(P_a\) is traceless,
\[
  \bra nO_a\ket n
  =
  \Tr(P_aX_n).
\]
Parseval on the traceless subsystem operator space gives
\[
  \sum_{a=1}^{q}
  |\Tr(P_aX_n)|^2
  =
  d\,\Tr(X_n^2).
\]
Consequently
\begin{align}
  \pi_{0,S}
  &=
  \frac1q
  \sum_a
  \norm{\mathsf Q_0O_a}_{2,D}^2
  \nonumber\\
  &=
  \frac{d}{qD}
  \sum_n\Tr(X_n^2).
  \label{eq:pi0-Xn-new}
\end{align}
Every eigenvalue of \(X_n\) lies in
\([-1/d,1-1/d]\), so \(\norm{X_n}_\infty\le1\).  Schatten H\"older gives
\[
  \Tr(X_n^2)
  =
  \norm{X_n}_2^2
  \le
  \norm{X_n}_\infty\norm{X_n}_1
  \le
  2\Dt\left(
  \rho_S^{(n)},\frac{\Id_S}{d}
  \right).
\]
Substitute this into Eq.~\eqref{eq:pi0-Xn-new} and use
\(q=d^2-1\).
\end{proof}

\begin{corollary}[Same-time random-ensemble criterion]
\label{cor:ensemble-same-time-new}
Let \(H_N\) be random, let \(S\) be fixed, and put \(T_N=N^c\).
Suppose
\[
  \mathbb E A_S(T_N)
  \le
  N^{-2\kappa_0}.
\]
Then for every \(0<\kappa<\kappa_0\),
\begin{equation}
  \mathbb P\left[
  \widetilde\chi_S(N^{-c};H_N)>N^{-\kappa}
  \right]
  \le
  \frac{2}{c_*^2}
  N^{-2(\kappa_0-\kappa)}.
  \label{eq:ensemble-same-time-new}
\end{equation}
If in addition
\[
  \mathbb E\pi_{0,S}(H_N)\le p_N,
\]
then for every \(b_N>0\),
\begin{equation}
  \mathbb P[
  \pi_{0,S}(H_N)>b_N]
  \le
  \frac{p_N}{b_N}.
  \label{eq:pi0-Markov-new}
\end{equation}
Thus a polynomial second-moment autocorrelation theorem and independent
zero-mode control imply high-probability local equilibration on the same
windows \(T_N=N^c\).
\end{corollary}

\begin{proof}
The first inequality of
Theorem~\ref{thm:same-scale-equivalence-new} gives
\[
  \widetilde\chi_S(N^{-c};H_N)^2
  \le
  \frac2{c_*^2}A_S(T_N).
\]
The event
\[
  \widetilde\chi_S(N^{-c};H_N)>N^{-\kappa}
\]
therefore implies
\[
  A_S(T_N)
  >
  \frac{c_*^2}{2}N^{-2\kappa}.
\]
Markov's inequality and the assumed expectation bound give
Eq.~\eqref{eq:ensemble-same-time-new}.  Equation
\eqref{eq:pi0-Markov-new} is Markov's inequality applied directly to
\(\pi_{0,S}\ge0\).  The final statement follows from
Corollary~\ref{cor:separated-same-scale}.
\end{proof}

\begin{corollary}[Weak ETH plus finite-time mixing, without exponent loss]
\label{cor:weak-ETH-plus-mixing-new}
Let \(S\) be fixed and suppose the energies are nondegenerate.
If \(\delta_S(H_N)=O(N^{-\gamma})\) and
\(A_S(N^c)=O(N^{-2\kappa})\), then
\begin{equation}
 \mathscr E_{H_N}(N^c)
 =O_S(N^{-2\gamma})+O_S(N^{-2\kappa}).
 \label{eq:weak-ETH-plus-mixing-new}
\end{equation}
\end{corollary}
\begin{proof}
Lemma~\ref{lem:weak-ETH-zero-new} gives
\(\pi_{0,S}=O_S(N^{-\gamma})\).
Insert this and the assumed mean-square estimate directly into
Eq.~\eqref{eq:direct-from-autocorrelation-new}.
No conversion through \(\widetilde\chi_S\) is needed.
\end{proof}

\section{The periodic LFSR clock}\label{sec:supp4}
\subsection{Resonant static application: the periodic LFSR clock}
\label{sec:LFSR}

We now verify the local-band condition for an explicit static
geometrically local Hamiltonian whose spectrum is highly resonant.  The
construction is the periodic Feynman--Kitaev Hamiltonian associated with a
maximal linear-feedback shift-register (LFSR) circuit
\cite{IppolitiLong2026}.

Let \(n\ge2\) be the number of spin qubits and let
\[
  u_0,\ldots,u_{n-1}
\]
be the local two-qubit gates of the LFSR circuit.  In the one-clock-particle
sector the Hamiltonian is
\begin{equation}
  H_{\mathrm{FK}}
  =
  \Id
  -
  \frac12
  \sum_{t=0}^{n-1}
  \left(
  u_t\otimes\ket{t+1}\!\bra t
  +
  u_t^\dagger\otimes\ket t\!\bra{t+1}
  \right),
  \label{eq:FK}
\end{equation}
with periodic clock labels.  Its Hilbert-space dimension is
\begin{equation}
  D=n2^n.
  \label{eq:FK-D}
\end{equation}
The physical clock may be represented by \(n\) local clock modes; the
corresponding second-quantized Hamiltonian conserves clock number, so the
one-particle sector is dynamically invariant.

\subsection{Periodic-clock spectrum}

Let
\[
  U_F=u_{n-1}\cdots u_0
\]
and suppose
\[
  U_F\ket\phi=e^{i\phi}\ket\phi.
\]
Define
\[
  U_{t:0}
  =
  \begin{cases}
  \Id,&t=0,\\
  u_{t-1}\cdots u_0,&1\le t\le n,
  \end{cases}
\]
and
\[
  \ket{K_{t,\phi}}
  =
  U_{t:0}\ket\phi\otimes\ket t.
\]

\begin{proposition}[Periodic-clock eigenstates]
\label{prop:clock-spectrum}
For \(m=0,\ldots,n-1\), set
\[
  \theta_{m,\phi}
  =
  \frac{2\pi m}{n}-\frac{\phi}{n}.
\]
Then
\begin{equation}
  \ket{\Psi_{m,\phi}}
  =
  \frac1{\sqrt n}
  \sum_{t=0}^{n-1}
  e^{i\theta_{m,\phi}t}
  \ket{K_{t,\phi}}
  \label{eq:clock-state}
\end{equation}
is an eigenstate of \(H_{\mathrm{FK}}\) with
\begin{equation}
  E_{m,\phi}
  =
  1-\cos\theta_{m,\phi}.
  \label{eq:clock-energy}
\end{equation}
\end{proposition}

\begin{proof}
For \(0\le t\le n-2\),
\[
  u_tU_{t:0}=U_{t+1:0}.
\]
At the periodic boundary,
\[
  U_{n:0}\ket\phi
  =
  U_F\ket\phi
  =
  e^{i\phi}\ket\phi.
\]
Thus the Krylov basis obeys the twisted boundary condition
\[
  \ket{K_{n,\phi}}
  =
  e^{i\phi}\ket{K_{0,\phi}}.
\]
The plane-wave coefficient is consistent with this twist because
\[
  e^{i\theta_{m,\phi}n}e^{i\phi}
  =
  e^{2\pi im}
  =
  1.
\]
Acting with Eq.~\eqref{eq:FK} on
Eq.~\eqref{eq:clock-state},
\begin{align}
  H_{\mathrm{FK}}\ket{\Psi_{m,\phi}}
  &=
  \left(
  1-\frac12e^{i\theta_{m,\phi}}
  -\frac12e^{-i\theta_{m,\phi}}
  \right)
  \ket{\Psi_{m,\phi}}
  \nonumber\\
  &=
  \left(
  1-\cos\theta_{m,\phi}
  \right)
  \ket{\Psi_{m,\phi}}.
\end{align}
\end{proof}

For a maximal \(n\)-bit LFSR, put
\[
  L=2^n-1.
\]

\begin{lemma}[Maximal-LFSR Floquet phases]
\label{lem:LFSR-cycle-phases}
The classical maximal LFSR fixes the all-zero bit string and forms one
cycle of length \(L\) on the nonzero bit strings.  Consequently its
permutation unitary \(U_F\) has a two-dimensional \(+1\) eigenspace and
the remaining eigenvalues are
\begin{equation}
  e^{i\phi_q}
  =
  e^{-2\pi iq/L},
  \qquad
  q=1,\ldots,L-1.
  \label{eq:LFSR-phases}
\end{equation}
\end{lemma}

\begin{proof}
Maximality of an \(n\)-bit LFSR means precisely that its classical update
has period \(2^n-1\) on every nonzero state.  Since there are exactly
\(2^n-1\) nonzero bit strings, they form one cycle.  The zero string is
fixed by linearity.

List the nonzero cycle as
\[
  x_0,x_1,\ldots,x_{L-1},
  \qquad
  U_F\ket{x_t}=\ket{x_{t+1}},
\]
with indices modulo \(L\).  The Fourier vectors
\[
  \ket{\psi_q}
  =
  \frac1{\sqrt L}
  \sum_{t=0}^{L-1}
  e^{2\pi iqt/L}\ket{x_t},
  \qquad
  q=0,\ldots,L-1,
\]
satisfy
\[
  U_F\ket{\psi_q}
  =
  e^{-2\pi iq/L}\ket{\psi_q}.
\]
The \(q=0\) vector and the fixed zero-string vector are independent
\(+1\) eigenvectors.  The cycle decomposition contains no other cycles,
so there are no further \(+1\) eigenvectors.  Removing \(q=0\) leaves
Eq.~\eqref{eq:LFSR-phases}.
\end{proof}

Combining Lemma~\ref{lem:LFSR-cycle-phases} with
Proposition~\ref{prop:clock-spectrum} gives the regular static energies
\begin{equation}
  E_{m,q}
  =
  1-\cos\left[
  \frac{2\pi}{nL}(mL+q)
  \right],
  \qquad
  0\le m<n,
  \quad
  1\le q<L.
  \label{eq:LFSR-energy}
\end{equation}

We use one non-elementary external input from
Ref.~\cite{IppolitiLong2026}.

\begin{theorem}[Maximal-LFSR Pauli matrix elements; external]
\label{thm:LFSR-external}
For every nonidentity spin Pauli operator \(P\), and every pair of
nontrivial Floquet eigenstates \(\ket{\psi_q},\ket{\psi_{q'}}\),
\begin{equation}
  \abs{\bra{\psi_q}P\ket{\psi_{q'}}}
  \le
  c_n,
  \qquad
  c_n
  :=
  \frac{2^{1+n/2}}{2^n-1}.
  \label{eq:LFSR-Pauli-input}
\end{equation}
\end{theorem}

\begin{remark}
Theorem~\ref{thm:LFSR-external} is an external literature input.  The
Floquet phase grid and exceptional multiplicity were proved above, and
all static spectral consequences needed below are derived explicitly.

\end{remark}

Each of the two exceptional Floquet eigenvectors produces \(n\) clock
momenta.  Therefore the exceptional static sector has dimension
\begin{equation}
  r=2n.
  \label{eq:LFSR-r}
\end{equation}

\begin{lemma}[Static local matrix elements]
\label{lem:LFSR-static-matrix}
Let \(P\ne\Id\) be a spin Pauli operator, extended trivially to the clock.
For regular static eigenstates,
\begin{equation}
  \abs{
  \bra{\Psi_{m,q}}
  P\otimes\Id_{\rm clock}
  \ket{\Psi_{m',q'}}
  }
  \le
  c_n.
  \label{eq:LFSR-static-matrix}
\end{equation}
\end{lemma}

\begin{proof}
Using Eq.~\eqref{eq:clock-state} and clock orthogonality,
\begin{align}
  &\bra{\Psi_{m,q}}
  P\otimes\Id_{\rm clock}
  \ket{\Psi_{m',q'}}
  \nonumber\\
  &\qquad=
  \frac1n
  \sum_{t=0}^{n-1}
  e^{i(\theta_{m',q'}-\theta_{m,q})t}
  \bra{\psi_q}
  U_{t:0}^\dagger P U_{t:0}
  \ket{\psi_{q'}}.
\end{align}
The LFSR gates are Clifford, so
\(U_{t:0}^\dagger P U_{t:0}\)
is again a nonidentity Pauli operator.  Apply
Eq.~\eqref{eq:LFSR-Pauli-input} term by term and use the triangle
inequality.
\end{proof}

\subsection{Cosine-grid difference anti-concentration}

Define
\[
  e(\theta)=1-\cos\theta,
  \qquad
  0\le\theta<2\pi.
\]
Let \(\Theta\) be uniform on \([0,2\pi)\), and let \(\nu\) denote the law
of \(e(\Theta)\).

\begin{lemma}[Arcsine density for the cosine energy]
\label{lem:arcsine}
The measure \(\nu\) is supported on \([0,2]\) and has density
\begin{equation}
  f_E(u)
  =
  \frac{
  \mathbbm 1_{(0,2)}(u)
  }{
  \pi\sqrt{u(2-u)}
  }.
  \label{eq:energy-density}
\end{equation}
\end{lemma}

\begin{proof}
Let
\[
  X=\cos\Theta.
\]
For \(x\in(-1,1)\), the equation \(\cos\theta=x\) has two solutions in
\([0,2\pi)\), and on each branch
\[
  \left|\frac{d\theta}{dx}\right|
  =
  \frac1{\sqrt{1-x^2}}.
\]
Since \(\Theta\) has density \(1/(2\pi)\), summing the two branches gives
\[
  f_X(x)
  =
  \frac1{\pi\sqrt{1-x^2}},
  \qquad -1<x<1.
\]
Now \(E=1-X\).  Therefore
\[
  f_E(u)
  =
  f_X(1-u)
  =
  \frac1{
  \pi\sqrt{1-(1-u)^2}
  }
  =
  \frac1{
  \pi\sqrt{u(2-u)}
  },
\]
for \(0<u<2\).
\end{proof}

\subsection{Discrete cosine-grid discrepancy}

For an integer \(M\ge1\), put
\[
  \theta_a=\frac{2\pi a}{M},
  \qquad
  e_a=1-\cos\theta_a,
  \qquad
  a=0,\ldots,M-1,
\]
and define the empirical energy measure
\begin{equation}
  \nu_M
  =
  \frac1M\sum_{a=0}^{M-1}\delta_{e_a}.
  \label{eq:nuM}
\end{equation}

\begin{lemma}[Counting equally spaced points in an interval]
\label{lem:grid-interval}
Let \(\Theta_M=\{2\pi a/M:a=0,\ldots,M-1\}\).
For every interval \(J\subset[0,2\pi)\), with any included or excluded
finite endpoints,
\begin{equation}
 \left|\#(\Theta_M\cap J)-\frac{M|J|}{2\pi}\right|\le1.
 \label{eq:grid-interval}
\end{equation}
\end{lemma}
\begin{proof}
Rescale the endpoints by \(M/(2\pi)\).
For an interval with endpoints \(a\le b\), the integer counts for
\([a,b]\), \([a,b)\), \((a,b]\), and \((a,b)\) are, respectively,
\[
 \max(0,\lfloor b\rfloor-\lceil a\rceil+1),\quad
 \max(0,\lceil b\rceil-\lceil a\rceil),\quad
 \max(0,\lfloor b\rfloor-\lfloor a\rfloor),\quad
 \max(0,\lceil b\rceil-\lfloor a\rfloor-1).
\]
The inequalities \(x-1<\lfloor x\rfloor\le x\) and
\(x\le\lceil x\rceil<x+1\), including the equality cases at integers,
show that each count differs from \(b-a\) by at most one.
The empty-interval and singleton cases obey the same bound.
Since \(J\subset[0,2\pi)\), rescaling includes no grid integers outside
\(0,\ldots,M-1\).  This proves the claim.
\end{proof}

\begin{lemma}[Kolmogorov discrepancy of the cosine grid]
\label{lem:kolmogorov}
Let
\[
  F_M(y)=\nu_M((-\infty,y]),
  \qquad
  F(y)=\nu((-\infty,y]).
\]
Then
\begin{equation}
  \sup_{y\in\mathbb R}
  |F_M(y)-F(y)|
  \le
  \frac2M.
  \label{eq:kolmogorov}
\end{equation}
\end{lemma}

\begin{proof}
For \(y<0\) or \(y\ge2\), the claim is immediate.  Let \(0\le y<2\) and
define
\[
  \alpha=\arccos(1-y)\in[0,\pi).
\]
The condition
\[
  1-\cos\theta\le y
\]
is equivalent to
\[
  \theta\in[0,\alpha]\cup[2\pi-\alpha,2\pi).
\]
The two intervals are disjoint and have total length \(2\alpha\).
At \(\alpha=0\) the first is a singleton and the second is empty.
Lemma~\ref{lem:grid-interval} bounds each interval's counting error in
absolute value by one, so their union has counting error at most two.  Dividing by \(M\) gives
\[
  |F_M(y)-\alpha/\pi|
  \le\frac2M.
\]
Since \(\alpha/\pi=F(y)\), the result follows.
\end{proof}

\begin{lemma}[Integration against a Kolmogorov perturbation]
\label{lem:BV}
Let \(\mu_1,\mu_2\) be probability measures with distribution functions
\(G_1,G_2\), and suppose
\(\sup_y|G_1(y)-G_2(y)|\le\eta\).
If \(h\) is absolutely continuous, has compact support, and
\(h'\in L^1(\mathbb R)\), then
\begin{equation}
 \left|\int h\,d\mu_1-\int h\,d\mu_2\right|
 \le\eta\int_{\mathbb R}|h'(y)|dy.
 \label{eq:BV}
\end{equation}
\end{lemma}
\begin{proof}
Put \(\sigma=\mu_1-\mu_2\).  Since \(h\) has compact support and is
absolutely continuous,
\(h(x)=\int_{-\infty}^x h'(y)dy\).
Absolute integrability with respect to \(|\sigma|\) and Lebesgue
measure follows from
\(|\sigma|(\mathbb R)\int|h'|\le2\int|h'|<\infty\).
Interchanging integrals consequently gives
\[
 \int h\,d\sigma
 =\int h'(y)\sigma([y,\infty))dy
 =-\int h'(y)\sigma((-\infty,y))dy.
\]
The two distribution functions have left limits, whose difference is
also bounded by \(\eta\).  Taking absolute values proves the claim.
The absolute-continuity hypothesis is sufficient for all uses below;
it avoids any convention about simultaneous Stieltjes jumps.
\end{proof}

\subsection{Difference of two arcsine variables}

Let \(X,Y\) be independent random variables with density
\[
  f(x)
  =
  \frac{
  \mathbbm 1_{(-1,1)}(x)
  }{
  \pi\sqrt{1-x^2}
  }.
\]
Let
\[
  Z=X-Y.
\]

\begin{lemma}[Density of the difference]
\label{lem:difference-density}
The random variable \(Z\) has an even density \(g\) supported on
\([-2,2]\).  For \(0<z<2\),
\begin{equation}
  g(z)
  =
  \frac1{\pi^2}
  \int_{z-1}^{1}
  \frac{dx}{
  \sqrt{
  (1-x^2)
  (1-(x-z)^2)
  }
  }.
  \label{eq:g-convolution}
\end{equation}
Furthermore,
\begin{equation}
  g(z)
  \le
  \log\left(\frac4{|z|}\right),
  \qquad
  0<|z|<2.
  \label{eq:g-log}
\end{equation}
\end{lemma}

\begin{proof}
The convolution formula gives
\[
  g(z)
  =
  \int_{\mathbb R}f(x)f(x-z)\,dx.
\]
For \(0<z<2\), both factors are nonzero precisely when
\[
  x\in[z-1,1],
\]
which proves Eq.~\eqref{eq:g-convolution}.  Since \(X-Y\) and \(Y-X\)
have the same distribution, \(g\) is even.

Fix \(0<z<2\), put
\[
  L=2-z,
\]
and make the change of variables
\[
  x=z-1+Lt,
  \qquad 0\le t\le1.
\]
Then
\begin{align*}
  1-x&=L(1-t),\\
  1+x&=z+Lt,\\
  1-x+z&=2-Lt,\\
  1+x-z&=Lt.
\end{align*}
Consequently,
\begin{equation}
  g(z)
  =
  \frac1{\pi^2}
  \int_0^1
  \frac{dt}{
  \sqrt{
  t(1-t)(z+Lt)(2-Lt)
  }
  }.
  \label{eq:g-t}
\end{equation}

First assume \(0<z\le1\), so \(L\ge1\).  On \(0<t\le1/2\),
\[
  1-t\ge\frac12,
  \qquad
  2-Lt
  \ge
  2-\frac L2
  =
  1+\frac z2
  \ge1,
\]
and
\[
  z+Lt\ge z+t.
\]
Thus the integrand in Eq.~\eqref{eq:g-t} is at most
\[
  \frac{\sqrt2}{\sqrt{t(z+t)}}.
\]
The integrand in Eq.~\eqref{eq:g-t} is invariant under
\(t\mapsto1-t\), because this transformation exchanges
\(z+Lt\) and \(2-Lt\).  Hence
\begin{align*}
  g(z)
  &\le
  \frac{2\sqrt2}{\pi^2}
  \int_0^{1/2}
  \frac{dt}{\sqrt{t(z+t)}}.
\end{align*}
With \(t=z\sinh^2u\),
\[
  \int_0^a
  \frac{dt}{\sqrt{t(z+t)}}
  =
  2\,\operatorname{arsinh}\sqrt{\frac az}.
\]
Therefore
\[
  g(z)
  \le
  \frac{4\sqrt2}{\pi^2}
  \operatorname{arsinh}
  \frac1{\sqrt{2z}}.
\]
Using
\[
  \operatorname{arsinh}y
  =
  \log(y+\sqrt{1+y^2})
  \le
  \log(1+2y)
\]
for \(y\ge0\), and
\[
  1+\sqrt{\frac2z}
  \le
  \frac4z
  \qquad (0<z\le1),
\]
we obtain
\[
  g(z)
  \le
  \frac{4\sqrt2}{\pi^2}
  \log\frac4z
  \le
  \log\frac4z,
\]
because \(4\sqrt2/\pi^2<1\).

Now assume \(1\le z<2\), so \(0<L\le1\).  For every \(t\in[0,1]\),
\[
  z+Lt\ge1,
  \qquad
  2-Lt\ge2-L=z\ge1.
\]
Equation~\eqref{eq:g-t} gives
\[
  g(z)
  \le
  \frac1{\pi^2}
  \int_0^1\frac{dt}{\sqrt{t(1-t)}}
  =
  \frac1\pi.
\]
Since
\[
  \log\frac4z\ge\log2>\frac1\pi
\]
for \(1\le z<2\), Eq.~\eqref{eq:g-log} follows in this case as well.
Evenness gives the negative-\(z\) half.
\end{proof}

\begin{lemma}[Continuum difference concentration]
\label{lem:continuum-concentration}
For every interval \(I\subset\mathbb R\) of length
\(0<\delta\le1\),
\begin{equation}
  \mathbb P(Z\in I)
  \le
  2\delta
  \log\left(\frac{4e}{\delta}\right).
  \label{eq:continuum-concentration}
\end{equation}
\end{lemma}

\begin{proof}
By Lemma~\ref{lem:difference-density},
\[
  g(z)\le\log\frac4{|z|}
\]
on the support \([-2,2]\), away from \(z=0\).  Use the
nonnegative envelope
\(h(z)=\mathbbm1_{\{0<|z|\le2\}}\log(4/|z|)\), with its value
at zero immaterial for integration.  It is even and nonincreasing
on the positive half-line.  The positive and negative parts of an
interval of length \(\delta\) each have length at most \(\delta\).
Their integrals are therefore each bounded by
\(\int_0^\delta h(z)\,dz\).  Consequently
\[
  \mathbb P(Z\in I)\le\int_I h(z)\,dz
  \le 2\int_0^\delta\log\frac4z\,dz.
\]
The integral is explicit:
\[
  \int_0^\delta\log\frac4z\,dz
  =
  \delta\log\frac4\delta+\delta
  =
  \delta\log\frac{4e}{\delta}.
\]
This proves Eq.~\eqref{eq:continuum-concentration}.
\end{proof}

\subsection{Discrete cosine-difference concentration}

Define the full-grid ordered-pair concentration
\begin{equation}
  N_M(\delta)
  =
  \sup_{x\in\mathbb R}
  \#\left\{
  (a,b)\in\mathbb Z_M^2:
  e_a-e_b\in[x,x+\delta)
  \right\}.
  \label{eq:NM}
\end{equation}

\begin{lemma}[Two-sided difference-law discrepancy, with half-open intervals]
\label{lem:cosine-difference-discrepancy}
Let \(U,V\) be independent with law \(\nu_M\), and \(E_1,E_2\)
independent with the continuous law \(\nu\).
For every \(I=[a,b)\), \(a<b\),
\begin{equation}
 \left|\mathbb P(U-V\in I)-\mathbb P(E_1-E_2\in I)\right|
 \le\frac8M.
 \label{eq:two-sided-difference-discrepancy}
\end{equation}
\end{lemma}
\begin{proof}
For fixed \(u\), the event \(u-V\in[a,b)\) is
\(V\in(u-b,u-a]\).  Its probability is
\[
 h_M(u)=F_M(u-a)-F_M(u-b).
\]
For the continuous law it is
\(h(u)=F(u-a)-F(u-b)\).
Lemma~\ref{lem:kolmogorov} gives
\(|h_M(u)-h(u)|\le4/M\).
The function \(h\) has compact support, is absolutely continuous, and
has derivative \(f_E(u-a)-f_E(u-b)\) almost everywhere.
Since \(\int f_E=1\), its derivative has \(L^1\) norm at most two.
Lemma~\ref{lem:BV} therefore gives
\[
 \left|\int h_M\,d\nu_M-\int h\,d\nu\right|
 \le\int|h_M-h|\,d\nu_M
       +\left|\int h\,d(\nu_M-\nu)\right|
 \le\frac4M+\frac2M\,2=\frac8M.
\]
These two integrals are the difference probabilities in question.
\end{proof}

\begin{theorem}[Cosine-grid difference count]
\label{thm:cosine-difference}
For every integer \(M\ge1\) and \(0<\delta\le1\),
\begin{equation}
 N_M(\delta)\le
 2M^2\delta\log(4e/\delta)+8M.
 \label{eq:cosine-difference}
\end{equation}
\end{theorem}
\begin{proof}
For \(I=[x,x+\delta)\), its normalized ordered-pair count is
\(\mathbb P(U-V\in I)\).
Lemma~\ref{lem:cosine-difference-discrepancy} bounds this by the
continuous difference probability plus \(8/M\).
Since \(E_1-E_2\) is a difference of two arcsine variables,
Lemma~\ref{lem:continuum-concentration} bounds that probability by
\(2\delta\log(4e/\delta)\).
Multiply by \(M^2\) and take the supremum over \(x\).
\end{proof}

\begin{remark}
Counting one energy at a time gives \(O(M^2\sqrt\delta)\).  Theorem
\ref{thm:cosine-difference} replaces the square root by
\(\delta\log(1/\delta)\).  The logarithm is natural: it is the singularity
of the difference density at zero.
\end{remark}

\subsection{Local frequency concentration for the LFSR clock}

Fix a subsystem \(S\) of \(s\) spin qubits only, with \(s\) independent of
\(n\).  Thus
\[
  d=2^s,
  \qquad
  q=d^2-1.
\]

For a normalized nonidentity local Pauli \(O_a\), and a frequency interval
\(I\), the local spectral weight is
\begin{equation}
  \langle O_a,\mathsf Q_IO_a\rangle
  =
  \frac1D
  \sum_{\substack{j,k\\E_j-E_k\in I}}
  |\bra jO_a\ket k|^2.
  \label{eq:local-weight-matrix}
\end{equation}

\begin{theorem}[LFSR local-band concentration]
\label{thm:improved-chi}
Under Theorem~\ref{thm:LFSR-external} and Lemma~\ref{lem:LFSR-static-matrix}, for every \(0<\delta\le1\),
\begin{equation}
  \chi_S(\delta;H_{\rm FK})
  \le
  8n\delta
  \log\left(\frac{4e}{\delta}\right)
  +
  \frac{32}{2^n-1}
  +
  \frac4{2^n}.
  \label{eq:improved-chi}
\end{equation}
In particular,
\begin{equation}
  \chi_S(\delta;H_{\rm FK})
  =
  O\!\left(
  n\delta\log\frac e\delta
  \right)
  +
  O(2^{-n}).
  \label{eq:improved-chi-O}
\end{equation}
\end{theorem}

\begin{proof}
Let \(I\) be a frequency interval of width \(\delta\).
The regular energies form a subset of the full \(M\)-point cosine grid.
Therefore the number \(N_{\mathcal R}(I)\) of ordered regular pairs
\((j,k)\) with \(E_j-E_k\in I\) is no larger than the full-grid count of
Theorem~\ref{thm:cosine-difference}:
\begin{equation}
  N_{\mathcal R}(I)
  \le
  2M^2\delta
  \log\left(\frac{4e}{\delta}\right)
  +
  8M.
  \label{eq:NR-improved}
\end{equation}

For regular pairs, Lemma~\ref{lem:LFSR-static-matrix} gives
\[
  |\bra jO_a\ket k|^2\le c_n^2.
\]
Hence their contribution to
Eq.~\eqref{eq:local-weight-matrix} is at most
\[
  \frac{N_{\mathcal R}(I)c_n^2}{D}.
\]

Now consider pairs for which at least one index is exceptional.  Since
\(O_a\) is unitary,
\[
  \sum_k|\bra jO_a\ket k|^2
  =
  \bra jO_a^\dagger O_a\ket j
  =
  1.
\]
There are \(r=2n\) exceptional choices of \(j\), so all pairs with
exceptional first index contribute at most \(r/D\).  The same argument
with the second index contributes at most another \(r/D\).  Counting their
intersection twice only increases the upper bound.  Therefore
\begin{equation}
  \langle O_a,\mathsf Q_IO_a\rangle
  \le
  \frac{N_{\mathcal R}(I)c_n^2}{D}
  +
  \frac{2r}{D}.
  \label{eq:weight-before-algebra}
\end{equation}

We now evaluate the coefficients exactly.  Recall
\[
  M=n(2^n-1),
  \qquad
  D=n2^n,
  \qquad
  c_n^2
  =
  \frac{2^{n+2}}{(2^n-1)^2}.
\]
Then
\begin{align}
  \frac{M^2c_n^2}{D}
  &=
  \frac{
  n^2(2^n-1)^2
  }{
  n2^n
  }
  \frac{
  2^{n+2}
  }{
  (2^n-1)^2
  }
  =
  4n,
  \label{eq:M2-exact}\\
  \frac{8Mc_n^2}{D}
  &=
  8
  \frac{
  n(2^n-1)
  }{
  n2^n
  }
  \frac{
  2^{n+2}
  }{
  (2^n-1)^2
  }
  =
  \frac{32}{2^n-1},
  \label{eq:M-exact}\\
  \frac{2r}{D}
  &=
  \frac{4n}{n2^n}
  =
  \frac4{2^n}.
  \label{eq:r-exact}
\end{align}
Substituting Eq.~\eqref{eq:NR-improved} into
Eq.~\eqref{eq:weight-before-algebra} and then using
Eqs.~\eqref{eq:M2-exact}--\eqref{eq:r-exact} gives
\[
  \langle O_a,\mathsf Q_IO_a\rangle
  \le
  8n\delta
  \log\left(\frac{4e}{\delta}\right)
  +
  \frac{32}{2^n-1}
  +
  \frac4{2^n}.
\]
The right-hand side is independent of \(a\).  Averaging over the
\(q\) local Pauli strings and taking the supremum over \(I\) proves
Eq.~\eqref{eq:improved-chi}.
\end{proof}

\begin{corollary}[Residual zero-frequency weight]
\label{cor:zero}
Let
\[
  \pi_{0,S}
  =
  \mu_S(\{0\}).
\]
Then
\begin{equation}
  \pi_{0,S}
  \le
  \frac{32}{2^n-1}
  +
  \frac4{2^n}
  =
  O(2^{-n}).
  \label{eq:zero}
\end{equation}
\end{corollary}

\begin{proof}
For every \(\delta>0\), the interval \([0,\delta)\) contains the zero
frequency, hence
\[
  \pi_{0,S}\le\chi_S(\delta;H_{\rm FK}).
\]
Apply Theorem~\ref{thm:improved-chi} and take
\(\delta\downarrow0\).  Since
\[
  \lim_{\delta\downarrow0}
  \delta\log\frac{4e}{\delta}
  =0,
\]
the claim follows.
\end{proof}

\subsection{Consequences for the maximal LFSR clock}

\begin{theorem}[Infinite-time LFSR equilibration]
\label{thm:LFSR-infinite-improved}
For every fixed spin-only subsystem \(S\),
\begin{equation}
  \av{
  \sup_{\rho_S}
  \Dt\left(
  \Phi_t(\rho_S),\frac{\Id_S}{2^s}
  \right)^2
  }
  =
  O_s(2^{-n}).
  \label{eq:LFSR-infinite-improved}
\end{equation}
\end{theorem}

\begin{proof}
For fixed \(0<\delta\le1\), take \(T\to\infty\) in
Theorem~\ref{thm:fejer-local-band}.  Since
\[
  \lim_{T\to\infty}
  \left[
  1+\frac{2\pi^2}{3(T\delta)^2}
  \right]
  =1,
\]
we obtain
\[
  \av{
  \sup_{\rho_S}
  \Dt\left(
  \Phi_t(\rho_S),\frac{\Id_S}{d}
  \right)^2
  }
  \le
  (d-1)q^2\chi_S(\delta;H_{\rm FK}).
\]
Now take \(\delta\downarrow0\) and use
Theorem~\ref{thm:improved-chi}.  Its first term tends to zero and the
remaining terms are \(O(2^{-n})\).  Since \(d=2^s\) is fixed, the
prefactor is independent of \(n\).
\end{proof}

\begin{corollary}[Stationary local state]
\label{cor:LFSR-stationary-improved}
For the temporal channel
\(\overline{\Phi}\),
\begin{equation}
  \sup_{\rho_S}
  \Dt\left(
  \overline{\Phi}(\rho_S),\frac{\Id_S}{2^s}
  \right)
  =
  O_s(2^{-n}).
  \label{eq:LFSR-stationary-improved}
\end{equation}
\end{corollary}

\begin{proof}
Theorem~\ref{thm:separated-stationary} gives
\[
  \Dt\left(
  \overline{\Phi}(\rho_S),\frac{\Id_S}{d}
  \right)
  \le
  \frac{\sqrt{d-1}}2\,q\,\pi_{0,S}.
\]
Corollary~\ref{cor:zero} gives
\(\pi_{0,S}=O(2^{-n})\).
\end{proof}

\begin{theorem}[Direct every-window LFSR bound]
\label{thm:LFSR-direct-T}
For every fixed spin-only subsystem \(S\), there exists
\(C_s>0\) such that for every \(n\) and every \(T\ge1\),
\begin{equation}
  \sup_{t_0\ge0}
  \frac1T
  \int_{t_0}^{t_0+T}
  \sup_{\rho_S}
  \Dt\left(
  \Phi_t(\rho_S),\frac{\Id_S}{2^s}
  \right)^2dt
  \le
  C_s
  \left[
  \frac{n\log(eT)}{T}
  +
  2^{-n}
  \right].
  \label{eq:LFSR-direct-T}
\end{equation}
\end{theorem}

\begin{proof}
Set \(\varepsilon=T^{-1}\) in
Theorem~\ref{thm:fejer-local-band}.  Since \(T\ge1\),
\(0<\varepsilon\le1\), and Theorem~\ref{thm:improved-chi} gives
\[
  \chi_S(T^{-1};H_{\rm FK})
  \le
  \frac{8n}{T}\log(4eT)
  +
  \frac{32}{2^n-1}
  +
  \frac4{2^n}.
\]
Furthermore
\[
  \log(4eT)
  \le
  (1+\log4)\log(eT)
\]
because \(\log(eT)\ge1\).  For \(n\ge2\),
\[
  \frac{32}{2^n-1}+\frac4{2^n}
  \le47\,2^{-n}.
\]
The finitely many smaller values of \(n\) can be absorbed into the same
constant \(C_s\).  Equation~\eqref{eq:fejer-local-trace} with
\(T\varepsilon=1\) completes the proof.
\end{proof}

\begin{corollary}[Vanishing-error scaling]
\label{cor:LFSR-alpha-improved}
For every fixed \(\alpha>0\),
\begin{equation}
  T_n
  =
  n^{1+\alpha}\log(en)
  \label{eq:LFSR-T-alpha}
\end{equation}
implies
\begin{equation}
  \sup_{t_0\ge0}
  \frac1{T_n}
  \int_{t_0}^{t_0+T_n}
  \sup_{\rho_S}
  \Dt\left(
  \Phi_t(\rho_S),\frac{\Id_S}{2^s}
  \right)^2dt
  =
  O_{s,\alpha}(n^{-\alpha})
  +
  O_s(2^{-n}).
  \label{eq:LFSR-alpha-improved}
\end{equation}
In particular,
\begin{equation}
  T_n=n^2\log(en)
  \label{eq:LFSR-n2log}
\end{equation}
is sufficient for mean-square error \(O_s(n^{-1})\).
\end{corollary}

\begin{proof}
Equation~\eqref{eq:LFSR-T-alpha} gives
\[
  \log(eT_n)
  =
  1+(1+\alpha)\log n+\log\log(en)
  =
  O_\alpha(\log(en)).
\]
Therefore
\[
  \frac{n\log(eT_n)}{T_n}
  =
  O_\alpha(n^{-\alpha}).
\]
Apply Theorem~\ref{thm:LFSR-direct-T}.  Setting \(\alpha=1\) gives
Eq.~\eqref{eq:LFSR-n2log}.
\end{proof}

\begin{proposition}[Exponentially fine microscopic energy]
\label{prop:LFSR-small-gap-improved}
Let \(L=2^n-1\).  The smallest positive energy in the
one-clock-particle sector is exactly
\begin{equation}
  \Delta_n
  =
  1-
  \cos\left(
  \frac{2\pi}{nL}
  \right).
  \label{eq:LFSR-small-gap-improved}
\end{equation}
Moreover,
\begin{equation}
  \frac{8}{n^2L^2}
  \le
  \Delta_n
  \le
  \frac{2\pi^2}{n^2L^2},
  \label{eq:LFSR-small-gap-two-sided}
\end{equation}
and consequently
\begin{equation}
  \Delta_n^{-1}
  =
  \Theta(n^24^n).
  \label{eq:LFSR-small-gap-theta}
\end{equation}
\end{proposition}

\begin{proof}
Lemma~\ref{lem:LFSR-cycle-phases} gives the regular Floquet labels
\(q=1,\ldots,L-1\), while the two exceptional Floquet eigenvectors
have phase zero.  Proposition~\ref{prop:clock-spectrum} therefore
produces static energies
\[
  1-\cos\left(
  \frac{2\pi a}{nL}
  \right),
  \qquad
  a=0,\ldots,nL-1.
\]
Indeed, the regular states give \(a=mL+q\) with
\(q=1,\ldots,L-1\), while the exceptional phase-zero states give
\(a=mL\).  Thus the whole \(nL\)-point cosine grid occurs in the
sector.

The only zero energy corresponds to \(a=0\) modulo \(nL\).  Since
\(1-\cos x\) increases on \([0,\pi]\) and depends only on the
distance of \(x\) from \(2\pi\mathbb Z\), the smallest positive value
is attained at \(a=1\) and \(a=nL-1\).  This proves
Eq.~\eqref{eq:LFSR-small-gap-improved}.

For \(0\le x\le\pi\),
\[
  \frac{2x^2}{\pi^2}
  \le
  1-\cos x
  \le
  \frac{x^2}{2}.
\]
The lower inequality follows from
\(1-\cos x=2\sin^2(x/2)\) and
\(\sin y\ge2y/\pi\) on \(0\le y\le\pi/2\).
Substituting \(x=2\pi/(nL)\) gives
Eq.~\eqref{eq:LFSR-small-gap-two-sided}.  Finally,
\(L=2^n-1=\Theta(2^n)\), which proves
Eq.~\eqref{eq:LFSR-small-gap-theta}.
\end{proof}

\subsection{Origin of the logarithm}

The Fej\'er window removes the harmonic-shell logarithm.  We now show that
the remaining logarithm in the cosine-difference concentration is genuine
for the \emph{unweighted} spectrum.

\begin{lemma}[Lower logarithmic singularity of the arcsine difference]
\label{lem:difference-lower}
Let \(g\) be the density in
Lemma~\ref{lem:difference-density}.  For every \(0<z\le1/2\),
\begin{equation}
  g(z)
  \ge
  \frac1{2\sqrt2\,\pi^2}
  \log\frac1z.
  \label{eq:g-lower}
\end{equation}
\end{lemma}

\begin{proof}
Starting from Eq.~\eqref{eq:g-convolution}, substitute
\[
  y=1-x.
\]
Then
\[
  1-x^2=y(2-y),
\]
and
\[
  1-(x-z)^2
  =
  (y+z)(2-y-z).
\]
For \(0<z\le1/2\), restrict the integral to
\[
  z\le y\le1.
\]
On this interval,
\[
  2-y\le2,
  \qquad
  2-y-z\le2.
\]
Therefore
\[
  \sqrt{
  (1-x^2)(1-(x-z)^2)
  }
  \le
  2\sqrt{y(y+z)}.
\]
Also, because \(y\ge z\),
\[
  y+z\le2y,
\]
hence
\[
  \sqrt{y(y+z)}
  \le
  \sqrt2\,y.
\]
Thus
\[
  g(z)
  \ge
  \frac1{\pi^2}
  \int_z^1
  \frac{dy}{2\sqrt2\,y}
  =
  \frac1{2\sqrt2\,\pi^2}
  \log\frac1z.
\]
\end{proof}

\begin{corollary}[The continuum width-\(\delta\) mass is
\(\Theta(\delta\log(1/\delta))\)]
\label{cor:theta-log}
There exist absolute constants \(c,C>0\) such that for every
\(0<\delta\le1/2\),
\begin{equation}
  c\,\delta\log\frac1\delta
  \le
  \mathbb P(0\le Z\le\delta)
  \le
  C\,\delta\log\frac e\delta.
  \label{eq:theta-log}
\end{equation}
\end{corollary}

\begin{proof}
The upper bound is
Lemma~\ref{lem:continuum-concentration}.  For the lower bound, integrate
Eq.~\eqref{eq:g-lower}:
\begin{align*}
  \mathbb P(0\le Z\le\delta)
  &\ge
  \frac1{2\sqrt2\,\pi^2}
  \int_0^\delta\log\frac1z\,dz\\
  &=
  \frac1{2\sqrt2\,\pi^2}
  \left[
  \delta\log\frac1\delta+\delta
  \right]\\
  &\ge
  \frac1{2\sqrt2\,\pi^2}
  \delta\log\frac1\delta.
\end{align*}
\end{proof}

\begin{proposition}[Discrete cosine differences inherit the logarithm]
\label{prop:discrete-lower}
Let \(N_M(\delta)\) be defined by Eq.~\eqref{eq:NM}.  There is an absolute
constant \(c>0\) such that, for \(0<\delta\le1/2\),
\begin{equation}
  N_M(\delta)
  \ge
  cM^2\delta\log\frac1\delta
  -
  8M.
  \label{eq:discrete-lower}
\end{equation}
In particular, if
\[
  M\delta\log(1/\delta)\longrightarrow\infty,
\]
then
\[
  N_M(\delta)
  =
  \Theta\left(
  M^2\delta\log\frac1\delta
  \right)
\]
up to absolute constants.
\end{proposition}

\begin{proof}
Use \(I=[0,\delta)\) in
Lemma~\ref{lem:cosine-difference-discrepancy}.  The continuous difference
law has no atoms, so Corollary~\ref{cor:theta-log} gives
\(\mathbb P(E_1-E_2\in I)\ge
 c\delta\log(1/\delta)\).
Hence its discrete counterpart is at least this quantity minus \(8/M\).
Multiplication by \(M^2\) proves the asserted lower bound for the
half-open interval \(I\), which is one of the intervals in the
definition of \(N_M\).
If \(M\delta\log(1/\delta)\to\infty\), the \(8M\) error is smaller than
half this leading lower bound for all sufficiently large sizes.
The upper bound from Theorem~\ref{thm:cosine-difference} has the same
order: on \(0<\delta\le1/2\),
\(\log(4e/\delta)\le
[1+\log(4e)/\log2]\log(1/\delta)\).
This proves the two-sided asymptotic assertion.
\end{proof}

\begin{remark}[The obstruction is to a counting step, not to the actual dynamics]
\label{rem:last-log}
Proposition~\ref{prop:discrete-lower} prevents replacing the unweighted
cosine-difference count by \(O(M^2\delta)\) throughout the stated
regime.  It does \emph{not} prove a matching
\(\log T/T\) lower bound on subsystem relaxation.
A direct weighted collision estimate, finer matrix-element information,
or another dynamical argument might improve the clock rate without
improving the maximum unweighted interval count.
The strongest clock rate established in these notes remains
Eq.~\eqref{eq:LFSR-direct-T}; no \(O_s(n/T)\) clock theorem is claimed.
\end{remark}

\section{Structural clock extension, prime powers, and sector preparation}\label{sec:supp5}
\subsection{A structural periodic-clock theorem}

The LFSR proof uses two model-specific properties:
distinct regular Floquet phases on a cyclic grid and a uniform
\(Q^{-1/2}\)-scale bound on local matrix elements.  It therefore admits a
model-independent periodic-clock formulation.

Let the spin Hilbert space have dimension \(Q\), let the periodic clock
have length \(K\ge2\), and let
\[
  D=KQ.
\]
Choose a complete orthonormal Floquet eigenbasis and partition it
into regular and exceptional eigenvectors.  Assume that the regular
eigenvectors have pairwise distinct eigenphases, each belonging to
\[
  \left\{-\frac{2\pi q}{L}:q\in\mathbb Z_L\right\},
  \qquad 1\le L\le Q.
\]
Thus each phase-grid label is assigned to at most one regular
eigenvector.  The corresponding regular static-clock states have
distinct lifted labels in \(\mathbb Z_{KL}\), and their energies belong to
\begin{equation}
  \left\{
  1-\cos\frac{2\pi a}{KL}:
  a\in\mathbb Z_{KL}
  \right\}.
  \label{eq:general-clock-grid}
\end{equation}
Let \(r_F\) be the number of exceptional Floquet eigenvectors, so there are
at most
\[
  r_{\rm stat}=Kr_F
\]
exceptional static-clock eigenstates.

\begin{assumption}[Uniform regular local matrix-element bound]
\label{ass:general-clock-matrix}
For every normalized traceless local unitary basis element \(O_a\) on the
fixed spin subsystem and every pair \(j,k\) of regular static-clock
eigenstates,
\begin{equation}
  |\bra jO_a\ket k|
  \le c_Q.
  \label{eq:general-clock-cQ}
\end{equation}
\end{assumption}

\begin{theorem}[General periodic-clock anti-concentration]
\label{thm:general-clock}
Under the distinct-phase grid hypothesis above and
Assumption~\ref{ass:general-clock-matrix}, for every \(0<\delta\le1\),
\begin{equation}
  \chi_S(\delta;H_{\rm clock})
  \le
  \frac{2KL^2c_Q^2}{Q}
  \delta\log\frac{4e}{\delta}
  +
  \frac{8Lc_Q^2}{Q}
  +
  \frac{2r_F}{Q}.
\label{eq:general-clock-chi}
\end{equation}
In particular, if
\[
  c_Q\le\frac{C}{\sqrt Q},
  \qquad
  L\le Q,
  \qquad
  r_F\le r_0
\]
with constants \(C,r_0\) independent of \(Q,K\), then
\begin{equation}
  \chi_S(\delta;H_{\rm clock})
  \le
  2C^2K\delta\log\frac{4e}{\delta}
  +
  \frac{8C^2+2r_0}{Q}.
  \label{eq:general-clock-simple}
\end{equation}
\end{theorem}

\begin{proof}
If a regular Floquet vector has phase label \(q\), its clock
branches have lifted labels \(mL+q\), \(0\le m<K\).
Distinct regular vectors have distinct \(q\), so the map from regular
static states to these \(M=KL\) labels is injective.  Equal cosine
energies can still occur; they are already counted with their distinct
grid labels in Theorem~\ref{thm:cosine-difference}.  Hence
Theorem~\ref{thm:cosine-difference} gives, for every frequency interval
\(I\) of width \(\delta\),
\[
  N_{\mathcal R}(I)
  \le
  2K^2L^2\delta\log\frac{4e}{\delta}
  +
  8KL.
\]
For a normalized local unitary \(O_a\), the regular-pair contribution to
the local spectral weight is at most
\[
  \frac{N_{\mathcal R}(I)c_Q^2}{D}
  =
  \frac{N_{\mathcal R}(I)c_Q^2}{KQ}.
\]
There are at most \(Kr_F\) exceptional static eigenstates.  Unitarity of
\(O_a\) gives
\[
  \sum_k|\bra jO_a\ket k|^2=1.
\]
Therefore pairs with an exceptional first index contribute at most
\[
  \frac{Kr_F}{D}
  =
  \frac{r_F}{Q},
\]
and the same bound holds for exceptional second index.  Adding the three
contributions gives Eq.~\eqref{eq:general-clock-chi}.

If \(c_Q^2\le C^2/Q\) and \(L\le Q\), then
\[
  \frac{L^2c_Q^2}{Q}
  \le C^2,
  \qquad
  \frac{Lc_Q^2}{Q}
  \le\frac{C^2}{Q}.
\]
Substitution proves Eq.~\eqref{eq:general-clock-simple}.
\end{proof}

\begin{corollary}[General periodic-clock every-window relaxation]
\label{cor:general-clock-window}
Under the simplified hypotheses of
Theorem~\ref{thm:general-clock}, for every fixed subsystem \(S\) there is
a constant \(C_S\) such that, for every \(T\ge1\),
\begin{equation}
  \sup_{t_0\ge0}
  \frac1T
  \int_{t_0}^{t_0+T}
  \sup_{\rho_S}
  \Dt\left(
  \Phi_t(\rho_S),\frac{\Id_S}{d}
  \right)^2dt
  \le
  C_S
  \left[
  \frac{K\log(eT)}{T}
  +
  \frac1Q
  \right].
  \label{eq:general-clock-window}
\end{equation}
For every fixed \(\alpha>0\), if \(Q\) grows exponentially in \(K\),
target mean-square error \(O(K^{-\alpha})\) is therefore reached on
\[
  T_K=O_\alpha(K^{1+\alpha}\log K).
\]
\end{corollary}

\begin{proof}
Set \(\delta=T^{-1}\) in
Eq.~\eqref{eq:general-clock-simple} and apply
Theorem~\ref{thm:fejer-local-band}.  Since
\(T\delta=1\), the Fej\'er prefactor is an absolute constant, and all
remaining subsystem factors are fixed.  This proves
Eq.~\eqref{eq:general-clock-window}.  If
\(T_K=K^{1+\alpha}\log(eK)\), then
\[
  \log(eT_K)=O_\alpha(\log(eK))
\]
and therefore
\[
  \frac{K\log(eT_K)}{T_K}
  =
  O_\alpha(K^{-\alpha}).
\]
If \(Q\) grows exponentially in \(K\), the term \(Q^{-1}\) is smaller
than every inverse power of \(K\), completing the proof.
\end{proof}

\subsection{Prime-power moving-head family}

We now verify the structural hypotheses for every fixed prime-power local
dimension without referring back to a separate notes file.

Fix
\[
  q_0=p^r
\]
with \(p\) prime, take \(n\ge2\), and let
\[
  f(x)
  =
  x^n+f_{n-1}x^{n-1}+\cdots+f_1x+f_0
  \in\mathbb F_{q_0}[x]
\]
be primitive.  Put
\[
  Q=q_0^n,
  \qquad
  L=Q-1.
\]

\begin{lemma}[Companion map and local staircase]
\label{lem:companion-staircase}
On the computational basis
\(\ket{x_0,\ldots,x_{n-1}}\), define the two-site permutation gates
\begin{equation}
  g_i\ket{a,b}
  =
  \ket{b,a-f_i b},
  \qquad
  1\le i\le n-1,
  \label{eq:gi-prime-power}
\end{equation}
where \(g_i\) acts on sites \(i-1,i\), and define the one-site gate
\begin{equation}
  g_0\ket a
  =
  \ket{-f_0a}
  \label{eq:g0-prime-power}
\end{equation}
on site \(0\).  Then
\begin{equation}
  U_F
  :=
  g_0g_1\cdots g_{n-1},
  \label{eq:UF-staircase}
\end{equation}
with \(g_{n-1}\) applied first, implements the companion map
\begin{align}
  F(x_0,\ldots,x_{n-1})
  &=
  \bigl(
  -f_0x_{n-1},
  x_0-f_1x_{n-1},
  \ldots,
  x_{n-2}-f_{n-1}x_{n-1}
  \bigr).
  \label{eq:companion-map}
\end{align}
\end{lemma}

\begin{proof}
We prove by descending induction on \(k\) the effect of applying
\[
  g_k g_{k+1}\cdots g_{n-1}
\]
with \(g_{n-1}\) first.  After these gates, the entries on sites
\(0,\ldots,k-2\) are still
\[
  x_0,\ldots,x_{k-2},
\]
site \(k-1\) contains \(x_{n-1}\), and sites
\(j=k,\ldots,n-1\) contain
\[
  x_{j-1}-f_jx_{n-1}.
\]
For \(k=n-1\), this follows directly from
Eq.~\eqref{eq:gi-prime-power}.  Assume it holds for \(k\ge2\).
Gate \(g_{k-1}\) acts on sites \(k-2,k-1\), whose entries before that gate
are \(x_{k-2}\) and \(x_{n-1}\).  It maps them to
\[
  x_{n-1},
  \qquad
  x_{k-2}-f_{k-1}x_{n-1},
\]
and leaves all other entries unchanged.  This is exactly the induction
statement with \(k\) replaced by \(k-1\).

After the final two-site gate \(g_1\), the register therefore equals
\[
  \bigl(
  x_{n-1},
  x_0-f_1x_{n-1},
  \ldots,
  x_{n-2}-f_{n-1}x_{n-1}
  \bigr).
\]
Applying \(g_0\) to the first site gives
Eq.~\eqref{eq:companion-map}.
\end{proof}

\begin{lemma}[Maximal cycle and Floquet phases]
\label{lem:prime-power-cycle}
The map \(F\) fixes the zero vector and acts as one cycle of length
\(L=Q-1\) on
\(\mathbb F_{q_0}^n\setminus\{0\}\).  Consequently the permutation
unitary \(U_F\) has a two-dimensional \(+1\) eigenspace.  Its remaining
Floquet eigenphases may be labelled as
\begin{equation}
  \phi_q
  =
  -\frac{2\pi q}{L},
  \qquad
  q=1,\ldots,L-1.
  \label{eq:prime-power-phases}
\end{equation}
\end{lemma}

\begin{proof}
Let \(\alpha\) be a root of the primitive polynomial \(f\).  The vector
space \(\mathbb F_{q_0}^n\) is canonically identified with
\(\mathbb F_{q_0^n}\) in the basis
\[
  1,\alpha,\ldots,\alpha^{n-1}.
\]
The companion map is multiplication by \(\alpha\).  Because \(f\) is
primitive, \(\alpha\) generates the multiplicative group
\[
  \mathbb F_{q_0^n}^{\times},
\]
which has order \(Q-1=L\).  Thus \(0\) is fixed and all nonzero vectors
form one orbit of length \(L\).

Let \(\ket{x_t}\), \(t\in\mathbb Z_L\), list that orbit so that
\(U_F\ket{x_t}=\ket{x_{t+1}}\).  The Fourier vectors
\[
  \ket{\psi_q}
  =
  \frac1{\sqrt L}
  \sum_{t=0}^{L-1}
  e^{2\pi iqt/L}\ket{x_t},
  \qquad
  q=0,\ldots,L-1,
\]
satisfy
\[
  U_F\ket{\psi_q}
  =
  e^{-2\pi iq/L}\ket{\psi_q}.
\]
The \(q=0\) Fourier vector and the fixed vector \(\ket0\) are linearly
independent \(+1\) eigenvectors, and there are no others because the
cycle decomposition has exactly two cycles.  This proves the claim.
\end{proof}

\begin{lemma}[Trace pairing and additive characters]
\label{lem:character-trace-pairing}
Let \(q_0=p^r\), \(Q=q_0^n\), and let
\(b_0,\ldots,b_{n-1}\) be an \(\mathbb F_{q_0}\)-basis of
\(\mathbb F_Q\).  There is a unique dual basis \(b_i^\vee\) satisfying
\[
 \operatorname{Tr}_{\mathbb F_Q/\mathbb F_{q_0}}(b_i^\vee b_k)
 =\delta_{ik}.
\]
For \(x=\sum_k x_kb_k\) and \(\beta_w=\sum_iw_i b_i^\vee\),
\[
 \operatorname{Tr}_{\mathbb F_Q/\mathbb F_p}(\beta_wx)
 =\operatorname{Tr}_{\mathbb F_{q_0}/\mathbb F_p}
     \left(\sum_iw_ix_i\right).
\]
If \(w\ne0\), the resulting additive character of \(\mathbb F_Q\)
\[
 \psi_w(x)=\exp\!\left[
  \frac{2\pi i}{p}\operatorname{Tr}_{\mathbb F_Q/\mathbb F_p}(\beta_wx)
 \right]
\]
is nontrivial.
\end{lemma}
\begin{proof}
Write
\[
 T(x)=\operatorname{Tr}_{\mathbb F_Q/\mathbb F_{q_0}}(x)
     =\sum_{k=0}^{n-1}x^{q_0^k}.
\]
The identities \((x+y)^{q_0^k}=x^{q_0^k}+y^{q_0^k}\) and
\(a^{q_0^k}=a\) for \(a\in\mathbb F_{q_0}\) show that \(T\) is
\(\mathbb F_{q_0}\)-linear.  Also \(T(x)^{q_0}=T(x)\), because
\(x^{q_0^n}=x\), so its values belong to \(\mathbb F_{q_0}\).
The polynomial \(\sum_{k=0}^{n-1}X^{q_0^k}\) is nonzero and has degree
\(q_0^{n-1}<Q\).  A nonzero polynomial of degree \(d\) over a field has
at most \(d\) roots: each root gives a linear factor, and successive
factorization reduces the degree.  Thus \(T\) is not identically zero.
Choose \(y\) with \(T(y)\ne0\).  For any \(z\ne0\), the choice
\(x=y/z\) gives \(T(zx)=T(y)\ne0\).  The bilinear pairing
\((z,x)\mapsto T(zx)\) is therefore nondegenerate.  Its Gram matrix in
any basis is invertible, which gives existence and uniqueness of the
dual basis by solving the associated linear system.

Relative and absolute traces satisfy transitivity.  Explicitly,
\[
 \operatorname{Tr}_{\mathbb F_{q_0}/\mathbb F_p}(T(z))
 =\sum_{a=0}^{r-1}\sum_{k=0}^{n-1}z^{p^{a+rk}}
 =\sum_{j=0}^{rn-1}z^{p^j}
 =\operatorname{Tr}_{\mathbb F_Q/\mathbb F_p}(z).
\]
Applying this identity to \(z=\beta_wx\), the defining property of the
dual basis gives \(T(\beta_wx)=\sum_iw_ix_i\), proving the stated
formula.  The same polynomial argument applied to the absolute trace
shows that it is a nonzero \(\mathbb F_p\)-linear map into
\(\mathbb F_p\), hence surjective.  If \(w\ne0\), then
\(\beta_w\ne0\).  Choose \(z\) of absolute trace one and set
\(x=z/\beta_w\).  Then \(\psi_w(x)=e^{2\pi i/p}\ne1\), proving
nontriviality.  None of these steps requires \(p\nmid n\).
\end{proof}

\begin{lemma}[Gauss and Jacobi sums]
\label{lem:elementary-gauss-jacobi}
Let \(\mathbb F_Q\) be a finite field and \(\psi\) a nontrivial
additive character.  Extend every multiplicative character by zero at
zero.  Define
\[
 G(A)=\sum_{x\ne0}A(x)\psi(x),\qquad
 J(A,B)=\sum_{x\ne0,1}A(x)B(1-x).
\]
If \(A\) is nontrivial, \(|G(A)|=\sqrt Q\); if \(A=1\) on
\(\mathbb F_Q^\times\), then \(G(A)=-1\).
If \(A,B\) are nontrivial, then
\[
 |J(A,B)|=\sqrt Q\quad\hbox{when }AB\ne1,
 \qquad
 J(A,A^{-1})=-A(-1).
\]
\end{lemma}
\begin{proof}
A nontrivial character sums to zero over its finite group.  Indeed,
choose a group element on which it is not one; translation (for an
additive group) or multiplication (for a multiplicative group) preserves
the sum but multiplies it by a character value different from one.  The sum must
therefore vanish.  Hence
\(\sum_x\psi(x)=0\) and, for nontrivial \(A\),
\(\sum_{u\ne0}A(u)=0\).
For nontrivial \(A\), substitute \(x=uy\) in the squared Gauss sum:
\begin{align*}
 |G(A)|^2
 &=\sum_{x,y\ne0}A(x)\overline{A(y)}\psi(x-y)\\
 &=\sum_{u\ne0}A(u)\sum_{y\ne0}\psi((u-1)y)\\
 &=(Q-1)-\sum_{u\ne0,1}A(u)=Q.
\end{align*}
The inner sum is \(Q-1\) for \(u=1\) and \(-1\) otherwise, since
multiplication by \(u-1\ne0\) permutes \(\mathbb F_Q^\times\).
For the trivial multiplicative character, the Gauss sum is directly
\(\sum_{x\ne0}\psi(x)=-1\).

Suppose \(A,B,AB\) are nontrivial.  Expand \(G(A)G(B)\) and group the
terms by \(t=x+y\).  The \(t=0\) contribution is
\[
 \sum_{x\ne0}A(x)B(-x)
 =B(-1)\sum_{x\ne0}(AB)(x)=0.
\]
For \(t\ne0\), write \(x=tu\) and \(y=t(1-u)\).  This is a
bijection between the allowed pairs and \(t\ne0\), \(u\ne0,1\).
Consequently
\[
 G(A)G(B)
 =\sum_{t\ne0}(AB)(t)\psi(t)
       \sum_{u\ne0,1}A(u)B(1-u)
 =G(AB)J(A,B).
\]
Taking moduli and using the Gauss-sum result proves
\(|J(A,B)|=\sqrt Q\).
Finally, if \(B=A^{-1}\), the substitution
\(v=u/(1-u)\) is a bijection from
\(\mathbb F_Q\setminus\{0,1\}\) to
\(\mathbb F_Q^\times\setminus\{-1\}\), with inverse
\(u=v/(1+v)\).  Therefore
\[
 J(A,A^{-1})=\sum_{v\ne0,-1}A(v)=-A(-1).
\]
This also covers characteristic two.
\end{proof}

\begin{theorem}[Weyl matrix elements from finite-field character sums]
\label{thm:character-sum-external}
Let \(q_0=p^r\), \(Q=q_0^n\), and let
\(\ket{\psi_j}\), \(1\le j\le Q-2\), be the nontrivial Fourier
eigenvectors of Lemma~\ref{lem:prime-power-cycle}.  Use the finite-field
Weyl operators
\[
 X_v\ket{x}=\ket{x+v},\qquad
 Z_w\ket{x}=
 \exp\!\left[\frac{2\pi i}{p}
     \operatorname{Tr}_{\mathbb F_{q_0}/\mathbb F_p}(w\cdot x)
     \right]\ket{x},\qquad
 W_{v,w}=X_vZ_w,
\]
where \(x,v,w\in\mathbb F_{q_0}^n\).  For \((v,w)\ne(0,0)\),
\begin{equation}
 \left|\bra{\psi_j}W_{v,w}\ket{\psi_{j'}}\right|
 \le c_Q:=\frac{2\sqrt Q}{Q-1}
 \le\frac4{\sqrt Q}.
 \label{eq:character-sum-cQ}
\end{equation}
This holds for every prime-power \(q_0\), including characteristic two.
\end{theorem}

\begin{proof}
Identify \(\mathbb F_{q_0}^n\) with \(\mathbb F_Q\) using the basis
\(b_i=\alpha^i\), \(0\le i<n\), from
Lemma~\ref{lem:prime-power-cycle}.  Let \(b_i^\vee\) be the dual basis
for the trace pairing,
\[
 \operatorname{Tr}_{\mathbb F_Q/\mathbb F_{q_0}}
       (b_i^\vee b_k)=\delta_{ik}.
\]
For a vector \(w\), put \(\beta_w=\sum_iw_i b_i^\vee\).
The phase of \(Z_w\) becomes the additive character
\[
 \psi_w(x)=\exp\!\left[
       \frac{2\pi i}{p}
       \operatorname{Tr}_{\mathbb F_Q/\mathbb F_p}(\beta_wx)
       \right].
\]
Lemma~\ref{lem:character-trace-pairing} gives this identity and shows that \(\psi_w\) is nontrivial whenever \(w\ne0\).
In particular, this argument does not require \(p\nmid n\).

Let \(L=Q-1\) and choose the multiplicative character \(\chi\) with
\(\chi(\alpha)=e^{2\pi i/L}\).  Taking the first point of the nonzero
cycle to be \(1\), we have
\[
 \ket{\psi_j}=\frac1{\sqrt L}
       \sum_{x\in\mathbb F_Q^\times}\chi(x)^j\ket{x}.
\]
Changing the starting point of the cycle only changes the overall phase
of each eigenvector.  Write \(a\in\mathbb F_Q\) for the field element
corresponding to \(v\).  Direct substitution gives
\begin{equation}
 \bra{\psi_j}W_{v,w}\ket{\psi_{j'}}
 =\frac1L\sum_{\substack{x\in\mathbb F_Q\\x\ne0,-a}}
       \chi(x)^{j'}\chi(x+a)^{-j}\psi_w(x).
 \label{eq:prime-power-character-reduction}
\end{equation}

If \(a=0\), then \(w\ne0\).  The numerator is the Gauss sum of
\(\chi^{j'-j}\) and the nontrivial additive character \(\psi_w\).
It has modulus \(\sqrt Q\) when \(j\ne j'\), and equals \(-1\)
when \(j=j'\).

If \(a\ne0\) and \(w=0\), the change of variable \(x=-ay\) gives,
up to a factor of modulus one, the Jacobi sum
\[
 J(\chi^{j'},\chi^{-j})
 =\sum_{y\ne0,1}\chi(y)^{j'}\chi(1-y)^{-j}.
\]
Both multiplicative characters are nontrivial.  Its modulus is
\(\sqrt Q\) if \(j\ne j'\), and one if \(j=j'\).
The required Gauss and Jacobi identities were proved in
Lemma~\ref{lem:elementary-gauss-jacobi}.

It remains to consider \(a\ne0\) and \(w\ne0\).  We use the mixed
character-sum bound for a degree-one additive polynomial and a rational
multiplicative argument with two distinct finite zeros or poles
\cite[Theorem~1.1]{CochranePinner2006}.  To check its hypotheses
explicitly, set
\[
 h=\gcd(j,j'),\qquad
 \widetilde\chi=\chi^h,\qquad
 G(x)=\frac{x^{j'/h}}{(x+a)^{j/h}}.
\]
Extend multiplicative characters by zero at zero.  The character
\(\widetilde\chi\) is nontrivial, since
\(1\le h\le Q-2\).  The zero and pole of \(G\) are distinct and have
coprime multiplicities, so \(G\) is not a nontrivial rational-function
power, even up to multiplication by a constant.  The additive argument
is \(F(x)=x\); \(\beta_w F(x)\) has a simple pole at infinity and hence is not of
Artin--Schreier form \(R(x)^p-R(x)+c\).  In the notation of the cited
bound, \(\deg F=1\) and there are two distinct finite zeros or poles of
\(G\).  Therefore
\[
 \left|\sum_{x\ne0,-a}
       \widetilde\chi(G(x))\psi_w(x)\right|
 \le(1+2-1)\sqrt Q=2\sqrt Q.
\]
This is exactly the numerator of
Eq.~\eqref{eq:prime-power-character-reduction}.  Dividing the three
case bounds by \(L\) proves the first inequality in
Eq.~\eqref{eq:character-sum-cQ}.  The second follows from
\(Q/(Q-1)\le2\).
\end{proof}

\begin{remark}
The non-elementary input is the classical mixed character-sum bound;
the reduction from Weyl matrix elements to that bound is given above.
The operators here use the additive group of the finite field
\(\mathbb F_{q_0}\).  Invertible \(\mathbb F_{q_0}\)-linear
permutations preserve this Weyl basis under conjugation, as required
below.  For \(r>1\), this convention should be distinguished from the
cyclic shift-and-phase convention on \(\mathbb Z_{p^r}\).
\end{remark}

\begin{lemma}[Moving-head clock length and locality]
\label{lem:moving-head-locality}
Define the gate sequence
\begin{align}
  u_t
  &=
  g_{n-1-t},
  &&0\le t\le n-2,
  \nonumber\\
  u_{n-1}
  &=
  g_0,
  \nonumber\\
  u_t
  &=
  \Id,
  &&n\le t\le2n-2.
  \label{eq:moving-head-sequence}
\end{align}
Its period is
\[
  K=2n-1,
\]
and its Floquet unitary is \(U_F\).  The associated periodic
Feynman--Kitaev Hamiltonian admits a static geometrically local
realisation on a strip of bounded width, with bounded-range terms and
bounded interaction degree for fixed \(q_0\).
\end{lemma}

\begin{proof}
By Eq.~\eqref{eq:moving-head-sequence}, the chronological nonidentity
gates are
\[
  g_{n-1},g_{n-2},\ldots,g_1,g_0.
\]
The Floquet product, with later gates multiplying on the left, is
therefore
\[
  u_{K-1}\cdots u_0
  =
  g_0g_1\cdots g_{n-1}
  =
  U_F
\]
by Lemma~\ref{lem:companion-staircase}.  The number of time steps is
\[
  n+(n-1)=2n-1.
\]

For locality, place the \(n\) data sites on one horizontal row.  During
the first \(n\) clock steps, place consecutive clock sites on a second row
so that the hopping edge \((t,t+1)\) lies adjacent to the one- or
two-site support of \(u_t\).  The gates \(g_{n-1},\ldots,g_1\) move
successively along neighbouring data bonds and \(g_0\) acts on the end
site.  During the remaining \(n-1\) identity steps, return the clock path
on a third adjacent row, and place its final site adjacent to the initial
clock site.  Hence the periodic closing edge also has bounded geometric
range.  Consecutive clock sites are therefore a bounded distance apart,
every controlled hopping term
\[
  u_t\otimes\ket{t+1}\!\bra t+\mathrm{h.c.}
\]
has support of bounded diameter, and each data or clock site belongs to
only a bounded number of terms.  Since \(q_0\) is fixed, all local Hilbert
spaces have fixed dimension.  This gives the claimed bounded-width local
strip.
\end{proof}

\begin{proposition}[Prime-power static spectrum]
\label{prop:prime-power-static-spectrum}
For the \(K\)-step moving-head clock, every regular Floquet phase
Eq.~\eqref{eq:prime-power-phases} produces static energies
\begin{equation}
  E_{m,q}
  =
  1-
  \cos\left[
  \frac{2\pi}{KL}(mL+q)
  \right],
  \qquad
  0\le m<K,
  \quad
  1\le q<L.
  \label{eq:prime-power-static-spectrum}
\end{equation}
Thus the regular energies are a subset of the
\(KL\)-point cosine grid.
\end{proposition}

\begin{proof}
For a general \(K\)-step periodic clock and Floquet eigenvector
\[
  U_F\ket\phi=e^{i\phi}\ket\phi,
\]
define
\[
  U_{t:0}
  =
  u_{t-1}\cdots u_0,
  \qquad
  \ket{K_{t,\phi}}
  =
  U_{t:0}\ket\phi\otimes\ket t.
\]
The periodic boundary is
\[
  \ket{K_{K,\phi}}
  =
  e^{i\phi}\ket{K_{0,\phi}}.
\]
For \(m=0,\ldots,K-1\), choose
\[
  \theta_{m,\phi}
  =
  \frac{2\pi m}{K}-\frac{\phi}{K}
\]
and
\[
  \ket{\Psi_{m,\phi}}
  =
  \frac1{\sqrt K}
  \sum_{t=0}^{K-1}
  e^{i\theta_{m,\phi}t}
  \ket{K_{t,\phi}}.
\]
Exactly as in the direct calculation in
Proposition~\ref{prop:clock-spectrum}, the periodic-clock Hamiltonian acts
as the nearest-neighbour twisted Laplacian on this Krylov basis:
\[
  H_{\rm FK}\ket{\Psi_{m,\phi}}
  =
  (1-\cos\theta_{m,\phi})
  \ket{\Psi_{m,\phi}}.
\]
For
\[
  \phi=\phi_q=-\frac{2\pi q}{L},
\]
one obtains
\[
  \theta_{m,q}
  =
  \frac{2\pi}{KL}(mL+q),
\]
which proves Eq.~\eqref{eq:prime-power-static-spectrum}.
\end{proof}

\begin{lemma}[Prime-power static local matrix elements]
\label{lem:prime-power-static-matrix}
Let \(O_a\) be a nonidentity local Weyl basis operator on a fixed set of
data sites, extended trivially to the clock.  For every pair of regular
static-clock eigenstates,
\begin{equation}
  \abs{
  \bra{\Psi_{m,q}}
  O_a\otimes\Id_{\rm clock}
  \ket{\Psi_{m',q'}}
  }
  \le
  c_Q
  \label{eq:prime-power-static-matrix}
\end{equation}
with \(c_Q\) from
Eq.~\eqref{eq:character-sum-cQ}.
\end{lemma}

\begin{proof}
Clock orthogonality gives
\begin{align}
  &\bra{\Psi_{m,q}}
  O_a\otimes\Id_{\rm clock}
  \ket{\Psi_{m',q'}}
  \nonumber\\
  &\qquad=
  \frac1K
  \sum_{t=0}^{K-1}
  e^{i(\theta_{m',q'}-\theta_{m,q})t}
  \bra{\psi_q}
  U_{t:0}^\dagger O_aU_{t:0}
  \ket{\psi_{q'}}.
  \label{eq:prime-power-clock-average}
\end{align}
Every gate \(g_i\) is an invertible linear permutation over
\(\mathbb F_{q_0}\). If \(U_A\ket{x}=\ket{Ax}\), then
\begin{equation}
  U_A X_vZ_wU_A^\dagger
  =X_{Av}Z_{(A^{-1})^{T}w},
  \label{eq:linear-permutation-weyl}
\end{equation}
where the transpose is taken with respect to the coordinate pairing
used in the definition of \(Z_w\). Thus conjugation maps the Weyl basis
to itself, and an invertible conjugation cannot map a nonidentity
operator to the identity. Every \(U_{t:0}^\dagger O_aU_{t:0}\) is
therefore a nonidentity Weyl operator.
Theorem~\ref{thm:character-sum-external} bounds every summand in
Eq.~\eqref{eq:prime-power-clock-average} by \(c_Q\) in modulus.  The
triangle inequality proves the result.
\end{proof}

\begin{corollary}[Prime-power local bands and time windows]
\label{cor:prime-power-improved-full}
For every fixed prime-power local dimension \(q_0\) and every fixed
data subsystem \(S\), the following bounds hold for
\(0<\delta\le1\) and \(T\ge1\), respectively:
\begin{align}
  \chi_S(\delta;H_{\rm FK}^{(q_0)})
  &=
  O_{q_0}\left(
  K\delta\log\frac e\delta
  \right)
  +
  O_{q_0}(q_0^{-n}),
  \label{eq:prime-power-chi-full}\\
  \sup_{t_0\ge0}
  \frac1T
  \int_{t_0}^{t_0+T}
  \sup_{\rho_S}
  \Dt\left(
  \Phi_t(\rho_S),\frac{\Id_S}{d}
  \right)^2dt
  &=
  O_{S,q_0}\left(
  \frac{K\log(eT)}{T}
  \right)
  +
  O_{S,q_0}(q_0^{-n}).
  \label{eq:prime-power-window-full}
\end{align}
Consequently,
\begin{equation}
  T_K
  =
  O_{\alpha,q_0}
  \left(
  K^{1+\alpha}\log K
  \right)
  \label{eq:prime-power-T-full}
\end{equation}
is sufficient for mean-square error \(O(K^{-\alpha})\).
\end{corollary}

\begin{proof}
By Lemma~\ref{lem:prime-power-cycle}, the exceptional Floquet dimension is
\[
  r_F=2.
\]
By Proposition~\ref{prop:prime-power-static-spectrum}, the regular static
energies have lifted grid labels with \(L=Q-1\le Q\).
The regular Floquet labels \(1,\ldots,Q-2\) are pairwise distinct,
so the injectivity hypothesis of Theorem~\ref{thm:general-clock} holds.  By
Lemma~\ref{lem:prime-power-static-matrix} and
Theorem~\ref{thm:character-sum-external},
\[
  c_Q=O_{q_0}(Q^{-1/2}).
\]
Thus the simplified form of
Theorem~\ref{thm:general-clock} applies and gives
Eq.~\eqref{eq:prime-power-chi-full}, because
\[
  Q^{-1}=q_0^{-n}.
\]
Equation~\eqref{eq:prime-power-window-full} follows by setting
\(\delta=T^{-1}\) in the local-band theorem.  Finally, for
\[
  T_K=K^{1+\alpha}\log(eK),
\]
one has
\[
  \frac{K\log(eT_K)}{T_K}
  =
  O_\alpha(K^{-\alpha}),
\]
while \(q_0^{-n}\) is exponentially small in
\(K=2n-1\).  This proves
Eq.~\eqref{eq:prime-power-T-full}.
\end{proof}

\subsection{Clock-sector factorisation, energy consistency, and operational typicality}

\begin{proposition}[Spin-local energy orthogonality]
\label{prop:FK-energy-orthogonal-new}
For any periodic Feynman--Kitaev Hamiltonian of the form
Eq.~\eqref{eq:FK} with at least two clock positions, and every
spin-only subsystem \(S\),
\begin{equation}
  \mathsf P_S
  \left(
  H_{\rm FK}
  -
  \frac{\Tr H_{\rm FK}}{D}\Id
  \right)
  =
  0.
  \label{eq:FK-energy-orthogonal-new}
\end{equation}
Thus the conserved total-energy direction is exactly orthogonal to the
traceless spin-local operator algebra.
\end{proposition}

\begin{proof}
Let the bath contain all spins outside \(S\) and the full one-particle
clock space.  Every clock hopping term contains
\(\ket{t+1}\!\bra t\) or its adjoint, whose trace on the clock space is
zero.  Therefore the partial trace of all hopping terms over the clock
vanishes.  The identity term gives
\[
  \Tr_BH_{\rm FK}=m\Id_S.
\]
For the same reason every hopping term has zero full trace, so
\[
  \Tr H_{\rm FK}=D.
\]
The explicit projection formula from
Lemma~\ref{lem:local-projection} is
\[
  \mathsf P_S(X)
  =
  \left[
  \frac1m\Tr_BX
  -
  \frac{\Tr X}{D}\Id_S
  \right]\otimes\Id_B.
\]
Apply it to \(X=H_{\rm FK}\) and substitute the two trace identities.
\end{proof}

\begin{proposition}[One-clock-particle factorisation]
\label{prop:clock-factorisation-new}
In the physical second-quantized clock realisation,
\[
  N_c=\sum_t c_t^\dagger c_t
\]
commutes with \(H_{\rm FK}\).  In the invariant sector \(N_c=1\), a
spin-only subsystem factors as
\begin{equation}
  \mathcal H^{(1)}
  =
  \mathcal H_S
  \otimes
  \left(
  \mathcal H_{B_s}
  \otimes
  \mathcal H_{\rm clock}^{(1)}
  \right).
  \label{eq:clock-factorisation-new}
\end{equation}
The bath state used in all preceding theorems is therefore exactly
\begin{equation}
  \tau_B
  =
  \frac{\Id_{B_s}}{\dim\mathcal H_{B_s}}
  \otimes
  \frac{\Id_{\rm clock}^{(1)}}{K},
  \label{eq:clock-bath-new}
\end{equation}
where \(K\) is the clock length.  Every product input
\(\rho_S\otimes\tau_B\) has the same mean energy as the maximally mixed
state in the sector.
\end{proposition}

\begin{proof}
Each hopping term in the second-quantized clock contains one clock
creation and one clock annihilation operator, so it preserves clock
particle number.  The identity term also preserves it.  Hence
\[
  [H_{\rm FK},N_c]=0.
\]
In the sector containing exactly one clock particle, the clock Hilbert
space is the \(K\)-dimensional span of the one-particle position states
\(\{\ket t\}_{t=0}^{K-1}\).  If \(S\) contains only data spins, the
remaining data spins and the entire one-particle clock belong to the
bath, which proves Eq.~\eqref{eq:clock-factorisation-new} and the
maximally mixed bath formula
Eq.~\eqref{eq:clock-bath-new}.

Finally,
\[
  \Tr[
  H_{\rm FK}(\rho_S\otimes\tau_B)]
  =
  \Tr\left[
  \rho_S
  \frac1m\Tr_BH_{\rm FK}
  \right].
\]
The proof of
Proposition~\ref{prop:FK-energy-orthogonal-new} gives
\[
  \frac1m\Tr_BH_{\rm FK}=\Id_S.
\]
Since \(\Tr\rho_S=1\), the mean energy equals \(1\) for every local
input.  The maximally mixed state in the one-particle sector has the same
mean energy because \(\Tr H_{\rm FK}/D=1\).
\end{proof}

\begin{remark}[Why the subsystem is spin-only]
If clock sites themselves are split between system and bath, the fixed
one-clock-particle sector is a direct sum of local charge sectors rather
than a tensor product across that spatial cut.  The preparation
\(\rho_S\otimes\tau_B\) used throughout these notes then does not have the
assumed tensor-product form.  All explicit clock theorems therefore refer
to fixed spin-only subsystems.
\end{remark}

\begin{lemma}[Bad-time fraction]
\label{lem:bad-time-fraction-new}
Let \(f(t)\ge0\).  If
\[
  \frac1T
  \int_{t_0}^{t_0+T}f(t)^2dt
  \le E_T,
\]
then for every \(\eta>0\),
\begin{equation}
  \frac1T
  \meas
  \left\{
  t\in[t_0,t_0+T]:
  f(t)\ge\eta
  \right\}
  \le
  \frac{E_T}{\eta^2}.
  \label{eq:bad-time-fraction-new}
\end{equation}
\end{lemma}

\begin{proof}
On the set where \(f(t)\ge\eta\), one has
\(f(t)^2\ge\eta^2\).  Therefore
\[
  E_T
  \ge
  \frac{\eta^2}{T}
  \meas\{t\in[t_0,t_0+T]:f(t)\ge\eta\}.
\]
Rearrange.
\end{proof}

\begin{corollary}[Quasi-linear fixed-accuracy windows]
\label{cor:quasilinear-fixed-accuracy-new}
For the binary LFSR clock, fix a trace-distance tolerance
\(\eta\in(0,1)\) and a tolerated bad-time fraction
\(\zeta\in(0,1)\).  For all sufficiently large \(n\),
\begin{equation}
  T
  =
  O_S\left(
  \frac{n}{\eta^2\zeta}
  \log\frac{en}{\eta^2\zeta}
  \right)
  \label{eq:quasilinear-window-new}
\end{equation}
is sufficient to guarantee, uniformly in the window starting time,
that the same good-time set works for all local inputs:
\begin{equation}
  \frac1T
  \meas
  \left\{
  t\in[t_0,t_0+T]:
  \sup_{\rho_S}
  \Dt\left(
  \Phi_t(\rho_S),\frac{\Id_S}{2^s}
  \right)
  \ge\eta
  \right\}
  \le\zeta.
  \label{eq:quasilinear-bad-new}
\end{equation}
For fixed \(\eta,\zeta\), this is \(T=O_S(n\log n)\).
\end{corollary}

\begin{proof}
Theorem~\ref{thm:LFSR-direct-T} gives
\[
  E_T
  \le
  C_S
  \left[
  \frac{n\log(eT)}{T}
  +
  2^{-n}
  \right].
\]
Choose
\[
  T
  =
  C\frac{n}{\eta^2\zeta}
  \log\frac{en}{\eta^2\zeta}
\]
with \(C\) large enough, depending only on \(S\), that the first term is
at most \(\eta^2\zeta/2\).  For sufficiently large \(n\), the
exponential term is also at most \(\eta^2\zeta/2\).  Hence
\(E_T\le\eta^2\zeta\). Apply
Lemma~\ref{lem:bad-time-fraction-new} with
\[
 f(t)=\sup_{\rho_S}
 \Dt\left(\Phi_t(\rho_S),\frac{\Id_S}{2^s}\right).
\]
\end{proof}

\begin{remark}
The physical spin--clock Hamiltonians are static and geometrically local,
but their one-clock-particle sectors are exactly solvable and decompose
into Floquet-labelled tight-binding sectors.  The explicit clock theorems
are therefore not generic interacting-Hamiltonian theorems.  Their role is
to demonstrate rigorously that locally accessible spectral
anti-concentration can coexist with exponentially fine microscopic
spectral structure.
\end{remark}

\section{The structured random commuting ensemble}\label{sec:supp6}
\subsection{Exact ensemble specification and locality}

Take an even integer \(M\ge2\), an integer \(\xi\ge1\), and
\[
 N=M\xi,\qquad Q=2^\xi,\qquad D=2^N.
\]
Partition the qubit chain into \(M\) consecutive patches of exactly
\(\xi\) qubits.  We consider sizes \(N=M\xi\), so every patch has
the same dimension.  On patch
\(i\), define
\[
 D_i(J)=\sum_{z=1}^{Q}J_{i,z}\ket z\!\bra z_i,\qquad -1\le J_{i,z}\le1.
\]
Let
\[
 V=V_{\mathrm{odd}}V_{\mathrm{even}},\qquad
 V_{\mathrm{odd}}=\bigotimes_{\substack{i=1\\i\ {\rm odd}}}^{M-1}V_{i,i+1},
 \qquad
 V_{\mathrm{even}}=\bigotimes_{\substack{i=2\\i\ {\rm even}}}^{M-2}V_{i,i+1}.
\]
The two-patch gates are mutually independent Haar unitaries; an empty
product is the identity.  They are independent of all patch energies.
Define one time-independent Hamiltonian by
\begin{equation}
 H_N(J,V)=V^\dagger\left(\sum_{i=1}^{M}D_i(J)\right)V
         =\sum_{i=1}^{M}h_i(J,V).
 \label{eq:nearly-local-H}
\end{equation}
Our reference energy law is independent uniform \(J_{i,z}\) on
\([-1,1]\).  Later we allow independent, nonuniform bounded energy
densities, without changing the law of \(V\).

\begin{proposition}[Commutativity, bounded term incidence, and conserved operators]
\label{prop:structured-scope}
The terms \(h_i=V^\dagger D_iV\) have norm at most one and commute
pairwise.  Each is supported on at most five adjacent patches.  A fixed
physical qubit belongs to at most five such term supports.
Every operator
\[
 \Pi_{i,z}=V^\dagger
       \left(\ket z\!\bra z_i\otimes\Id_{\bar i}\right)V
\]
commutes with \(H_N\).  Thus the model has fixed on-site qubit dimension,
\(O(\xi)\)-body interactions of range \(O(\xi)\), and an extensive
commuting algebra of conserved operators.
\end{proposition}
\begin{proof}
Unitary conjugation preserves operator norm, and
\(\|D_i\|_\infty\le1\).  Since \(D_i,D_j\) act on disjoint tensor
factors,
\[
 [h_i,h_j]=V^\dagger[D_i,D_j]V=0.
\]
Conjugation by one matching layer can enlarge the support of an operator
only by adding patches paired to its present support.  After two layers,
an operator initially on patch \(i\) has support within
\(\{i-2,i-1,i,i+1,i+2\}\cap\{1,\ldots,M\}\).
Conversely, patch \(j\) can belong to the support of \(h_i\) only if
\(|i-j|\le2\); at most five indices satisfy this.  The same conjugation
argument applies to \(\Pi_{i,z}\), and its unconjugated form commutes with
every \(D_j\).  More explicitly, the \(N\) operators
\(Q_j=V^\dagger Z_jV\), where \(Z_j\) is the Pauli operator at
physical site \(j\), commute with one another and with \(H_N\).
Unitary invariance gives
\(D^{-1}\Tr(Q_jQ_k)=D^{-1}\Tr(Z_jZ_k)=\delta_{jk}\),
so there are \(N\) independent conserved directions.  This proves
the stated extensivity.
\end{proof}

\begin{remark}[Support of evolved observables]
``Bounded incidence'' here counts Hamiltonian terms containing a qubit;
it does not assert a bounded number of other qubits in each interaction.
For a local observable \(O\), terms \(h_i\) initially disjoint from its
support commute with \(O\) and with all other \(h_j\).  In the product
\(e^{itH}Oe^{-itH}\), their two exponential factors can therefore be
brought next to \(O\) and canceled.  Only terms initially overlapping
\(O\) contribute.  For a fixed-size subsystem their union has
\(O(\xi)\) qubits, independently of time.  Local relaxation in this
model does not establish transport through an extended noncommuting
chain.
\end{remark}

\subsection{Reference-time design estimate}

For a unitary ensemble \(\mathcal U\) on dimension \(D\), write
\[
 \mathcal M_{\mathcal U}^{(2)}(X)
 =\mathbb E_{U\sim\mathcal U}
     U^{\otimes2}XU^{\dagger\otimes2}.
\]
An additive two-design error \(\varepsilon\) means
\(\|\mathcal M_{\mathcal U}^{(2)}-
\mathcal M_{\rm Haar}^{(2)}\|_\diamond\le\varepsilon\).
The diamond norm is the supremum of the induced trace norm after
tensoring the map with an identity map on an arbitrary finite ancilla.
In particular it bounds the induced trace norm with no ancilla.

\begin{theorem}[Reference-time design input: exact-Haar specialization]
\label{thm:Cui-design-input}
There is a universal constant \(C_0\) such that, for the ensemble
Eq.~\eqref{eq:nearly-local-H} with the reference uniform energy law,
\begin{equation}
 \xi\ge2\log_2(2N/\varepsilon)+C_0
 \quad\Longrightarrow\quad
 \left\|\mathcal M_{\{e^{-i\pi H_N}\}}^{(2)}
       -\mathcal M_{\rm Haar}^{(2)}\right\|_\diamond\le\varepsilon.
 \label{eq:reference-design-condition}
\end{equation}
This is the \(k=2\), exact-Haar-gate specialization of Theorem 3 of
Ref.~\cite{CuiSchusterMaoHuangBrandao2025}. At time \(\pi\),
independent uniform energies on \([-1,1]\) give independent uniform
phases on the unit circle. A non-adaptive two-copy experiment is a
special case of the \(k=2\) experiments allowed in the measurable-error
definition. The cited theorem therefore also bounds the additive
two-copy moment map used below.
\end{theorem}

This input is used only for the reference distribution at time \(\pi\).
No all-late-time global design theorem is assumed or claimed in the
argument below.  Exact Haar gates also avoid confusing a usual two-design
with the stronger local-gate design notion required in the source.

\begin{lemma}[Admissible logarithmic patch sizes]
\label{lem:admissible-patches}
For every fixed \(b>0\), there are arbitrarily large system sizes
\(N=M\xi\), with \(M\) even, \(\xi=O_b(\log N)\), satisfying
Eq.~\eqref{eq:reference-design-condition} with
\(\varepsilon=N^{-b}\).
\end{lemma}
\begin{proof}
Let \(M\) run over sufficiently large even integers and choose
\[
 \xi=\left\lceil4(1+b)\log_2 M+C_0+2\right\rceil,\qquad N=M\xi.
\]
For all sufficiently large \(M\), \(\xi\le M\): the ratio
\(\log M/M\) tends to zero, as follows for example from
\(\log x=\int_1^x u^{-1}du\le2\sqrt x\) for \(x\ge1\).
It follows that \(\log_2N\le2\log_2M\).  Hence
\[
 2\log_2(2N/N^{-b})+C_0
 =2+2(1+b)\log_2N+C_0
 \le 4(1+b)\log_2M+C_0+2
 \le\xi.
\]
Also \(\xi=O_b(\log M)=O_b(\log N)\), and \(N\) is unbounded.
\end{proof}

\subsection{Two-copy moments}

\begin{lemma}[Haar two-copy twirl, with its elementary commutant proof]
\label{lem:haar-two-copy-audited}
Let \(F\) swap the two copies of \(\mathbb C^D\), where \(D\ge2\).
For every operator \(X\),
\begin{equation}
 \mathcal M_{\rm Haar}^{(2)}(X)
 =\frac{D\Tr X-\Tr(FX)}{D(D^2-1)}\,\Id
  +\frac{D\Tr(FX)-\Tr X}{D(D^2-1)}\,F.
 \label{eq:haar-two-copy-audited}
\end{equation}
\end{lemma}
\begin{proof}
Haar invariance implies that the twirled operator commutes with
\(U\otimes U\) for every unitary \(U\).
We verify that this commutant is spanned by \(\Id,F\).
Commutation with all diagonal phase matrices implies that its matrix
element \(\langle ij|Z|kl\rangle\) vanishes unless the multisets
\(\{i,j\}\) and \(\{k,l\}\) coincide.  To see this, vary each diagonal
phase independently; any nonzero difference of its two multiplicities
would change that matrix element by a nontrivial phase.
Commutation with all permutation matrices then makes all coefficients
of each of the three remaining types equal.  Consequently
\[
 Z=b\Id+cF+(a-b-c)\sum_i\ket{ii}\!\bra{ii}.
\]
A unitary taking \(\ket1\) to
\(\ket+=(\ket1+\ket2)/\sqrt2\) must take the eigenvector \(\ket{11}\),
of eigenvalue \(a\), to an eigenvector \(\ket{++}\) of the same
eigenvalue.  The coefficient of \(\ket{12}\) in \(Z\ket{++}\) is
\((b+c)/2\), while that in \(a\ket{++}\) is \(a/2\).
Thus \(a=b+c\), proving the commutant assertion.

Write the twirl as \(\alpha\Id+\beta F\).  Unitary invariance of trace,
and \(F(U\otimes U)=(U\otimes U)F\), give
\[
 D^2\alpha+D\beta=\Tr X,\qquad
 D\alpha+D^2\beta=\Tr(FX).
\]
Solving these two equations gives Eq.~\eqref{eq:haar-two-copy-audited}.
\end{proof}

\begin{lemma}[Scalar autocorrelation at the reference time]
\label{lem:design-suppresses-phi}
Let \(S\) consist of \(s\ge1\) qubits, \(d=2^s\), \(q=d^2-1\),
and let \(O_1,\ldots,O_q\) be its embedded nonidentity Pauli basis.
Define
\[
 \varphi_S(U)=\frac1q\sum_a c_a(U),\qquad
 c_a(U)=D^{-1}\Tr(O_aUO_aU^\dagger).
\]
If \(\mathcal U\) has additive two-design error at most \(\varepsilon\),
then
\begin{align}
 \mathbb E_{\mathcal U}|\varphi_S(U)|^2
 &\le \frac1{q(D^2-1)}+\varepsilon,\label{eq:design-suppresses-phi}\\
 \left|\mathbb E_{\mathcal U}\varphi_S(U)\right|
 &\le\varepsilon.\label{eq:design-controls-first}
\end{align}
\end{lemma}
\begin{proof}
Every \(c_a(U)\) is real because conjugating its trace and cycling the
two Hermitian factors leaves it unchanged.  Therefore the correct
two-copy representation is
\begin{equation}
 |\varphi_S(U)|^2
 =\frac1{q^2D^2}\sum_{a,b}
   \Tr\!\left[
    (O_a\otimes O_b)U^{\otimes2}
    (O_a\otimes O_b)U^{\dagger\otimes2}\right].
 \label{eq:correct-two-copy-contraction}
\end{equation}
For each pair \(a,b\),
\(\|O_a\otimes O_b\|_\infty=1\) and
\(\|O_a\otimes O_b\|_1=D^2\).
The two-design norm bound therefore changes the expectation of each
normalized summand by at most \(\varepsilon\).  Indeed,
\(|\Tr XY|\le\|X\|_\infty\|Y\|_1\): writing a singular-value
decomposition \(Y=\sum_j s_j\ket{u_j}\!\bra{v_j}\) bounds the trace by
\(\sum_j s_j|\langle v_j|X|u_j\rangle|
 \le\|X\|_\infty\sum_j s_j\).
Averaging over \(q^2\) pairs preserves the same error bound.

For the Haar value, \(\Tr(O_a\otimes O_b)=0\) and
\(\Tr[F(O_a\otimes O_b)]=\Tr(O_aO_b)=D\delta_{ab}\).
Using Eq.~\eqref{eq:haar-two-copy-audited} and contracting again with
\(O_a\otimes O_b\) gives
\[
 \mathbb E_{\rm Haar}c_a(U)c_b(U)=\frac{\delta_{ab}}{D^2-1}.
\]
Summation over \(a,b\) gives the Haar value \(1/[q(D^2-1)]\).

For the first moment, use instead
\[
 c_a(U)=D^{-2}\Tr\!\left[
 (O_a\otimes\Id)U^{\otimes2}(O_a\otimes\Id)U^{\dagger\otimes2}\right].
\]
The same norm estimate bounds its design error by \(\varepsilon\).
Both \(\Tr(O_a\otimes\Id)\) and
\(\Tr[F(O_a\otimes\Id)]\) vanish, so its Haar expectation is zero.
Averaging over \(a\) proves Eq.~\eqref{eq:design-controls-first}.
\end{proof}

\subsection{Structural Bohr-frequency labels and the collision plateau}

Write \(\mathbf z=(z_1,\ldots,z_M)\) for a patch configuration.
The eigenvectors are \(V^\dagger\ket{\mathbf z}\), independent of \(J\),
with energies \(E_{\mathbf z}(J)=\sum_iJ_{i,z_i}\).
Define the nonnegative pair weights
\begin{equation}
 w_{\mathbf z\mathbf w}(V)
 =\frac1{qD}\sum_a
       |\langle\mathbf z|VO_aV^\dagger|\mathbf w\rangle|^2.
 \label{eq:structured-pair-weights}
\end{equation}
For every ordered pair define the integer vector, indexed by \((i,u)\),
\[
 \ell^{\mathbf z\mathbf w}_{i,u}
 =\mathbbm1_{\{u=z_i\}}-\mathbbm1_{\{u=w_i\}}.
\]
Let \(\mathcal I\) be the finite set of these vectors and group their
weights:
\begin{equation}
 W_\ell(V)=\sum_{\ell^{\mathbf z\mathbf w}=\ell}
                 w_{\mathbf z\mathbf w}(V),\qquad
 \beta(V)=\sum_{\ell\in\mathcal I}W_\ell(V)^2.
 \label{eq:structural-collision}
\end{equation}
The dot product \(\ell\cdot J\) sums over all indices \((i,u)\).

\begin{lemma}[Positive spectral representation and almost-sure frequency identification]
\label{lem:structural-frequencies}
The weights satisfy \(W_\ell\ge0\) and \(\sum_\ell W_\ell=1\), and
\begin{equation}
 \varphi_S(t;J,V)=\sum_{\ell\in\mathcal I}
           W_\ell(V)e^{-it\ell\cdot J}.
 \label{eq:structured-phi}
\end{equation}
If the \(J_{i,u}\) are independent and absolutely continuous, then outside
a probability-zero set, all frequencies with different vectors \(\ell\)
are distinct, and only \(\ell=0\) has frequency zero.  On this event
\begin{align}
 \pi_{0,S}(J,V)&=W_0(V),\\
 \widetilde\varphi_S(t;J,V)
     &=\sum_{\ell\ne0}W_\ell(V)e^{-it\ell\cdot J},\\
 \overline{|\varphi_S|^2}&=\beta(V).
 \label{eq:structural-zero-and-plateau}
\end{align}
No assumption of nondegenerate ordered energy gaps is made.
\end{lemma}
\begin{proof}
For each \(a\), summing squared matrix entries of \(VO_aV^\dagger\)
gives \(\Tr(O_a^\dagger O_a)=D\).  Thus the pair weights sum to one.
Expanding the normalized trace in the eigenbasis gives the pair
representation of \(\varphi_S\), and grouping pairs with equal vectors
gives Eq.~\eqref{eq:structured-phi}.

For \(\ell\ne\ell'\), choose a coordinate on which their difference is
nonzero.  Conditional on all other energies,
\((\ell-\ell')\cdot J=0\) fixes at most one value of that energy.
Absolute continuity makes its conditional probability zero.
There are finitely many pairs of vectors, so the union of these events
also has probability zero.  The same argument with \(\ell'=0\)
excludes additional zero frequencies.
The temporal identities now follow from
Lemma~\ref{lem:Cesaro-orthogonality}.  Distinct ordered pairs can still
have the same vector \(\ell\); their structural resonances have already
been included in \(W_\ell\).
\end{proof}

\begin{lemma}[Reference phases measure the collision weight exactly]
\label{lem:reference-collision}
For fixed \(V\), let the superscript ``ref'' denote expectation over the
independent uniform energy law.  Then
\begin{align}
 \mathbb E_J^{\rm ref}|\varphi_S(\pi;J,V)|^2&=\beta(V),\\
 \mathbb E_J^{\rm ref}\varphi_S(\pi;J,V)&=W_0(V).
 \label{eq:reference-collision-identities}
\end{align}
Consequently, under Eq.~\eqref{eq:reference-design-condition},
\begin{equation}
 \mathbb E_V\beta(V)\le b_N,\qquad
 \mathbb E_VW_0(V)\le\varepsilon,\qquad
 b_N:=\varepsilon+\frac1{q(D^2-1)}.
 \label{eq:plateau-design-bound}
\end{equation}
These conclusions about \(V\) remain valid when the actual energy
distribution is changed.
\end{lemma}
\begin{proof}
For an integer \(k\) and \(J\) uniform on \([-1,1]\),
\[
 \mathbb E e^{-itkJ}=
 \begin{cases}\sin(kt)/(kt),&kt\ne0,\\1,&kt=0.\end{cases}
\]
At \(t=\pi\) this equals zero for nonzero \(k\).
Independence applied to Eq.~\eqref{eq:structured-phi} leaves just
\(\ell=0\) in the first moment.  In its squared modulus it leaves just
\(\ell=\ell'\), because all vector differences are integer.
This proves both identities.  Average over \(V\), apply
Theorem~\ref{thm:Cui-design-input} and
Lemma~\ref{lem:design-suppresses-phi}, and note \(W_0\ge0\).
The last assertion holds because \(W_\ell,\beta\) are functions of \(V\)
alone, not of any energy variables.
\end{proof}

\subsection{Dimension-free scalar dephasing, including nonuniform spectra}

\begin{lemma}[Fourier estimate for bounded regular energy densities]
\label{lem:bounded-density-fourier}
Suppose a random energy has a probability density \(p\) on \([-1,1]\)
which is absolutely continuous there, with derivative in \(L^1\).
If
\[
 |p(-1)|+|p(1)|+\int_{-1}^{1}|p'(x)|dx\le L_0,
\]
then its characteristic function \(g(u)=\int_{-1}^{1}p(x)e^{-iux}dx\)
satisfies
\[
 |g(u)|\le\min\{1,L_0/|u|\}\quad(u\ne0).
\]
The uniform density has \(L_0=1\).
\end{lemma}
\begin{proof}
The bound by one follows from nonnegativity and normalization of \(p\).
For \(u\ne0\), integration by parts gives
\[
 g(u)=\left[\frac{p(x)e^{-iux}}{-iu}\right]_{-1}^{1}
       +\frac1{iu}\int_{-1}^{1}p'(x)e^{-iux}dx.
\]
Taking absolute values proves the stated estimate.
For \(p=1/2\), the boundary terms sum to one and the derivative vanishes.
\end{proof}

\begin{theorem}[Scalar dephasing with no Hilbert-space or patch-count prefactor]
\label{thm:scalar-dephasing}
Let all energies be independent, independent of \(V\), and have densities
satisfying Lemma~\ref{lem:bounded-density-fourier} with a common fixed
\(L_0\ge1\).  Put
\[
 r_{L_0}(t)=
 \begin{cases}
 1,&|t|\le L_0,\\
 L_0^2/t^2,&|t|>L_0.
 \end{cases}
\]
For each fixed \(V\),
\begin{align}
 \left|\mathbb E_J|\varphi_S(t)|^2-\beta(V)\right|
 &\le r_{L_0}(t)\,[1-\beta(V)],\\
 \left|\mathbb E_J|\widetilde\varphi_S(t)|^2
          -[\beta(V)-W_0(V)^2]\right|
 &\le r_{L_0}(t).
 \label{eq:scalar-dephasing}
\end{align}
In particular, under the reference-design hypotheses,
\begin{equation}
 \mathbb E|\varphi_S(t)|^2\le b_N+r_{L_0}(t),\qquad
 \mathbb E|\widetilde\varphi_S(t)|^2\le b_N+r_{L_0}(t).
 \label{eq:scalar-dephasing-averaged}
\end{equation}
For the uniform law this gives a remainder \(\min\{1,t^{-2}\}\).
\end{theorem}
\begin{proof}
Expanding Eq.~\eqref{eq:structured-phi} gives
\[
 \mathbb E_J|\varphi_S(t)|^2
 =\sum_{\ell,\ell'}W_\ell W_{\ell'}
      \prod_{i,u}g_{i,u}\bigl(t(\ell_{i,u}-\ell'_{i,u})\bigr).
\]
For \(\ell=\ell'\), the product equals one and the corresponding sum is
\(\beta(V)\).  If \(\ell\ne\ell'\), choose a patch on which
\(m_{i,u}=\ell_{i,u}-\ell'_{i,u}\) is nonzero.  Each vector satisfies
\(\sum_u\ell_{i,u}=0\), so \(\sum_um_{i,u}=0\).
Therefore at least two coordinates of \(m\) are nonzero.  These
coordinates refer to independent energy variables; each nonzero integer
has modulus at least one.  For \(|t|>L_0\), at least two factors in
the product consequently have modulus at most \(L_0/|t|\), and all
remaining factors have modulus at most one.  For \(|t|\le L_0\) every
factor has modulus at most one.  The product is thus bounded by
\(r_{L_0}(t)\) in all cases.

The coefficients \(W_\ell W_{\ell'}\) are nonnegative.  Their sum over
\(\ell\ne\ell'\) is
\((\sum_\ell W_\ell)^2-\sum_\ell W_\ell^2=1-\beta(V)\).
Taking absolute values of this finite off-diagonal sum proves the first
inequality.

For \(\widetilde\varphi_S\), use
Lemma~\ref{lem:structural-frequencies} to restrict both indices to
nonzero vectors.  The diagonal sum is \(\beta-W_0^2\); the remaining
coefficients sum to
\[
 (1-W_0)^2-(\beta-W_0^2)\le1.
\]
The identical product estimate proves the second inequality.
The exceptional probability-zero set in that lemma does not change any
expectation.  Finally average over \(V\) and use
Eq.~\eqref{eq:plateau-design-bound}.
\end{proof}

\subsection{Ensemble relaxation and probability bounds}

\begin{theorem}[Autocorrelation and every-window relaxation at rate \(b_N+1/T\)]
\label{thm:generic-nearly-local-AS}
Assume the reference-design condition and the independent energy-density
hypotheses of Theorem~\ref{thm:scalar-dephasing}.  Then for every \(T>0\),
\begin{align}
 \mathbb E A_S(T)&\le b_N+\frac{2L_0}{T},\\
 \mathbb E\mathcal A_S(T;H)&\le b_N+\frac{2L_0}{T},\\
 \mathbb E\mathcal F_S(T;H)&\le b_N+\frac{4L_0}{T},\\
 \mathbb E\mathscr E_H(T)
      \le\mathfrak c_d\left(b_N+\frac{4L_0}{T}\right).
 \label{eq:structured-every-window-main}
\end{align}
Moreover,
\[
 \mathbb E\pi_{0,S}\le\varepsilon,\qquad
 \mathbb E\pi_{0,S}^2\le\varepsilon.
\]
The expectation in Eq.~\eqref{eq:structured-every-window-main} is outside
both the starting-time supremum and the input supremum.
\end{theorem}
\begin{proof}
For \(T\le L_0\), \(\int_0^T r_{L_0}(t)dt=T\le2L_0\).
For \(T>L_0\), direct integration gives
\[
 \int_0^T r_{L_0}(t)dt
 =L_0+L_0^2\left(\frac1{L_0}-\frac1T\right)
 =2L_0-\frac{L_0^2}{T}\le2L_0.
\]
All squared correlations are nonnegative and bounded by one, so
interchanging the finite time integral and expectation is justified.
Equation~\eqref{eq:scalar-dephasing-averaged} and the last integral
estimate prove the two rectangular bounds.  The triangular weight has
normalized integral one and is at most one, so the same argument with
its \(1/T\) normalization gives the third bound.
Apply Eq.~\eqref{eq:positive-kernel-full} for each individual
Hamiltonian \emph{before} taking expectation to obtain the last bound.

The zero-mode statements follow from
\(\pi_{0,S}=W_0\) almost surely,
Eq.~\eqref{eq:plateau-design-bound}, and \(0\le W_0^2\le W_0\le1\).
\end{proof}

\begin{corollary}[Explicit polynomial rates and high probability]
\label{cor:generic-nearly-local-bands}
Fix \(b,c>0\).  Choose the reference error \(\varepsilon=N^{-b}\),
logarithmic patches as in Lemma~\ref{lem:admissible-patches}, and a common
density constant \(L_0\) independent of \(N\).  At \(T_N=N^c\),
\begin{equation}
 \mathbb E A_S(T_N)
     =O_{b,L_0}\!\left(N^{-\min\{b,c\}}\right),\qquad
 \mathbb E\mathscr E_H(T_N)
     =O_{s,b,L_0}\!\left(N^{-\min\{b,c\}}\right).
 \label{eq:structured-polynomial-rates}
\end{equation}
For every \(0<a<\min\{b,c\}\),
\begin{equation}
 \mathbb P[\mathscr E_H(N^c)>N^{-a}]
 \le \mathfrak c_d
       \left(N^{a-b}+\frac{N^a}{q(2^{2N}-1)}
                   +4L_0N^{a-c}\right).
 \label{eq:structured-high-probability}
\end{equation}
For example, \(b=c=2\) gives expected every-window mean-square error
\(O_s(N^{-2})\), and error at most \(N^{-1}\) with failure probability
\(O_s(N^{-1})\), on windows \(N^2\).
\end{corollary}
\begin{proof}
Substitute \(T=N^c\), \(D=2^N\), and
\(b_N=N^{-b}+1/[q(2^{2N}-1)]\) into
Theorem~\ref{thm:generic-nearly-local-AS}.  The exponential term is
smaller than any fixed inverse power of \(N\).  Apply the elementary
Markov inequality to the nonnegative Borel random variable
\(\mathscr E_H(N^c)\), with threshold \(N^{-a}\).
Setting \(b=c=2\) and \(a=1\) gives the example.
\end{proof}

\begin{proposition}[The \(1/T\) dependence is sharp for the every-window criterion]
\label{prop:structured-short-time-lower}
In the structured ensemble, fix \(s\ge1\) qubits, put \(d=2^s\) and
\(a_d=1-1/d\).  Every realization satisfies, for \(T\ge1\),
\begin{equation}
 \mathscr E_H(T)\ge\frac{a_d^3}{40s\,T}.
 \label{eq:structured-short-time-lower}
\end{equation}
Consequently, for \(\varepsilon=N^{-b}\) and \(1\le T\le N^b\),
\(\mathbb E\mathscr E_H(T)=\Theta_{s,b,L_0}(T^{-1})\) for all
sufficiently large admissible \(N\).
This lower bound records the unavoidable initial transient in a
criterion that includes the window starting at zero; it is not a
pointwise long-time lower bound.
\end{proposition}
\begin{proof}
Choose a pure subsystem input \(\rho_S\) and
\(\rho_0=\rho_S\otimes\Id_B/m\).
A term whose support misses \(S\) commutes with \(\rho_0\).
There are at most \(5s\) terms meeting \(S\), each of norm at most one,
by Proposition~\ref{prop:structured-scope}.  The triangle inequality and
\(\|A X\|_1\le\|A\|_\infty\|X\|_1\) therefore give
\[
 \|[H,\rho_0]\|_1
 \le\sum_{i:\,\operatorname{supp}h_i\cap S\ne\varnothing}
                 2\|h_i\|_\infty\|\rho_0\|_1
 \le10s.
\]
Unitary evolution satisfies
\([H,\rho(t)]=e^{-itH}[H,\rho_0]e^{itH}\).
Integrating its derivative thus gives
\(\|\rho(t)-\rho_0\|_1\le10s\,t\) for \(t\ge0\).

Partial trace contracts trace norm on Hermitian differences.
For completeness, if \(X=X_+-X_-\) is its positive/negative
spectral decomposition, positivity of partial trace and the triangle
inequality give
\[
 \|\Tr_BX\|_1
 \le\Tr(\Tr_BX_+)+\Tr(\Tr_BX_-)
 =\Tr X_++\Tr X_-=\|X\|_1.
\]
Hence \(\Dt(\Phi_t(\rho_S),\rho_S)\le5s\,t\).
The eigenvalues of \(\rho_S-\Id_S/d\) are \(1-1/d\) and
\(d-1\) copies of \(-1/d\), so
\(\Dt(\rho_S,\Id_S/d)=a_d\).
The triangle inequality yields
\[
 \Dt(\Phi_t(\rho_S),\Id_S/d)\ge a_d-5s\,t.
\]
For \(0\le t\le t_*=a_d/(10s)\), this distance is at least \(a_d/2\).
Since \(t_*\le1\), the length-\(T\) window starting at zero contains this
whole interval whenever \(T\ge1\).  Its average squared error is
therefore at least \(t_*a_d^2/(4T)=a_d^3/(40sT)\).
This proves Eq.~\eqref{eq:structured-short-time-lower}.

For \(T\le N^b\), the plateau \(N^{-b}\) is at most \(1/T\).
The exponential Haar term is also at most \(1/T\) for sufficiently large
\(N\).  The upper bound in
Theorem~\ref{thm:generic-nearly-local-AS} is consequently
\(O_{s,b,L_0}(1/T)\); the lower bound just proved holds for every
realization and hence after expectation.
\end{proof}

\begin{corollary}[One event covers every longer window and every local input]
\label{cor:all-longer-windows}
For a fixed Hamiltonian, if \(\mathscr E_H(T)\le u\), then for every
\(L\ge T\), every \(t_0\ge0\), and every trace-distance tolerance
\(\eta>0\),
\begin{equation}
 \frac1L\meas\left\{t\in[t_0,t_0+L]:
       \sup_{\rho_S}\Dt(\Phi_t(\rho_S),\Id_S/d)\ge\eta\right\}
 \le\frac{2u}{\eta^2}.
 \label{eq:common-good-times}
\end{equation}
For \(L=T\) the factor two may be omitted.
\end{corollary}
\begin{proof}
Write \(L=kT+r\), where \(k\ge1\) is an integer and \(0\le r<T\).
Split the integral of \(e_H\) into \(k\) complete windows and a final
interval of length \(r\).  The latter integral is at most that over the
length-\(T\) window starting at the same point, because \(e_H\ge0\).
Thus the whole integral is at most \((k+1)Tu\le2Lu\).
For \(L=T\), the definition gives \(Tu\) directly.
On the set in Eq.~\eqref{eq:common-good-times}, \(e_H(t)\ge\eta^2\).
Integrating this lower bound over that set proves the claim.
No union bound over a continuum of times or states is used.
\end{proof}

\begin{remark}[Several subsystems]
For a specified collection \(\mathcal S_N\) of at most \(N^r\)
subsystems of a fixed maximum size, a union bound applied to
Eq.~\eqref{eq:structured-high-probability} gives the corresponding
simultaneous statement with failure probability multiplied by \(N^r\).
For threshold \(N^{-a}\), the failure probability therefore tends
to zero whenever \(\min\{b,c\}>a+r\).  The maximum subsystem size
remains fixed.
\end{remark}

\subsection{Norm-small perturbations}

\begin{theorem}[Perturbation stability without resolving microscopic gaps]
\label{thm:positive-kernel-stability}
Let \(G,H\) be arbitrary finite Hamiltonians on the same Hilbert space
and suppose \(\|H-G\|_\infty\le\delta\).  For every \(T>0\),
\begin{equation}
 \mathscr E_H(T)
 \le\mathfrak c_d\left[
       2\mathcal F_S(T;G)+\frac43\delta^2T^2\right].
 \label{eq:positive-kernel-stability}
\end{equation}
In particular, if \(G\) is drawn from the structured ensemble and
\(H=G+W_N\), with \(W_N\) an arbitrary measurable perturbation,
possibly dependent on \(G\), and
\(\|W_N\|_\infty\le\delta_N\) for every realization, then
\begin{equation}
 \mathbb E\mathscr E_H(T)
 \le\mathfrak c_d\left[
       2b_N+\frac{8L_0}{T}+\frac43\delta_N^2T^2\right].
 \label{eq:structured-stability}
\end{equation}
The perturbation need not preserve the commuting terms or their
conserved projectors.
\end{theorem}
\begin{proof}
For \(t\ge0\), differentiating
\(e^{-i(t-s)H}e^{-isG}\) and integrating from \(s=0\) to \(s=t\)
gives
\[
 \|e^{-itH}-e^{-itG}\|_\infty\le t\|H-G\|_\infty.
\]
For negative \(t\) apply the same argument to the adjoints.
For an operator \(O\), inserting and subtracting
\(e^{-itG}Oe^{itH}\), and using unitary invariance of the
Hilbert--Schmidt norm, gives
\[
 \|e^{-itH}Oe^{itH}-e^{-itG}Oe^{itG}\|_{2,D}
 \le2\delta |t|\,\|O\|_{2,D}.
\]
Every basis element \(O_a\) has normalized Hilbert--Schmidt norm one.
Cauchy--Schwarz in the defining trace and averaging over \(a\) therefore
give
\[
 |\varphi_S(t;H)-\varphi_S(t;G)|\le2\delta|t|.
\]
It follows that
\(|\varphi_S(t;H)|^2\le2|\varphi_S(t;G)|^2+8\delta^2t^2\).
The normalized second moment of the triangular weight is
\[
 \frac1T\int_{-T}^{T}
  \left(1-\frac{|t|}{T}\right)t^2dt
 =\frac2T\left(\frac{T^3}{3}-\frac{T^3}{4}\right)=\frac{T^2}{6}.
\]
Consequently
\(\mathcal F_S(T;H)\le
2\mathcal F_S(T;G)+(4/3)\delta^2T^2\).
Apply Theorem~\ref{thm:positive-kernel-return}, then take expectations
and use Theorem~\ref{thm:generic-nearly-local-AS}.
\end{proof}

\begin{remark}[Scope of the stability result]
At \(T=N^c\) and \(\delta_N\le N^{-v}\), the perturbation term is
\(O(N^{-2(v-c)})\); it is small when \(v>c\).
A sum of at most \(C N\) bounded local perturbations, each of norm at
most \(\lambda_N\), meets this hypothesis with
\(\delta_N=C N\lambda_N\).
Thus sufficiently small noncommuting perturbations are allowed, but
their required strength vanishes with \(N\).
This is not a theorem for generic order-one noncommuting couplings.
The uniformity in the starting time comes from the positive-kernel
certificate, not from a Duhamel estimate at the possibly enormous
starting time \(t_0\).
\end{remark}
\end{document}